\documentclass[11pt]{article} 

\usepackage{amsmath,amsthm,amssymb}
\usepackage[top=2cm, bottom=2.5cm, left=3cm, right=3cm]{geometry}
\usepackage{url}
\usepackage[colorlinks=true,bookmarks=false,citecolor=blue,urlcolor=blue]{hyperref} 
\usepackage{comment}
\usepackage{caption}
\usepackage{float}
\usepackage{subcaption}
\usepackage{graphicx}

\usepackage[sorting=none]{biblatex}

\usepackage{mathtools}
\usepackage{authblk}
\usepackage{bbm}
\DeclareMathOperator*{\esssup}{ess\,sup}

\title{\textbf{Uncovering Market Disorder and Liquidity Trends Detection}}

\author{
Etienne CHEVALIER\footnote{Laboratoire de Math\'ematiques et Mod\'elistation d'Evry, Universit\'e Paris-Saclay, UEVE,
 UMR 8071 CNRS, France; email: etienne.chevalier@univ-evry.fr} $\,$
 Yadh HAFSI \footnote{Laboratoire de Math\'ematiques et Mod\'elistation d'Evry, Universit\'e Paris-Saclay, UEVE,
 UMR 8071 CNRS,  France; email: yadh.hafsi@universite-paris-saclay.fr} 
 $\,$ 
 Vathana LY VATH \footnote{Laboratoire de Math\'ematiques et Mod\'elistation d'Evry, Universit\'e Paris-Saclay, ENSIIE, UEVE,
 UMR 8071 CNRS, France; email: vathana.lyvath@ensiie.fr}
}

\newcommand{\nunder}[2][5]{\mathrlap{\mkern\the\numexpr#1/2mu\relax\underline{\phantom{\mathrm{#2}\mkern-#1mu}}}#2}

\bibliography{main.bib}
\providecommand{\keywords}[1]
{
  \small	
  \textbf{Keywords:} #1
}
\usepackage{hyperref}
\usepackage{multirow}
\usepackage{fancyhdr}
\usepackage{ulem}
\usepackage{multicol}
\begin{document}
\maketitle
\pagenumbering{arabic}
\newcounter{axiom}
\newtheorem{axiome}[axiom]{Axiome}
\newtheorem{defi}{Definition}[section]
\newtheorem{theo}{Theorem}[section]
\newtheorem{prop}{Proposition}[section]
\newtheorem{rque}{Remark}[section]
\newtheorem{nota}{Notation}[section]
\newtheorem{demo}{Demonstration}
\newtheorem{propt}{Propriété}
\newtheorem{lemma}{Lemma}
\newtheorem{assump}{Assumptions}[section]
\newcounter{cases}
\newcounter{subcases}[cases]
\newenvironment{mycase}
{
    \setcounter{cases}{0}
    \setcounter{subcases}{0}
    \newcommand{\case}
    {
        \par\indent\stepcounter{cases}\textbf{Case \thecases.}
    }
    \newcommand{\subcase}
    {
        \par\indent\stepcounter{subcases}\textit{Subcase (\thesubcases):}
    }
}
{
    \par
}
\renewcommand*\thecases{\arabic{cases}}
\renewcommand*\thesubcases{\roman{subcases}}


\begin{abstract}
The primary objective of this paper is to conceive and develop a new methodology to detect notable changes in liquidity within an order-driven market. We study a market liquidity model which allows us to dynamically quantify the level of liquidity of a traded asset using its limit order book data. The proposed metric holds potential for enhancing the aggressiveness of optimal execution algorithms, minimizing market impact and transaction costs, and serving as a reliable indicator of market liquidity for market makers. As part of our approach, we employ Marked Hawkes processes to model trades-through which constitute our liquidity proxy. Subsequently, our focus lies in accurately identifying the moment when a significant increase or decrease in its intensity takes place. We consider the minimax quickest detection problem of unobservable changes in the intensity of a doubly-stochastic Poisson process. The goal is to develop a stopping rule that minimizes the robust Lorden criterion, measured in terms of the number of events until detection, for both worst-case delay and false alarm constraint. We prove our procedure's optimality in the case of a Cox process with simultaneous jumps, while considering a finite time horizon. Finally, this novel approach is empirically validated by means of real market data analyses.
\end{abstract}\hspace{10pt}
\\

\keywords{Liquidity Risk, Quickest Detection, Change-point Detection, Minimax Optimality, Marked Hawkes Processes, Limit Order Book.}

\section{Introduction}
Assets liquidity is an important factor in ensuring the efficient functioning of a market. Glosten and Harris \cite{GLOSTEN1988123} define liquidity as the ability of an asset to be traded rapidly, in significant volumes and with minimal price impact. Measuring liquidity, therefore, involves three aspects of the trading process: time, volume and price. This is reflected in Kyle's description of liquidity as a measure of the tightness of the bid-ask spread, the depth of the limit order book and its resilience \cite{kyle1985}. Hence, as highlighted in Lybek and Sarr \cite{Lybek} and in Binkowski and Lehalle \cite{lehalle_liquidity}, it is essential for liquidity-driven variables to not only capture the transaction cost associated with a relatively small quantity of shares but also to assess the depth accessible to large market participants. Additionally, measures related to the pace or speed of the market are considered since trading slowly can help mitigate implicit transaction costs. Typical indicators in this category include value traded and volatility over a reference time interval.

Several studies have been devoted to identifying different market regimes rapidly and in an automated fashion. Hamilton \cite{HAMILTON1988385} proposes the initial notion of regime switches, wherein he establishes a relationship between cycles of economic activity and business cycle regimes. The idea remains relevant, with ongoing exploration of novel approaches. Hovath et al. \cite{blanka2023} suggest an unsupervised learning algorithm for partitioning financial time series into different market regimes based on the Wassertein k-means algorithm for example. Bucci et al. \cite{BUCCI2022105832} employ covariance matrices to identify market regimes and identify shifts in volatile states. 

This paper introduces a novel liquidity regime change detection methodology aimed at assessing the resilience of an order book using tick-by-tick market data. 
More precisely, our objective is to study methods derived from the theory of disorder detection and to develop a liquidity proxy that effectively captures its inherent characteristics. This will enable us to thoroughly examine how the distribution of the liquidity proxy evolves over time. By employing these methods, we seek to gain a detailed understanding of the dynamics involved in liquidity changes and their impact within the distribution patterns of our chosen proxy.
To accomplish this, our methodology centers around the introduction of a liquidity proxy known as "trades-through" (see Definition \ref{def_tt}). Trades-through possess a high informational content, which is of interest in this study. Specifically, trades-through serve as an indicator of the order book's resilience, allowing the examination of activity at various levels of depth. They can also provide insightful information on the volatility of the financial instrument under examination. Additionally, trades-through can indicate whether an order book has been depleted or not in comparison to its previous states and to which extent. Pomponio and Abergel \cite{Pomponio2010} investigate several stylized facts related to trades-through, which are similar to other microstructural features observed in the literature. Their research highlights the statistical robustness of trades-through concerning the size of orders placed. This implies that trades-though are not solely a result of low quantities available on the best limits and that the information they provide is of significant importance. Another noteworthy result they observe is the presence of self-excitement and clustering patterns in trades-through. In other words, a trade-through is more likely to occur soon after another trade-through than after any other trade. This result supports the idea of modelling trades-through using Hawkes processes as in Ioane Muni Toke and Fabrizio Pomponio \cite{toke2012modelling}. This approach is particularly interesting given that Hawkes processes fall into the class of branched processes. Their dynamics can therefore be illustrated by a representation of immigration and birth that allows interpretation. Indeed, the intensity, at which events take place, is formed by an exogenous term representing the arrival of "immigrants" and a self-reflexive endogenous term representing the fertility of "immigrants" and their "descendants". We may refer to Brémaud et al. \cite{massoulie} for more comprehensive details.

Our focus here will be to use the Marked-Hawkes process to model trades-through rather than standard Hawkes processes in order to capture the impact of the orders' volumes on the intensity of the counting processes. Our research centers on employing the Marked-Hawkes process as a modelling technique for trades-through, as opposed to conventional Hawkes processes. The rationale behind this choice lies in our intention to account for the influence of order volumes on the intensity of the counting processes.

Upon the conclusion of the modelling phase, we will utilize this proxy to identify intraday liquidity regimes\footnote{Moments of transition between liquidity regimes will be defined as instances where the distribution of trade-throughs experiences a change.}. In order to achieve this, our model involves detecting the times at which the distribution of liquidity (represented by a counting process) undergoes changes as fast as possible. This entails comparing the intensities of two doubly stochastic Poisson processes. This type of problem is commonly known in literature as "Quickest change-point detection" or "Disorder detection". Disorder detection has been studied through two research paths, namely the Bayesian approach and the minimax approach. The Bayesian perspective, originally proposed by Girschick and Rubin \cite{Girshick}, typically assumes that the change point is a random variable and provides prior knowledge about its distribution. As an example, Dayanik et al. \cite{Dayanik2006CompoundPD} conduct their study on compound Poisson processes assuming an exponential prior distribution for the disorder time. In their work, Bayraktar et al. \cite{Bayraktar2006} address a Poisson disorder problem with an exponential penalty for delays, assuming an exponential distribution for the delay. However, the minimax approach does not include any prior on the distribution of the disorder times, see for instance, Page \cite{page54} and Lorden \cite{lorden71}. Due to the scarcity of relevant literature and data pertaining to disorder detection problems, accurately estimating the prior distribution in the context of market microstructure poses a challenging task. Hence, our preference lies in adopting a non-Bayesian min-max approach. El Karoui et al. \cite{ElKaroui2017} prove the optimality of the CUSUM\footnote{The abbreviation CUSUM stands for CUmulative SUM.} procedure \ref{def_cusum} for $\rho < 1$ and $\rho > 1$ for doubly stochastic poisson processes which extends the proof proposed by Moustakides \cite{Moustakides2008} and Poor and Hadjiliadis \cite{VPoor2009} for $\rho < 1$. However, the proposed solution is limited to cases where point processes do not exhibit simultaneous jumps. This presents a problem in our case since definition \ref{def_tt} suggests that each time a $n$-limit trade-through event is observed, $n-1$ other trade-through events occur simultaneously. Hence, we need to demonstrate the optimality of these results by considering scenarios in which point processes exhibit a finite number of simultaneous arrival times which is precisely what we obtained. We are also able to produce a tractable formula of the average run delay. More importantly, our findings encompass the intricacies inherent to modelling limit order books and investigating market microstructure. We use our results to identify changes in liquidity regimes within an order book. To do this, we apply our detection procedure to real market data, enabling us to achieve convincing detection performance. To the best of our knowledge, we believe that this is the first attempt to quantify market liquidity using a methodology combining the notion of trades-through, disorder detection theory, and marked-hawkes processes.

The article is structured in the following manner. In Section \ref{Modelling Trades-Through Using Hawkes Processes}, the primary aim is to establish a model for trades-through using Marked Hawkes Processes (MHP) and subsequently provide the relevant mathematical framework. In Section \ref{Sequential Change-Point Detection: CUSUM-based optimal stopping scheme}, various results related to the optimality of the CUSUM procedure within the framework of sequential test analysis for a simultaneous jump Cox process will be presented. Finally, in Section \ref{experimental_results}, we provide an overview of the goodness-of-fit results obtained from applying a Marked Hawkes Process to analyze trades-through data in the Limit Order Book of BNP Paribas stock. Additionally, we present the outcomes of utilizing our disorder detection methodology to identify various liquidity regimes.


\section{Trades-Through modelling and Hawkes Processes}
\label{Modelling Trades-Through Using Hawkes Processes}
The aim of this section is to build a model for trades-through by means of Marked Hawkes Processes (MHP) and to present the corresponding mathematical framework. We begin by defining the concept of trades-through.
\begin{defi}[Trade-through]
 Formally, a trade-through can be defined by the vector $(type, depth,\\ volume) \in \{-1,1\}\times\mathbb{N}\times\mathbb{R}_+$ where type equals $1$ if the trade-through occurs on the bid side of the book and $-1$ otherwise, depth represents the number of limits consumed by the market event relative to the trade-through, and volume represents the corresponding traded volume. 
 \label{def_tt}
\end{defi}
\noindent
In brief, an $n$-limit trade-through is defined as an event that exhausts the available $n$-th limit in the order book. It is considered that a trade-through at limit $m$ is also a trade-through of limit $n$ if $m$ is greater than $n$ \cite{Pomponio2010}. In this context, the term 'limit $n$' refers to any order that is positioned $n$ ticks away from either the best bid price or the best ask price. This definition is broad as it encompasses different types of orders, including market orders, cancellations, and hidden orders which underlines its significant information content.\\
\begin{figure}[H]
\centering
\captionsetup[subfigure]{labelformat=empty}

\begin{minipage}{1\linewidth}
\centering
\subfloat
{\includegraphics[width=5cm]{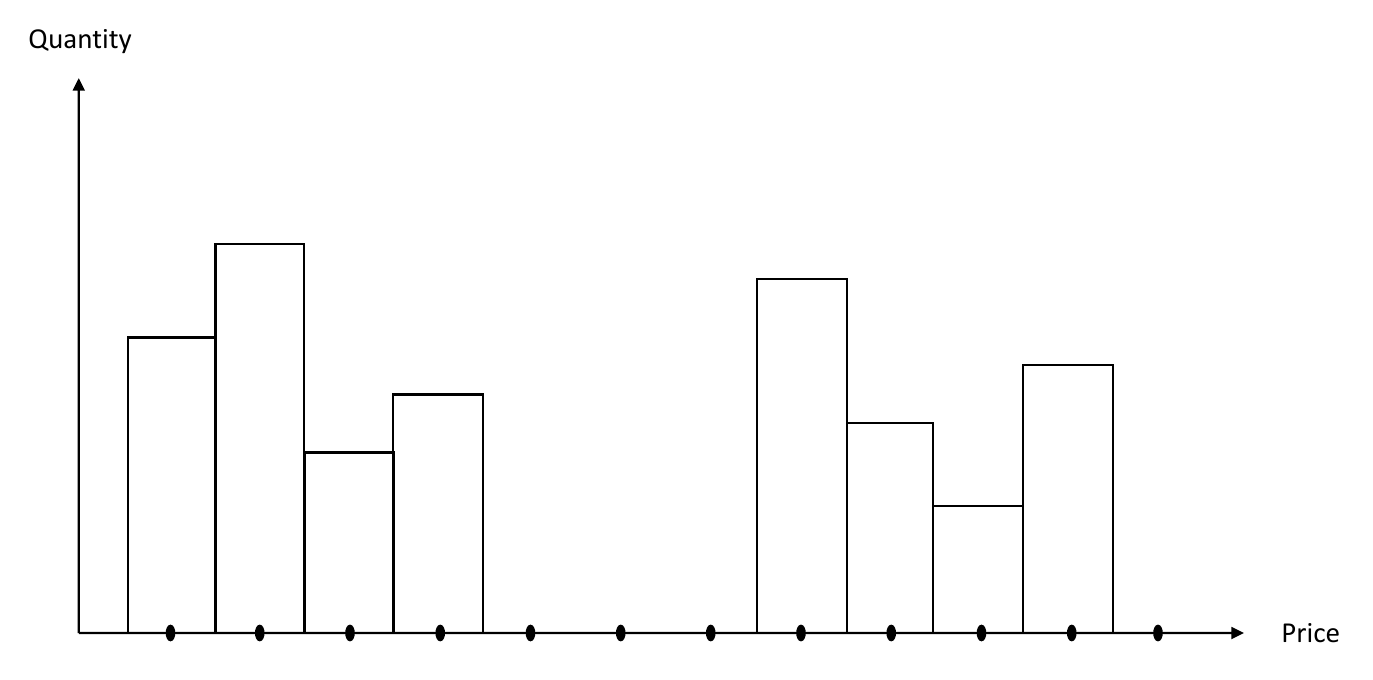}}
\subfloat{{ }\includegraphics[width=5cm]{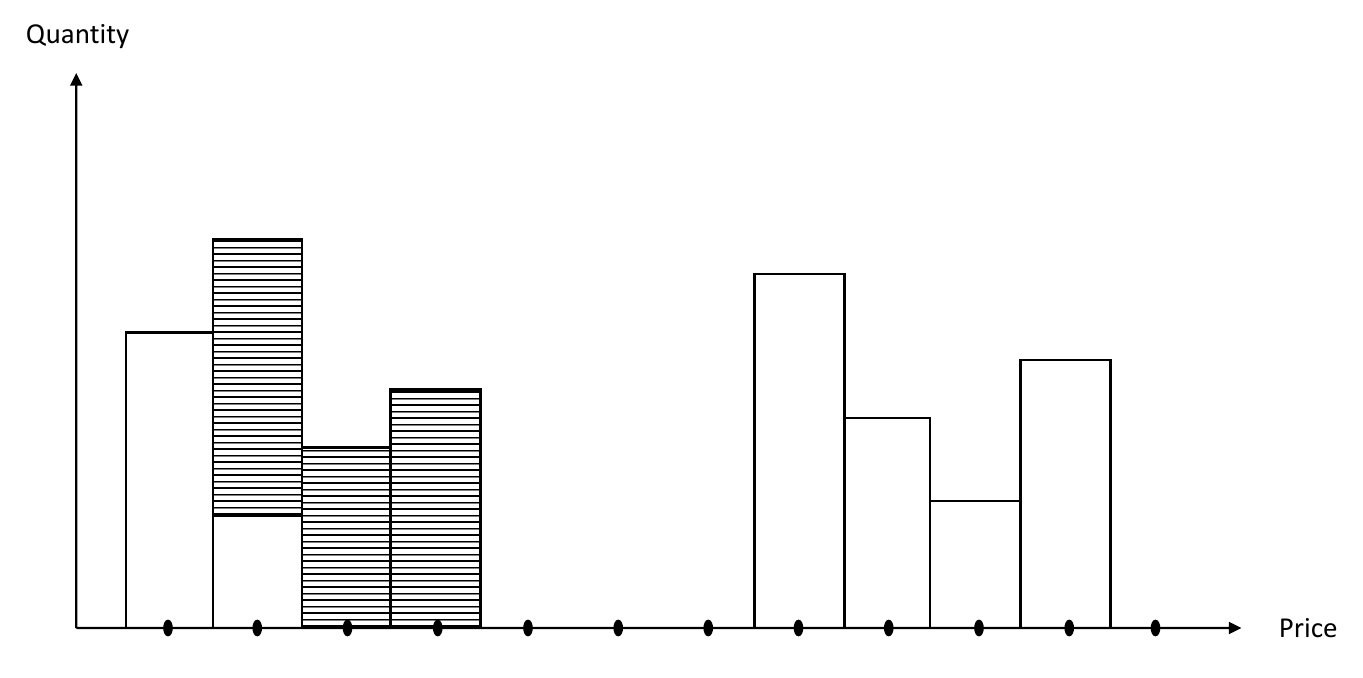}}
\subfloat{{ }\includegraphics[width=5cm]{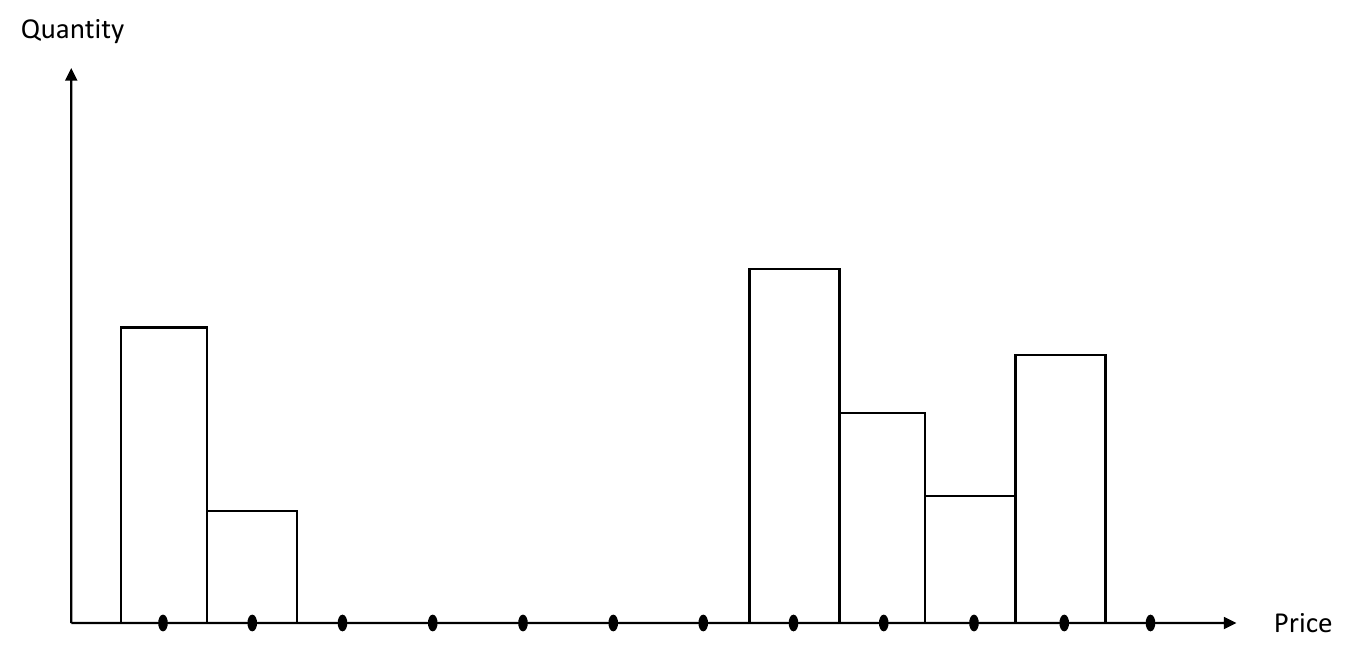}}
\label{fig:illustration_tt}
\end{minipage}

\caption{Illustration of a trade-through event of limit 2 on the bid side.}
\label{fig:main}
\end{figure}

 Let $(\Omega ,\mathbb{F}, \mathcal{F} = \{\mathcal{F}_t\}_{t \geq 0},\mathbb{P})$ be a complete filtered probability space endowed with right continuous filtration $\{\mathcal{F}_t\}_{t \geq 0}$. Going forward, whenever we mention equalities between random variables, we mean that they hold true almost surely under the probability measure $\mathbb{P}$. We may not explicitly indicate this notion ($\mathbb{P}$-a.s) in most cases. 
 Let $0 < \tau^{A,i}_0<\tau^{A,i}_1<\tau^{A,i}_2<\dots$ (resp. $0 < \tau^{B,i}_0<\tau^{B,i}_1<\tau^{B,i}_2<\dots$) form a sequence of increasing $\mathcal{F}$-measurable stopping times and $(v^{A,i}_k)_{k \geq 0}$ (resp. $(v^{B,i}_k)_{k \geq 0}$) be a sequence of i.i.d $\mathbb{R}_+$-valued random variables for each $i\in\{1,\dots,M\}$. Here, $(\tau^{A,i}_k)_{k \geq 0}$ (resp. $(\tau^{B,i}_k)_{k \geq 0}$) represent the arrival times of trade-throughs at limit $i$ on the ask (resp. bid) side of the limit order book while $(v^{A,i}_k)_{k \geq 0}$ (resp. $(v^{B,i}_k)_{k \geq 0}$) represent the sizes of volume jumps of trades-though for these arrival times.
 \begin{rque}
 The $(\tau^{A,i}_k)_{k\geq 0}$ (resp. $(\tau^{B,i}_k)_{k\geq 0}$) values are arranged in a specific order for a given $i$, but it is essential to note that this order may not be consistently maintained when comparing them to the $(\tau^{A,j}_k)_{k\geq 0}$ (resp. $(\tau^{B,j}_k)_{k\geq 0}$) values for a different $j$. We proceed by gathering the arrival times of trade-through orders for all the limits on the bid (resp. ask) side. Subsequently, these arrival times are organized into sequences $(\tau^A_k)_{k\geq 0}$ and $(\tau^B_k)_{k\geq 0}$ respectively that are sorted in ascending order.
 \end{rque}
 Our goal is to design a model that takes into account the volumes of the trades-through for all pertinent limits and that replicates the clustering phenomenon observed in the corresponding stylized facts (See \cite{Pomponio2010}). To achieve this, we suggest using a modified version of Cox processes\footnote{Also referred to as doubly stochastic Poisson process.}, specifically Marked Hawkes Processes. We define a 2-dimensional marked point process $N^{A\times B} = \left(N^A,N^B\right)$ where $N^A$ (resp. $N^B$) is the counting measure on the measurable space $\left(\mathbb{R}_+\times\mathbb{R}_+,\mathcal{B}\otimes\mathcal{L}\right)$ associated to the set of points $\left\{(\tau_k^A,v^A_k); k \geq 0\right\}$ (resp. $\left\{(\tau_k^B,v^B_k); k \geq 0\right\}$), $\mathcal{B}$ being the Borel $\sigma$-field on $\mathbb{R}_+$ while $\mathcal{L}$ is the Borel $\sigma$-field on $\mathbb{R}_+$. This means that $\forall C\in\mathcal{B}\otimes\mathcal{L}$ ~:
\begin{align*}
  N^A(C) &= \sum_{k\geq 0}\mathbf{1}_C(\tau_k^{A},v^{A}_k)\quad \quad\quad\quad\quad, &      N^B(C) &= \sum_{k\geq 0}\mathbf{1}_C(\tau_k^{B},v^{B}_k)  \\
        &= \sum_{1\leq i\leq M}\sum_{k\geq 0}\mathbf{1}_C(\tau_k^{A,i},v^{A,i}_k)                      &  &= \sum_{1\leq i\leq M}\sum_{k\geq 0}\mathbf{1}_C(\tau_k^{B,i},v^{B,i}_k)
\end{align*}
Let $\mathcal{F}^{A\times B} = \left(\mathcal{F}_t^{A\times B}\right)_{t\geq 0}$  be a sub-$\sigma$-field of the complete $\sigma$-field $\mathcal{F}$ where \\
$\mathcal{F}_t^{A\times B}=\sigma\left(N^{A\times B}(s,v);\\ s \in[0, t],v\in\mathbb{R}_+\right)$. Simply put, $\mathcal{F}^{A\times B}_t$ contains all the information about the process  $N^{A\times B}$ up to time $t$. The collection of all such sub-$\sigma$-fields, denoted by $\left(\mathcal{F}_t^{A\times B}, t \geq 0\right)$, is called the internal history of the process $N^{A\times B}$. One of the core concepts underlying the dynamics of $N^{A\times B}$ is the notion of conditional intensity. Specifically, the dynamics of $N^{A\times B}$ are characterized by the non-negative, $\mathcal{F}^{A\times B}$-progressively measurable processes $\lambda^A$ and $\lambda^B$, which model the conditional intensities. These processes satisfy the following conditions:
\begin{equation}
    \begin{aligned}       
 &\mathbb{E}\left[\int_K N^A(t,v)-N^A(s,v)\mathrm{d}v \big\lvert \mathcal{F}^{A\times B}_s\right]=\mathbb{E}\left[\int_{[s,t[\times K} \lambda^A(u,v) \mathrm{d} u\mathrm{d}v \big\lvert \mathcal{F}^{A\times B}_s\right],\quad 0 \leq s \leq t \\
&\mathbb{E}\left[\int_K N^B(t,v)-N^B(s,v)\mathrm{d}v \big\lvert \mathcal{F}^{A\times B}_s\right]=\mathbb{E}\left[\int_{[s,t[\times K} \lambda^B(u,v) \mathrm{d} u\mathrm{d}v \big\lvert \mathcal{F}^{A\times B}_s\right],\quad 0 \leq s \leq t
\end{aligned}
\end{equation}
Moving forward, the point process $N^{A\times B}$ will be represented as a bivariate Hawkes process, with its intensity wherein its intensity is dependent on the marks $(v^{A}_k)_{k \geq 0}$ and $(v^{B}_k)_{k \geq 0}$. 
By viewing the marked point processes $N^A$ and $N^B$ as a processes on the product space $\mathbb{R}_+\times \mathbb{R}_+$, we can derive the marginal processes $N^A_g(.) = N^A(.,\mathbb{R}_+)$ and $N^B_g(.) = N^B(.,\mathbb{R}_+)$, which refer to as ground process. In other worlds, for $ t\in\mathbb{R}_+$,
\begin{equation}
N^A_g(t)=\int_{[0, t[ \times \mathbb{R}_+} N^A(\mathrm{~d} u \times \mathrm{d} v^A),\quad N^B_g(t)=\int_{[0, t[ \times \mathbb{R}_+} N^B(\mathrm{~d} u \times \mathrm{d} v^B)
\end{equation}
These ground processes therefore describe solely the arrival times of the trades-through and can be understood as being Hawkes processes as well. 
In instances where $N^{A\times B}$ exhibit sufficient regularity (See The definition 6.4.III and section 7.3 in Daley and Vere-Jones \cite{daley2007introduction}), it is possible to factorize its conditional intensities $\lambda^A$ and $\lambda^B$ without presupposing stationarity and under the assumption that the mark's conditional distributions at time $t$ given $\mathcal{F}^{A\times B}_{t^-}$ have densities of $\mathbb{R}_+\rightarrow \mathbb{R}_+:~v\mapsto f_A(v\big\lvert \mathcal{F}^{A\times B}_{t^-})$ and $\mathbb{R}_+\rightarrow \mathbb{R}_+:~v\mapsto f_B(v\big\lvert \mathcal{F}^{A\times B}_{t^-})$. The conditional intensity factors out as~:
\begin{equation}
    \lambda^A (t,v)= \lambda^A_g(t)f_A(v\big\lvert \mathcal{F}^{A\times B}_{t^-}),\quad \lambda^B (t,v)= \lambda^B_g(t)f_B(v\big\lvert \mathcal{F}^{A\times B}_{t^-}),\quad (t,v)\in \mathbb{R}_+\times\mathbb{R}_+
    \label{intensity_factor}
\end{equation}
where $\lambda_g(t)$ is referred to as the $\mathcal{F}^{A\times B}_{t^-}$-intensity of the ground intensity.\\
In a heuristic manner, one can summarize equation \ref{intensity_factor} for $(t,v)\in \mathbb{R}_+\times\mathbb{R}_+$ as being in the form of~:
\begin{equation*}
\label{decomp_intensity}
\begin{aligned}
    &\lambda^A(t,v)\mathrm{d}t\mathrm{d}v\approx\mathbb{E}\left[N^A(\mathrm{d}t\times \mathrm{d}v)\big\lvert \mathcal{F}^{A\times B}_{t^-}\right]\approx\lambda^A_g(t)f_A(v\big\lvert\mathcal{F}^{A\times B}_{t^-})\mathrm{d}t\mathrm{d}v\\ &\lambda^B(t,v)\mathrm{d}t\mathrm{d}v\approx\mathbb{E}\left[N^B(\mathrm{d}t\times \mathrm{d}v)\big\lvert \mathcal{F}^{A\times B}_{t^-}\right]\approx\lambda^B_g(t)f_B(v\big\lvert \mathcal{F}^{A\times B}_{t^-})\mathrm{d}t\mathrm{d}v
    \end{aligned}
\end{equation*}
\begin{rque}
It is noteworthy that the conditional intensity of the ground processes considered in this context is evaluated with respect to the filtration $\mathcal{F}^{A\times B} = (\mathcal{F}_t^{A\times B})_{t\geq 0}$ rather than $\mathcal{F}^{N_g} = (\mathcal{F}_t^{N_g})_{t>\geq 0}$ with $\mathcal{F}_t^{N_g}=\sigma\left(N_g(s), s \in[0, t]\right)$. This is attributed to the fact that unlike $\mathcal{F}^{N_g}$, the filtration $\mathcal{F}^{A\times B}$ accounts for the information pertaining to the marks associated with the process $N^{A\times B}$.    
\end{rque}
We assume that the conditional intensities $\lambda_g^A$ and $\lambda_g^B$ have the following form~:
\begin{equation}
\begin{aligned}
\lambda_g^i(t) & =\mu^i(t)+\sum_{j=A,B}\sum_{k\geq 0} \gamma_{ij}(t-\tau_k^j,v_k^j)  \\
& =\mu^i(t)+\sum_{j=A,B}\int_{[0, t[ \times \mathbb{R}_+} \gamma_{ij}(t-s, v) N^j(\mathrm{~d} s \times \mathrm{d}v),\quad i=A,B
\end{aligned}
\end{equation}
where $\mu^i: \mathbb{R}_+ \rightarrow \mathbb{R}_+$ and $\gamma_{ij}: \mathbb{R}_+ \times \mathbb{R}_+ \rightarrow \mathbb{R}_+$ are non-negative measurable functions.\\
We choose a standard factorized form for the ground intensity kernel $\gamma_{ij}(t-u,v) = \alpha_{ij}e^{\beta_{ij}(t-u)}\times g_j(v)$ to describe the weights of the marks, $g_j$ being a measurable function defined on $\mathbb{R}_+$ that characterizes the mark's impact on the intensity. The exponential shape of the kernel is derived from the findings of Toke et al. \cite{toke2012modelling}. The properties of functions $g_A$ and $g_B$ will be examined in Section \ref{experimental_results}.
\begin{rque}
As we progress to proposition \ref{res_analysis_recursive}, it will become apparent that the factorization of the latter kernel presents a way to represent the conditional intensities of these processes in a Markovian manner.
\end{rque} 
Thus, the conditional ground intensities of this model have the following integral representation~:
\begin{equation}
\begin{aligned}
\lambda^{A}_g(t)=\mu^{A}(t)+&\int_{[0, t[ \times \mathbb{R}_+} \alpha_{A A} e^{-\beta_{A A}(t-u)}g_A\left(v\right) N^{A}(\mathrm{d} u \times \mathrm{d} v)\\&+\int_{[0, t[ \times \mathbb{R}_+} \alpha_{A B} e^{-\beta_{A B}(t-u)}g_B\left(v\right) N^{B}(\mathrm{d} u \times \mathrm{d} v),\quad t\geq 0 \end{aligned}
\end{equation}
\begin{equation}
\begin{aligned}
\lambda^B_g(t)=\mu^B(t)+&\int_{[0, t[ \times \mathbb{R}_+} \alpha_{B A} e^{-\beta_{B A}(t-u)}g_A\left(v\right) N^A(\mathrm{d} u \times \mathrm{d} v)\\&+\int_{[0, t[ \times \mathbb{R}_+} \alpha_{B B} e^{-\beta_{B B}(t-u)}g_B\left(v\right) N^B(\mathrm{d} u \times \mathrm{d} v),\quad t\geq 0 
\end{aligned}
\end{equation}
where $\left(\alpha_{i j}, \beta_{i j}\right)_{(i, j) \in\{A, B\}^2}$ are non-negative reals.\\
Conventionally, the impact functions $g_i$ have to satisfy the following normalizing condition which enhances the overall stability of the numerical results (see section 1.3 of \cite{Liniger2009} for further details)~:
\begin{equation}
\int_0^{+\infty}g_i(v)f_i(v)\mathrm{d}v = 1, \quad i = A,B.
\label{normalizing_cond}
\end{equation}
This means that the impact functions of the Marked-Hawkes based model would be equal to $g_i(v)=\frac{v^{\eta_i}}{\mathbb{E}\left(v^{\eta_i}\right)}$. The following stability conditions of our model are based on Theroem 8 of Brémaud et al. \cite{massioule_bremaud}.
\begin{prop}[Existence and uniqueness]
\label{existence_uniqueness}
Suppose that the following conditions hold~:
\begin{enumerate}
    \item The spectral radius of the branching matrix (also referred to as branching ratio) satisfies~: $$\sup _{\lambda \in \rho\left(\|\Gamma\|\right)}|\lambda|<1$$ 
    \item The decay functions satisfy
$$
\int_0^{+\infty} t \gamma^{ij}(t) \mathrm{d} t<+\infty, \quad \forall i,j \in\{A,B\}
$$
\end{enumerate}
 where $\|\Gamma\|=\left\{\left\|\gamma^{i j}\right\|\right\}_{(i, j) \in\{A, B\}^2}$, $\gamma^{i j}(t) = \alpha_{i j} e^{-\beta_{i j}t}g_j\left(v(t)\right)$ and $\rho\left(\|\Gamma\|\right)$ represents the set of all eigenvalues of $\|\Gamma(t)\|$,
then there exists a unique point process $N^{A\times B}$ with associated intensity process $\lambda^{A\times B}$.
\end{prop}
\begin{rque}
Existence in proposition \ref{existence_uniqueness} means that we can find a probability space $(\Omega, \mathbb{F}, \mathbb{P})$ which is rich enough to support such the process $N^{A\times B}$. Uniqueness means that any two processes complying with  the above conditions have the same distribution.
In our case, the conditions mentioned above are described as follows~:
    \begin{enumerate}
    \item $\frac{1}{2}\left(\frac{\alpha^{AA}}{\beta^{AA}}+\frac{\alpha^{BB}}{\beta^{BB}}+\sqrt{\left(\frac{\alpha^{AA}}{\beta^{AA}}-\frac{\alpha^{BB}}{\beta^{BB}}\right)^2+4 \frac{\alpha^{AB}}{\beta^{AB}} \frac{\alpha^{BA}}{\beta^{BA}}}\right)<1$
    \item $\frac{\alpha^{ij}}{{\beta^{ij}}^2}<+\infty,\quad \forall i,j \in\{A,B\}$
\end{enumerate}
\label{non_exposition_cond}
\end{rque}

\section{Sequential Change-Point Detection: CUSUM-based optimal stopping scheme}
\label{Sequential Change-Point Detection: CUSUM-based optimal stopping scheme}
We now shift our focus towards developing a methodology that can distinguish between distinct liquidity regimes. This can be done by detecting changes or disruptions in the distribution of a given liquidity proxy, which, in our case, is represented by the number of trades-through. Here, we will focus on identifying disruptions that affect the intensity of the marked multivariate Hawkes process discussed in the previous section. This will enable us to compare the resilience of our order book to that of previous days over a finite horizon $T$. The sequential detection methodology involves comparing the distribution of the observations to a predefined target distribution. The objective is to detect changes in the state of the observations as rapidly as possible while minimizing the number of false alarms which correspond to sudden changes in the arrival rate in the case of Poisson processes with stochastic intensity. To be specific, let us consider a general point process $N$ with a known rate $\lambda$, where the arrivals of certain events are associated to this process. At a certain point in time $\theta$, the rate of occurrence of events of the process $N$ undergoes an abrupt shift from $\lambda$ to $\rho\lambda$. However, the value of the disorder time $\theta$ is unobservable. The goal is to identify a stopping time $\tau$ that relies exclusively on past and current observations of the point process $N$ and can detect the occurrence of the disorder time $\theta$ swiftly. Thus, we will focus in the following on solving the sequential hypothesis testing problems in the aforementioned general form~:
\begin{equation}
\begin{array}{lll}
H_0:  &\lambda(t), \\
H_1: &\check{\lambda}^{\theta_1}(t)=\lambda(t) \mathbbm{1}_{\{t<\theta_1\}}+\rho_1 \lambda(t) \mathbbm{1}_{\{t \geq \theta_1\}}, \quad \rho_1 > 1 \\
\end{array}
\end{equation}
and,
\begin{equation}
\begin{array}{lll}
H_0: & \lambda(t), \\
H_2: & \check{\lambda}^{\theta_2}(t)=\lambda(t) \mathbbm{1}_{\{t<\theta_2\}}+\rho_2 \lambda(t) \mathbf{1}_{\{t \geq \theta_2\}}, \quad \rho_2 < 1 \\
\end{array}
\end{equation} 

We propose the introduce a change-point detection procedure that is applicable to our model. The goal is to study the distribution of the process $N^{A\times B}$ process. The problem at hand can be tackled through various approaches. Here, we opt for a non-Bayesian approach, specifically the min-max approach introduced by Page \cite{page54}, which assumes that the timing of the change is unknown and non-random. Previous research has shown the optimality of the CUSUM procedure in solving this problem (see \cite{qdetectionpoor}\cite{moustakides2004optimality}\cite{ElKaroui2017}). However, these demonstrations are confined to scenarios where the process involves Cox jumps of size 1. This constraint poses significant limitations in our case, as the definition \ref{def_tt} of trade-throughs implies for of simultaneous jumps. Therefore, the objective of this chapter is to demonstrate the optimality of the CUSUM procedure in a broader context that considers multiple jumps and to formulate the relevant detection delay.

The identification of different liquidity regimes in an order book can be done in various ways. One can apply our disorder detection methodology separately on processes $N^A$ and $N^B$ to study liquidity either on the bid or on the ask, or consider process $N^A + N^B$ and take both into account at the same time. To enhance the clarity and maintain the general applicability of our results, we consider the point process $N$ as a general Cox process that decomposes into the form $\sum_{i=1}^D N^i$ where $0<D<+\infty$ and $N^i$ is a finite point process defined on the same probability space introduced in the previous section and whose arrivals are non-overlapping. We denote $\lambda^i$ the $\mathcal{F}$-intensity of $N^i$ and by $t\mapsto \Lambda^i(t)=\int_0^t\lambda^i(s)\mathrm{d}s$ it's compensator for $i = 1,\dots,D$. 
    
We denote the probability measure for the case with no disorder by $\mathbb{P}_{\infty}$, where $\theta=+\infty$. The probability measure for the case with a disorder that starts from the beginning, where $\theta=0$, is denoted by $\mathbb{P}_{0}$. If there is a disorder at the instant $\theta$, represented by the value of the intensity of the counting process equal to $\check{\lambda}^{\theta}(t)$, then the probability measure is denoted by $\mathbb{P}_\theta$. The constructed measure satisfies the following conditions~:
$$\mathbb{P}_\theta= \begin{cases}\mathbb{P}_0, & \text { if } \theta=0 \\ \mathbb{P}_{\infty}, & \text { if } \theta=+\infty\end{cases}$$
\begin{rque}
    In what follows, we use the notation $\mathbb{E}$ to refer to the expectations that are evaluated under the probability measure $\mathbb{P}_{\infty}$, while $\mathbb{E}^{\theta}$ represents the expectations evaluated under the probability measure $\mathbb{P}_\theta$.
\end{rque}
We initiate the discourse by introducing a few notations that will be instrumental in presenting our results. Let $U^i=N^i-\beta(\rho) \Lambda^i$ and $i\in \{1,\dots, D\}$. The process $\sum_{i=1}^D U^i$ is referred to as the Log Sequential Probability Ratio (LSPR) between the two probability measure $\mathbb{P}_{\infty}$ and $\mathbb{P}_0$ where $\beta(\rho)=(\rho-1) / \log \rho$ and $\Lambda^i(t)=\int_0^t \lambda^i(s) \mathrm{d}s$ for all $0\leq t\leq T$. The probability measure $\mathbb{P}_0$ is defined as the measure equivalent to $\mathbb{P}_\infty$ with density~:
\begin{equation}
    \frac{\mathrm{d}\mathbb{P}_0}{\mathrm{d}\mathbb{P}_{\infty}}\big\lvert_{\mathcal{F}_t}= \rho^{\sum_{i=1}^{D} U^i(t)},~~~~0\leq t \leq T
    \end{equation}
Or, more generally~: \begin{equation}\begin{split}
    \frac{\mathrm{d}\mathbb{P}_\theta}{\mathrm{d}\mathbb{P}_{\infty}}\big\lvert_{\mathcal{F}_t} &= \prod_{i=1}^{D} \frac{\rho^{U^i(t)}}{\rho^{U^i(\theta)}}\quad~~,~~~~0\leq \theta \leq t \\&= \rho^{\sum_{i=1}^{D} U^i(t) - U^i(\theta)}\end{split}\end{equation}
According to Girsanov's theorem, these densities are defined if $\mathbb{E}\left(\rho^{\sum_{i=1}^{D} U^i(t) - U^i(\theta)}\right)=1$ which is the case since $\rho^{\sum_{i=1}^{D}U_i(t)} = e^{\log (\rho) \sum_{i=1}^{D}N^i(t)-(\rho-1) \sum_{i=1}^{D}\Lambda^i(t)}$ is a local martingale if $\sum_{i=1}^{D}\Lambda^i(t)$ is càdlàg (see Øksendal et al. \cite{oksendal2006applied}). Hence, $\sum_{i =1}^D N^i(t)-\rho\Lambda^{i}(t)$ is a $\mathbb{P}_{\theta}$-martingale iff $\rho^{\sum_{i=1}^{D} U^i(t) - U^i(\theta)}$ is an $\mathbb{P}_{\infty}$-martingale. We can therefore now define the log-likelihood ratio~:
\begin{equation}
    \begin{split}
        \ell_{t, \theta} :&=\sum_{i=1}^D \int_\theta^t \log \left(\frac{\check{\lambda}^i_{\theta}(s)}{\lambda^{i}(s)}\right) \mathrm{d} N^i(s)-\sum_{i=1}^D \int_\theta^t\left(\check{\lambda}^i_{\theta}(s)-\lambda^{i}(s)\right) \mathrm{d} s \\
& = \log \left(\rho\right) \left(\sum_{i=1}^D \left(N^i(t) - N^i(\theta)\right)-\frac{\rho - 1}{\log \left(\rho\right) }\left(\Lambda^{i}(t)-\Lambda^{i}(\theta)\right)\right) \\
& = \log \left(\rho\right)\left(\sum_{i=1}^D U^i(t) - U^i(\theta) \right)
    \end{split}
    \label{llr}
\end{equation}
with
$$
\check{\lambda}^i_{\theta}(t)= \begin{cases} \lambda^{i}(t), & 0 \leq t \leq \theta\\
\rho \lambda^{i}(t), & t>\theta 
\end{cases}
$$
which is the intensity of the node $i$ if the disorder takes place at the time $\theta$. 

In statistical problems, it is common to observe two distinct components of loss. The first component is typically associated with the costs of conducting the experiment, or any delay in reaching a final decision that may incur additional losses. The second component is related to the accuracy of the analysis. The min-max problem (see \cite{page54}) can be formulated as finding a stopping rule that minimizes the worst-case expected cost under the Lorden criterion. The objective is therefore to find the stopping time that minimizes the following cost~:
\begin{equation}
    \begin{split}
    \inf_{\tau\in \mathcal{T}_{\pi}} C(\tau)  
    \end{split}
\end{equation}
Where $C(\tau)=\sup _{\theta \in[0, +\infty]} {\text { ess sup }} \mathbb{E}^{\theta}\left[(\tau-\theta)^{+} \big\lvert \mathcal{F}_{\theta}\right]$, $\mathcal{T}_{\pi}$ is the class of $\mathcal{F}$-stopping times that statisfies $\mathbb{E}(\tau) \geq \pi$ and $\pi>0$ a constant.\\
As explained by Lorden \cite{lorden71}, the formulation of the problem under criterion $C(\tau)$ is robust in the sense that it seeks to estimate the worst possible delay before the disorder time. The false alarm rate can be controlled by controlling $\mathbb{E}(\tau)$ through $\pi$. It has been observed that Lorden's optimality criterion is characterized by a high level of stringency as it necessitates the optimization of the maximum average detection delay over all feasible observation paths leading up to and following the changepoint. This rigorous approach often produces conservative outcomes. However, it is worth noting that the motivation behind the use of min-max optimization can be linked to the Kullback-Leibler divergence. This has notably been the case in the works of Moustakides \cite{moustakides2004optimality} for detecting disorder occurring in the drift of Ito processes. In our case, for $\tau\in \mathcal{T}_{\pi}$ and $0\leq \theta \leq \tau$, equation \ref{llr} yields~:
\begin{equation}
\begin{split}
\mathbb{E}^{\theta}\left[\log\left(\frac{\mathrm{d}\mathbb{P}_\theta}{\mathrm{d}\mathbb{P}_{\infty}}\right)\big\lvert \mathcal{F}_\theta\right] &= \log \left(\rho\right) \mathbb{E}^{\theta}\left[\sum_{i=1}^D \left(N^i(\tau) - N^i(\theta)\right)-\frac{\rho - 1}{\log \left(\rho\right) }\left(\Lambda^{i}(\tau)-\Lambda^{i}(\theta)\right)\big\lvert \mathcal{F}_\theta\right] \\&=\left(\log \left(\rho\right)-\rho+1\right)\sum_{i=1}^D\mathbb{E}^{\theta}\left[\left(N^i(\tau)-N^i(\theta)\right)^{+} \big\lvert \mathcal{F}_\theta\right]
\end{split}
\end{equation}
Hence, the following minimax formulation of the detection problem is proposed~:
 \begin{equation}
    \begin{split}
     \inf_{\tau} \widetilde{C}(\tau) &=\inf_{\tau} \sup _{\theta \in[0, +\infty]} {\text { ess sup }} \sum_{i=1}^D\mathbb{E}^{\theta}\left[\left(N^i(\tau)-N^i(\theta)\right)^{+} \big\lvert \mathcal{F}_\theta\right] \\
    & \text{with the constraint}~~ \sum_{i=1}^D\mathbb{E}\left(N^i(\tau)\right) \geq \pi
    \end{split}
    \label{detection_prob_count}
\end{equation} 
El Karoui et al. \cite{ElKaroui2017} suggests another justification of this relaxation based on the Time-rescaling theorem (see Daley and Vere-Jones \cite{daley2007introduction}) which allows us to transform a Cox process into a Poisson process verifying~:
\begin{equation}   \mathbb{E}^{\theta}\left[\sum_{i=1}^D\left(N^i(\tau)-N^i(\theta)\right)^{+} \big\lvert \mathcal{F}_\theta\right]
    =\rho\sum_{i=1}^D \mathbb{E}^{\theta}\left[\int^{\tau}_{\tau \wedge \theta}\lambda^i(s) \mathrm{d}s\big\lvert \mathcal{F}_{\theta}\right] 
\end{equation}
Heuristically, this is equivalent to the formulation referenced in \ref{detection_prob_count}. Indeed, $\sum_{i=1}^D \mathbb{E}^{\theta}\left[\int^{\tau}_{\tau \wedge \theta}\lambda^i(s) \mathrm{d}s\big\lvert \mathcal{F}_\theta\right] = \left(\sum_{i=1}^D\lambda^{i,cst}\right) \mathbb{E}^{\theta}\left[\left(\tau - \theta\right)^+\big\lvert \mathcal{F}_\theta\right]$ in the case of a Poisson process with constant intensity.
\begin{rque}
    The formulation \ref{detection_prob_count} is particularly useful in finance trading applications, where it is more advantageous to focus solely on transactions/events that occurred, rather than the state of the market at every moment.
\end{rque}
\begin{rque}
Another alternative is to consider the criterion $\mathbb{E}^{\theta}\left[\sup_{1\leq i \leq D}\left(N^i(\tau)-N^i(\theta)\right)^{+} \big\lvert \mathcal{F}_\theta\right]$, which would be to minimize the worst possible detection delay among all processes $(N^i(t))_{1 \leq i \leq D}$. Since the Lorden criterion is already a conservative measure, adding another criteria function could result in excessive restrictions that may cause significant delays in the detection process. Therefore, it is preferable to refrain from using additional criteria functions. Notably, the Lorden criterion produces larger values compared to the criteria proposed by Pollak \cite{pollak} or Shiryaev \cite{shiryaev2009stochastic}, which further supports its sufficiency in detecting changes (see \cite{Shiryaev2019} p. 9-15).
\end{rque}
Given that the likelihood ratio test is considered the most powerful test for comparing simple hypotheses according to the Neyman-Pearson lemma, we will utilize $\rho^{\sum_{i=1}^{D} U^i(t) - U^i(\theta)}$ to assess the shift in this distribution. The log-likelihood ratio $\sum_{i=1}^D U^i(t) - U^i(\theta)$ usually exhibits a negative trend before a change and a positive trend after a change (see section 8.2.1 of Tartakovski et al. \cite{Tartakovsky2014} for more details). Consequently, the key factor in detecting a change is the difference between the current value of $\sum_{i=1}^{D} U^i(t) - U^i(\theta)$ and its minimum value. 
As a result, we are prompted to introduce the CUSUM process.
\begin{defi}
\label{def_cusum}
Let $\rho \ne 1$, $m>0$, $0\leq t \leq T$ and $U = \sum_{i=1}^DN^i - \beta \Lambda^i = \sum_{i=1}^DU^i$. The CUSUM processes are defined as the reflected processes of $U$~:
\begin{equation}
    \begin{split}
        \hat{U}(t)= U(t) - \inf _{0 \leq s \leq t} U(s), ~if~\rho>1 \\
        \widetilde{U}(t)= \sup _{0 \leq s \leq t} U(s) - U(t), ~if~\rho<1
    \end{split}
\end{equation}
 The stopping times for the CUSUM procedures are given by~: 
$$
\hat{T}_{\mathrm{C}}=\inf \left\{t: \hat{U}(t)>m\right\}, \quad \widetilde{T}_{\mathrm{C}}=\inf \left\{t: \widetilde{U}(t)>m\right\}
$$
where the infimum of the empty set is $+\infty$.
\end{defi}
\begin{rque}
We note that $\rho \beta(1 / \rho)=\beta(\rho)$. Therefore, for $t>0$,
\begin{align*}
    \sup _{\theta<t} \ell_{t, \theta} 
    &\leq \ell^C_{t,\theta}:= \begin{cases}
    \log \left(\rho\right)\left(\sum_{i=1}^D U^i(t) -\inf _{\theta<t}U^i(\theta) \right) & \text{if
    ~} \rho>1\\
    \log \left(\frac{1}{\rho}\right)\left(\sum_{i=1}^D \sup _{\theta<t}U^i(\theta) -U^i(t)\right) & \text{if
    ~} \rho<1
\end{cases}
\end{align*}
and
\begin{align*}
   \left\{t: \sup _{\theta<t} \ell_{t,\theta}>m\right\} \subseteq \left\{t: \ell^C_{t, \theta}>m\right\}
\end{align*}
Consequently, initial lower bounds can be derived for the stopping times $\hat{T}_{\mathrm{C}}$ and $\widetilde{T}_{\mathrm{C}}$ based on the reflected CUSUM processes associated with each component $N^i$.
\end{rque}
    \begin{rque}
    Considering the increasing event times $\left\{\tau_i, i=1,2, \ldots\right\}$ of the process $\sum_{i=1}^DN^i$, and for any fixed $t\geq 0$ and $k\geq 1$ where $t>t_k$, the following relationships hold (see Lemma 4.1 in Wang et al. \cite{Wang2021})~: $\sup _{\tau_k<s \leq \min \left\{\tau_{k+1}, t\right\}} \ell_{t, s}=\lim _{s \rightarrow \tau_k^+} \ell_{t, s}= \ell_{t, \tau_k^{+}}$, and $\sup _{0 \leq s \leq \tau_1} \ell_{t, s}=\ell_{t, 0}$ . This implies for $k>0$ and $t>\tau_k$ that~: 
\begin{align*}
   \sup _{\theta<t} \ell_{t, \theta} &= \max \{\sup _{0<s \leq \tau_{1}} \ell_{t, s}, \sup _{\tau_1<s \leq \tau_{2}} \ell_{t, s}, ~\dots, \sup _{\tau_k<s \leq \min \left\{\tau_{k+1}, t\right\}} \ell_{t, s}\}\\
   &= \max \{\ell_{t, 0}, \ell_{t, \tau_1^{+}}, ~\dots~, \ell_{t, \tau_k^{+}}\} 
\end{align*}
This allows us to extract a more time-efficient method for calculating the optimal stopping times $\widetilde{T}_{\mathrm{C}}$ and $\hat{T}_{\mathrm{C}}$. To do this, we formulate the log-likelihood ratio recursively~:
\begin{align*}
   \ell_{t,\tau_{n+1}^+} &= \sum_{i=1}^D U^i(t) - U^i(\tau_{n+1}^+) \\&= \sum_{i=1}^D U^i(t) - U^i(\tau_{n+1}^+) + U^i(\tau_{n}^+) - U^i(\tau_{n}^+)\\
   &= \ell_{t,\tau_{n}^+} - \sum_{i=1}^D\left( N^i(\tau_{n+1}^+) - N^i(\tau_{n}^+)\right) + \beta\left(\rho\right)\sum_{i=1}^D \int_{\tau_{n}^+}^{\tau_{n+1}^+}\lambda^i(s) \mathrm{d}s\\
   &= \ell_{t,\tau_{n}^+} - \sum_{i=1}^D\left( N^i(\tau_{n+1}^+) - N^i(\tau_{n}^+)\right) + \beta\left(\rho\right)\sum_{i=1}^D\left( \Lambda^i(\tau_{n+1}^+)-\Lambda^i(\tau_{n}^+)\right) 
\end{align*}
    
\end{rque}
One noteworthy characteristic of $\sum_{i=1}^D N^i(\widetilde{T}_{\mathrm{C}})$ and $\sum_{i=1}^D N^i(\hat{T}_{\mathrm{C}})$ is that their conditional expectations can be explicitly calculated based on their initial values and the parameter $m$. As a result, the theorems \ref{ARL_rho_inf_1} and \ref{ARL_rho_inf_2} offer a valuable tool for controlling the ARL constraint through the parameter $m$, i.e, expressing $\sum_{i=1}^D\mathbb{E}\left(N^i(\widetilde{T}_{\mathrm{C}})\right)$ and $\sum_{i=1}^D\mathbb{E}\left(N^i(\hat{T}_{\mathrm{C}})\right)$ as tractable formulas for $\rho>1$ and $\rho<1$.
\begin{theo}
Let $\widetilde{T}_{\mathrm{C}}$ the CUSUM stopping time defined by $\widetilde{T}_{\mathrm{C}}=\inf \left\{t \leq T: \widetilde{U}(t)>m\right\}$ for $\rho<1$.
Assume that the intensities $\lambda^i$ of the processes $N^i$ are càdlàg, $\forall i \in \{1,\dots,D\}$. The constraint on the Average Run Length (ARL) in $\widetilde{T}_{\mathrm{C}}$ is equal to~:
\begin{equation}
\begin{split}
    \mathbb{E}_y\left(\sum_{i=1}^D N^i(\widetilde{T}_{\mathrm{C}})\right) &=\int_y^{m} \frac{1}{\beta(\rho)} \sum_{k=0}^{\lfloor x\rfloor} \frac{(-1)^k}{k !}((x-k) / \beta(\rho))^k \exp ((x-k) / \beta(\rho)) \mathrm{d}x 
    \end{split}
\end{equation}
with $\mathbb{P}_y = \mathbb{P}(.\big\lvert \widetilde{U}(0) = y)$.
\label{ARL_rho_inf_1}
\end{theo}
\begin{proof}
    Proof postponed to appendix.
\end{proof}
\begin{theo}
Let $\hat{T}_{\mathrm{C}}$ the CUSUM stopping time defined by $\hat{T}_{\mathrm{C}}=\inf \left\{t \leq T: \hat{U}(t)>m\right\}$ for $\rho>1$.
Assume that the intensities $\lambda^i$ of the processes $N^i$ are càdlàg, $\forall i \in \{1,\dots,D\}$. The constraint on the Average Run Length (ARL) in $\hat{T}_{\mathrm{C}}$ is equal to~:
\begin{equation}
\begin{split}
    \mathbb{E}_v\left(\sum_{i=1}^D N^i(\hat{T}_{\mathrm{C}})\right) &= W(m-v)\frac{W(m)}{W^{\prime}(m)}-\int_0^{m-v} W(y) \mathrm{d} y,\\ \quad  \mathbb{E}_{m^-}\left(\sum_{i=1}^D N^i(\hat{T}_{\mathrm{C}})\right) &= \frac{W(m)}{\beta W^{\prime}(m)}
    \end{split}
\end{equation}
with $\mathbb{P}_v = \mathbb{P}(.\big\lvert \hat{U}(0) = v)$, $W(x)=\frac{1}{\beta(\rho)} \sum_{k=0}^{\lfloor x\rfloor} \frac{(-1)^k}{k !}((x-k) / \beta(\rho))^k \exp ((x-k) / \beta(\rho))$ and $\int_0^x W(y) \mathrm{d} y = \sum_{k=0}^{\lfloor x\rfloor}\left(e^{(x-k) / \beta(\rho)}\left(\sum_{i=0}^k \frac{(-1)^j}{j !}((x-k) / \beta(\rho))^j\right)-1\right)$.
\label{ARL_rho_inf_2}
\end{theo}
\begin{proof}
    Proof postponed to appendix.
\end{proof}
Figures \ref{fig:gm(0)} and \ref{fig:hm(0)} give a visual image of the values taken by Average Run Length (ARL) when varying the ratio $\rho$ and the parameter $m$.
 \begin{figure}[H]

        \centering
    \begin{subfigure}[b]{0.45\textwidth}
        \centering
        \captionsetup{justification=centering}
        \includegraphics[width=\textwidth]{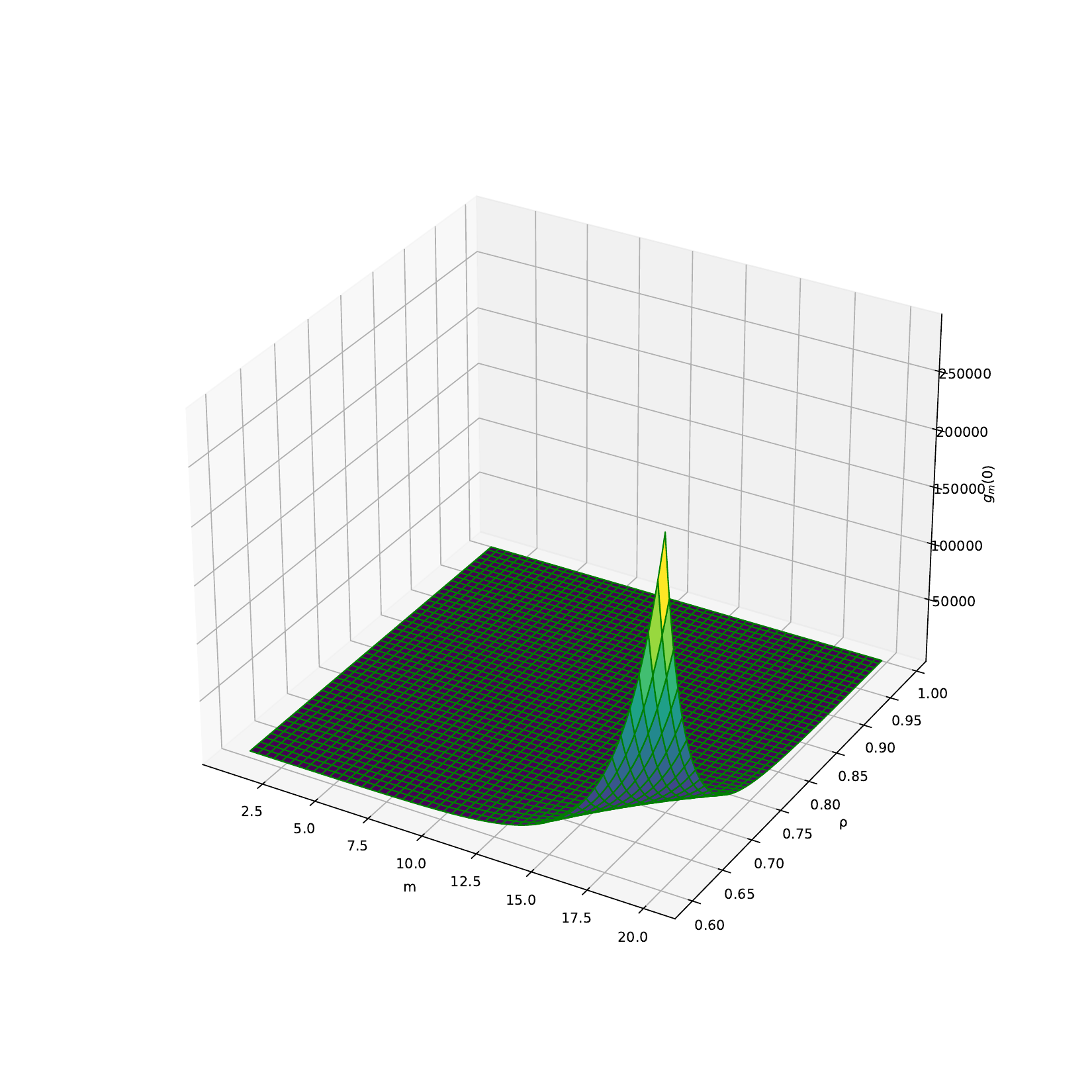}
        \caption{Average Run Length for $\rho\in\left[0.6,1\right[$ \\and $m\in\left[1,15\right]$}
        \label{fig:gm(0)}
    \end{subfigure}
    \begin{subfigure}[b]{0.45\textwidth}
        \centering
        \captionsetup{justification=centering}
        \includegraphics[width=\textwidth]{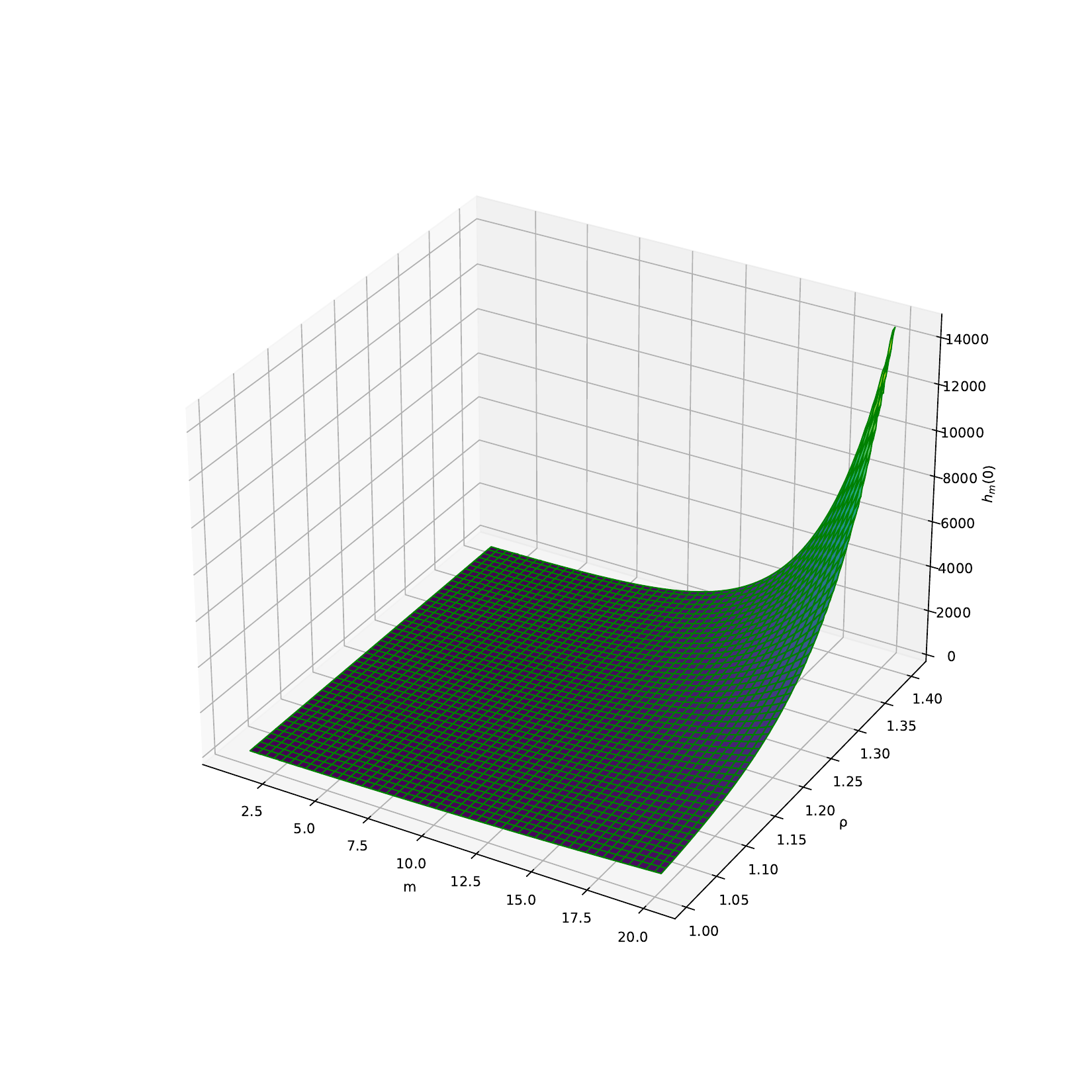}
        \caption{Average Run Length for $\rho\in\left]1,1.4\right]$ \\and $m\in\left[1,15\right]$}
        \label{fig:hm(0)}
    \end{subfigure} 
\end{figure}
Two subsequent theorems \ref{opt_decrease} and \ref{opt_increase} substantiate the optimality of the CUSUM procedure by separately considering the cases where $\rho<1$ and $\rho>1$. These theorems establish that the CUSUM stopping time attains the lower bound of the Lorden criterion $\widetilde{C}(\tau)=\sup _{\theta \in[0, +\infty[} {\text { ess sup }} \mathbb{E}^{\theta}\left[(\tau-\theta)^{+} \big\lvert \mathcal{F}_{\theta}\right]$ in each case, thereby demonstrating its optimality in detecting changes or shifts in sequential observations for both scenarios.
\begin{theo}
\label{opt_decrease}
Let $\widetilde{T}_{\mathrm{C}} =\inf \left\{t \leq T: \widetilde{U}(t)>m\right\}$ denote the CUSUM stopping time. Then, $\widetilde{T}_{\mathrm{C}}$ is proven to be the optimal solution for \ref{detection_prob_count} for $\rho<1$ where
$\sum_{i=1}^D\mathbb{E}\left(N^i(T)\right) = \sum_{i=1}^D\mathbb{E}\left(N^i(\widetilde{T}_{\mathrm{C}})\right)$.
\end{theo}
\begin{proof}
    Proof postponed to appendix.
\end{proof}
\begin{theo}
\label{opt_increase}
Let $\hat{T}_{\mathrm{C}} =\inf \left\{t \leq T: \hat{U}(t)>m\right\}$ denote the CUSUM stopping time. Then, $\hat{T}_{\mathrm{C}}$ is proven to be the optimal solution for \ref{detection_prob_count} for $\rho>1$ where
$\sum_{i=1}^D\mathbb{E}\left(N^i(T)\right) = \sum_{i=1}^D\mathbb{E}\left(N^i(\hat{T}_{\mathrm{C}})\right)$.
\end{theo}
\begin{proof}
    Proof postponed to appendix.
\end{proof}

\section{Experimental Results}
We proceed to presenting the experimental findings related to our methodology. To begin with, we provide a comprehensive outline of the final form of the kernel in the Marked Hawkes Process. Subsequently, we conduct an assessment of its goodness of fit to the BNP Paribas data. Following that, we perform a thorough analysis to evaluate the methodology's effectiveness in detecting changes in liquidity regimes.
\label{experimental_results}
\subsection{Data Description}
We use Refinitiv limit order book tick-by-tick data (timestamp to the nanosecond) for BNP-Paribas (BNP.PA) stock from 01/01/2022 to 31/12/2022 to fit our Marked Hawkes models and carry out our research on disorder detection. This data give us access to the full depth of the order book and include prices and volumes on the bid and ask.
Throughout the study, we only consider trades flagged as ’normal’ trades. In particular, we did not consider any block-trade or off-book trade in the following.
We restrict our data time-frame to the period of the day not impacted by auction phases (from 9:30 PM to 17 PM, Paris time). Only trades ranging from limit one to four are considered here. The selection is motivated by the relatively rare instances of trade-through events when the limit exceeds four. Indeed, the comprehensive analysis of transactions involving the BNP stock reveals that approximately 4.89\% of the overall transaction count led to at least one trade-through event. While these trade-through events bear significance, it is pertinent to note that trades-through with limit surpassing four constitute a low percentage of merely 0.0126\% within the total transaction count during that specific timeframe. These findings align with the research conducted by Pomponio et al \cite{Pomponio2010}.\\
The VSTOXX® Volatility Index serves as a tool for assessing the price dynamics and liquidity of stocks comprising the EURO STOXX 50® Index throughout the year 2022. This index will facilitate the identification of appropriate days for conducting our comparative analysis. Figure \ref{fig:vstoxx} shows that day 05/01/2022 corresponds to the day with the lowest volatility, while 04/03/2022 is the day with the highest volatility. We will consider 31/08/2022 as the reference day where the volatility index is set equal to the average of the VSTOXX® ($index_{vol}\approx 27$) over the entire year and where the volatility index on the previous and following days does not deviate greatly from the latter.
\begin{figure}[H]
    \centering
        \includegraphics[width=1\textwidth]{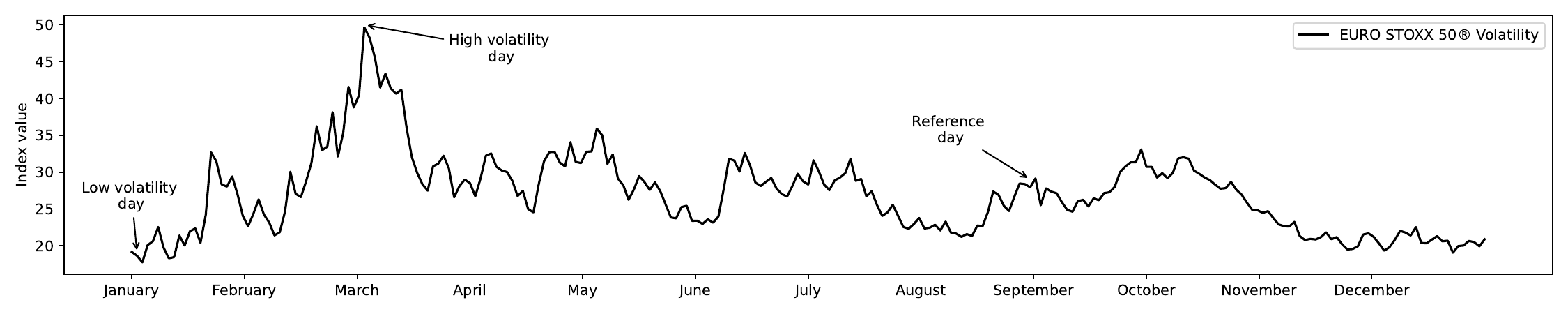}
       \caption{EURO STOXX 50® Volatility index (VSTOXX®) from 01/01/2022 to 31/12/2022.}
       \label{fig:vstoxx}
\end{figure}
\begin{rque}
     It is worth noting that obtaining data on trade-throughs, as defined in this study, does not require the reconstruction of the order flow of the order book, as previously reported by Toke \cite{Toke_2016}. 
\end{rque}
\subsection{Estimation}
The current state of the art regarding the estimation of Multivariate Hawkes processes revolves around methods~: Method of moments, Maximum likelihood estimation (MLE) and Least squares estimation.
  A number of papers have reviewed each method. For exemple, Cartea et Al. \cite{Cartea2021} construct an adaptive
stratified sampling parametric estimator of the gradient of the least squares estimator. Hawkes \cite{Hawkes1971} used the techniques developed by Bartlett to analyze the spectra of point processes that where which were later retrieved by Bacry et Al. \cite{BacryDayri2011} to propose a non-parametric estimation method for stationary Multivariate Hawkes processes with symmetric kernels. Bacry et Al. \cite{bacrymuzy2016} relaxed these assumptions, except for stationarity, and showed that the Multivariate Hawkes process parameters solve a system of Wiener-Hopf equations. In their work, Daley and Vere-Jones \cite{daley2007introduction} put forth the idea of maximizing the log-likelihood of the sample path. We adopt this approach due to the efficient computation of the likelihood function in our case, which consequently eases the estimation of the model's parameters.

The subsequent analysis involves parameter estimation within our model, utilizing the maximum likelihood method. We define an observation period denoted as $[0, T]$, which corresponds to the time span during which empirical data was collected. To construct the likelihood function \ref{likelihood}, we initiate by defining the probability densities $f_A$ and $f_B$ of the marks. Figure \ref{fig:ccdf_volumes} provides an illustrative representation of the distribution of volume of trades-through on the bid and ask sides of the order book for BNP Paribas stock during May 2022.
\begin{figure}[H]
    \centering
    \begin{subfigure}[b]{0.47\textwidth}
        \centering
        \includegraphics[width=\textwidth]{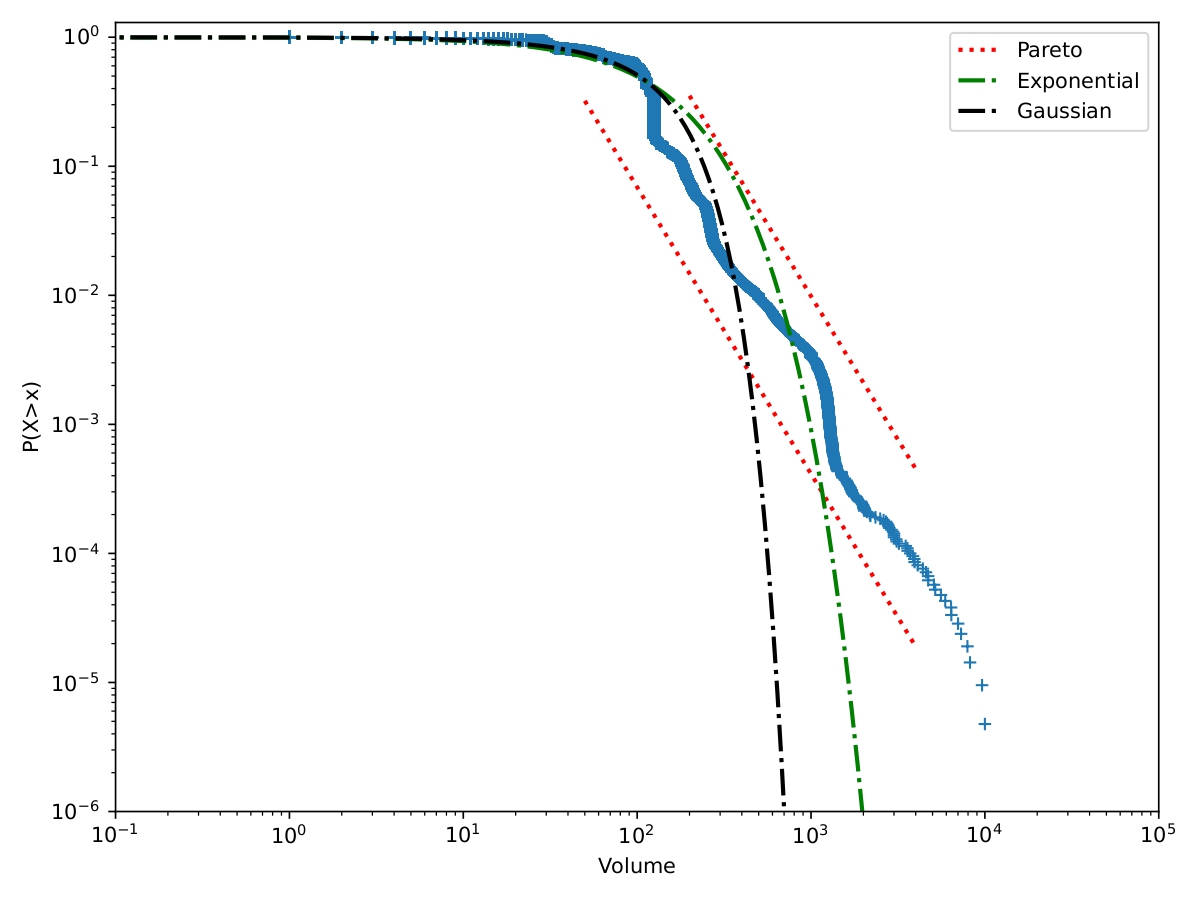}
    \end{subfigure}
    \begin{subfigure}[b]{0.47\textwidth}
        \centering
        \includegraphics[width=\textwidth]{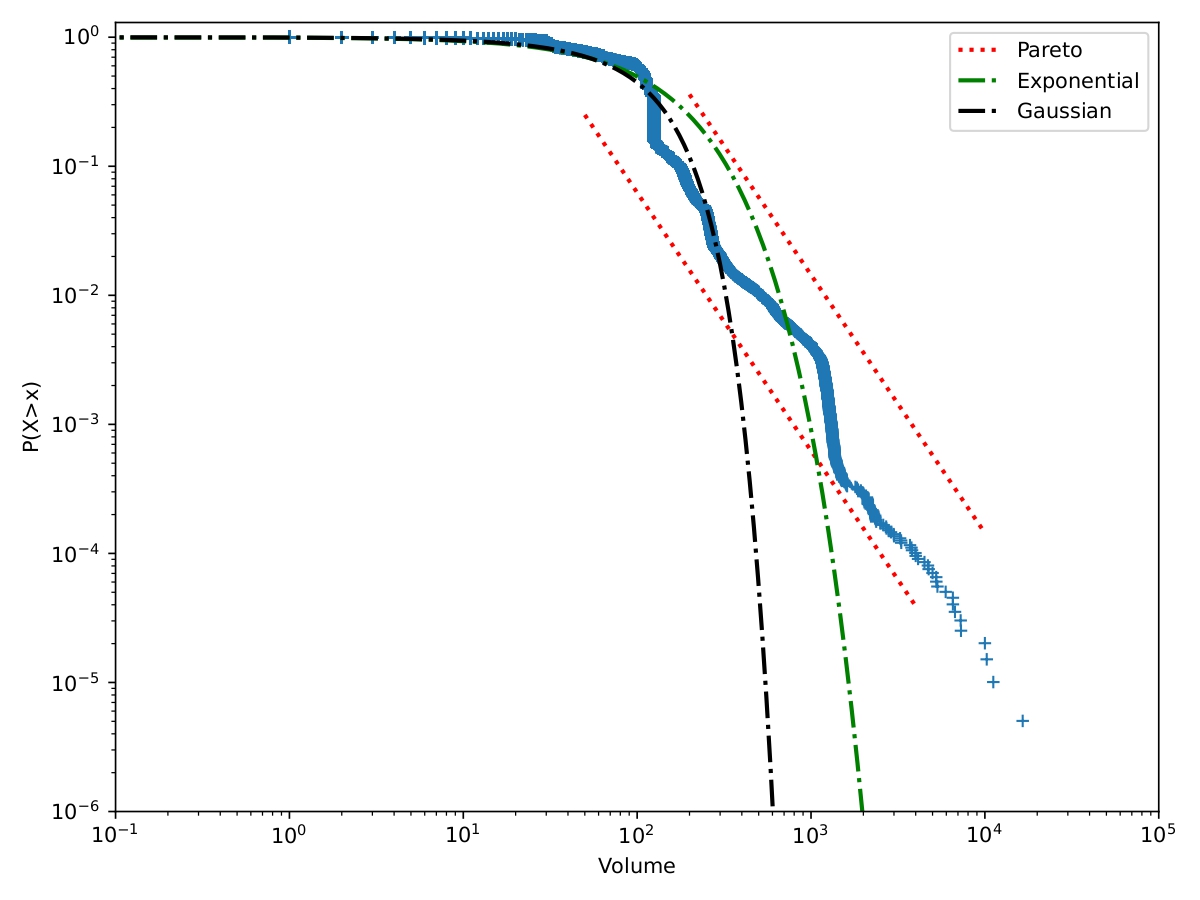}
    \end{subfigure}
       \caption{Log-log scale Complementary Cumulative Density Function (CCDF) of the volume of trades-through on the bid (left figure) and ask (right figure) sides from May 2nd, 2022, 09:30:00, to May 31, 2022, 17:00:00. The dashed red lines correspond to the CCDF of a Pareto distribution with parameter $\eta = 0.45$, the dashed black line corresponds to the CCDF of a Gaussian distribution, and the dashed green line corresponds to the CCDF of an Exponential distribution with a coefficient equal to $0.01$.}
       \label{fig:ccdf_volumes}
\end{figure}
The observations exhibit a good fit with an exponential distribution. We will characterize the impact of volume on trades-through using a normalized power decay function\footnote{Also referred to as \textit{Omori's Law} in the context of earthquake modelling.}. This means that $g_i(v)=c_i\times v^{\eta_i}$ with $\eta_i$ as a non-negative real and $c_i$ as a normalizing constant, $i = A,B$. By normalizing the impact function using its first moment (as specified in \ref{normalizing_cond}), we derive the expression: $$g_i(v)=\frac{v^{\eta_i}}{\mathbb{E}\left(v^{\eta_i}\right)}=\frac{\beta_i^{\eta_i}}{\Gamma(1+\eta_i)}v^{\eta_i}$$ 
In this equation, $\beta_i>0$ denotes the parameter of the exponential distribution of volumes, defined on the half-line $]0,+\infty[$, with $i = A,B$.

Furthermore, the baselines $\mu^A$ and $\mu^B$ are defined as a non-homogeneous piecewise linear continuous functions $t\mapsto \sum_{i=1}^{14}\mu^A_i\mathbbm{1}_{\frac{i-1}{14}T<t\leq \frac{i}{14}T}$ and $t\mapsto \sum_{i=1}^{14}\mu^B_i\mathbbm{1}_{\frac{i-1}{14}T<t\leq \frac{i}{14}T}$ over a subdivision of the time interval $\left[0, T\right]$ into half-hour intervals with $\mu_i^A>0$ and $\mu_i^B>0$ for all $i\in\{1,\dots,14\}$. This adaptive approach allows the model to effectively accommodate fluctuations in intraday market activity.
 \begin{rque}
The model under consideration aligns with the one proposed by Toke et al. \cite{toke2012modelling}, specifically when $g_i \equiv 1$ where $i = A,B$.
\end{rque} 
We proceed by computing the likelihood function $L\left(\{N(t)\}_{t \leq T}\right)$, which represents the joint probability of the observed data as a function of the underlying parameters. 
\begin{prop}[Likelihood]
    Consider a marked Hawkes process $N^{A\times B}=\left(N^A, N^B\right)$ with intensity $ \lambda^{A\times B}=\left(\lambda^A, \lambda^B\right)$ and compensator $ \Lambda^{A\times B}=$ $\left(\Lambda^A, \Lambda^B\right)$. We denote the probability densities of the marks by $(f_A,f_B)$. The log-likelihood of $\{N^{A\times B}(t)\}_{t \leq T}$ relative to the unit rate Poisson process $\ell$ is equal to~:
    \begin{equation}
    \label{likelihood}
    \begin{split}
        \ell &=T-\sum_{i\in \{A,B\}}\left( \int_0^T\mu^{i}\left(u\right)du+\sum_{j\in \{A,B\}} \sum_{\tau_{k^\prime}^j<T} \frac{\alpha_{i j}}{\beta_{i j}}\frac{\beta_j^{\eta_j}}{\Gamma(1+\eta_j)}v_{k^\prime}^{\eta_j}\left(1 - e^{-\beta_{i j}\left(T-\tau_{k^\prime}^j\right)} \right)\right)\\&+\sum_{i\in \{A,B\}}\sum_{ \tau^i_k \leq T} \ln \left(\mu^{i}\left(\tau^i_k\right)+\sum_{j\in \{A,B\}} \sum_{\tau_{k^\prime}^j<\tau^i_k} \alpha_{i j}\frac{\beta_j^{\eta_j}}{\Gamma(1+\eta_j)}v_{k^\prime}^{\eta_j}e^{-\beta_{i j}\left(\tau^i_k-\tau_{k^\prime}^j\right)}\right) +\sum_{i\in \{A,B\}}\sum_{\tau^i_k \leq T} \ln \left(f_i\left(v_{k}^i\right)\right)
        \end{split}
    \end{equation}
    where $f_i\left(v\right) = \beta_i e^{-\beta_i v}$, $i=A,B$.
\end{prop}
\begin{proof}
    The computation of the log-likelihood of a multidimensional Hawkes process involves summing up the likelihood of each coordinate $i\in \{A,B\}$~:
    \begin{equation}
    \begin{aligned}
        \ell &:= \ln \left(\frac{L\left(\{N(t)\}_{t \leq T}\right)}{L_0}\right)\\&= T + \sum_{i \in \{A,B\}} \ln L\left(\left\{N^i(t)\right\}_{t \leq T}\right)
        \end{aligned}
\end{equation}
According to Theorem 7.3.III of \cite{daley2007introduction}, one can express the likelihood of a realization $\left(\tau^i_1,v^i_1\right)$,$\dots$, $(\tau^i_{N^i_g(T)},v^i_{N^i_g(T)})$ of $N^i(t)$ in the following form~:
\begin{equation}
    L\left(\{N^i(t)\}_{t \leq T}\right) = \left[\prod_{k=1}^{N^i_{\mathrm{g}}(T)} \lambda_{\mathrm{g}}^i\left(\tau^i_k\right)\right]\left[\prod_{k=1}^{N^i_{\mathrm{g}}(T)} f_i\left(v^i_k \right)\right] \exp \left(-\int_0^T \lambda_{\mathrm{g}}^i(u) \mathrm{d} u\right)\end{equation}
The likelihood of the process $N^{A\times B}$ relative to the unit rate Poisson process is thus expressed as follows~:
    \begin{align*}
        \ell &:= \ln \left(\frac{L\left(\{N^{A\times B}(t)\}_{t \leq T}\right)}{L_0}\right)\\&= T + \sum_{i \in \{A,B\}} \ln L\left(\left\{N^i(t)\right\}_{t \leq T}\right) \\&=T-\int_0^T\lambda^A_g\left(u\right) \mathrm{d} u+\int_0^T \ln \lambda^A_g\left(u\right) N^A(\mathrm{d}u\times \mathrm{d}v) + \int_0^T \ln f_A\left(v\right) N^A(\mathrm{d}u\times \mathrm{d}v)\\ &\quad\quad\quad -\int_0^T\lambda^B_g(u) \mathrm{d} u+\int_0^T \ln \lambda^B_g(u) N^B(\mathrm{d}u\times \mathrm{d}v) + \int_0^T \ln f_B\left(v\right) N^B(\mathrm{d}u\times \mathrm{d}v)\\
        &=T-\sum_{i\in \{A,B\}}\left( \int_0^T\mu^{i}\left(u\right)du+\sum_{j\in \{A,B\}} \sum_{\tau_{k^\prime}^j<T} \frac{\alpha_{i j}}{\beta_{i j}}\frac{\beta_j^{\eta_j}}{\Gamma(1+\eta_j)}v_{k^\prime}^{\eta_j}\left(1 - e^{-\beta_{i j}\left(T-\tau_{k^\prime}^j\right)} \right)\right)\\&+\sum_{i\in \{A,B\}}\sum_{ \tau^i_k \leq T} \ln \left(\mu^{i}\left(\tau^i_k\right)+\sum_{j\in \{A,B\}} \sum_{\tau_{k^\prime}^j<\tau^i_k} \alpha_{i j}\frac{\beta_j^{\eta_j}}{\Gamma(1+\eta_j)}v_{k^\prime}^{\eta_j} e^{-\beta_{i j}\left(\tau^i_k-\tau_{k^\prime}^j\right)}\right) +\sum_{i\in \{A,B\}}\sum_{\tau^i_k \leq T} \ln \left(f_i\left(v_{k}^i\right)\right)\\&= T+\ell_{\lambda_g} + \ell_{v},
\end{align*}
where $f_A$ (resp. $f_B$) is the density function relative to the distribution of the volumes and the pre-factor $c_j$ is the normalizing constant. 
\end{proof}
\begin{rque}
Equation \ref{decomp_intensity} allows us to decompose the log-likelihood into the sum of two terms. The first term being generated by the ground intensity and the second by the conditional distribution of marks. Note that the absence of any shared parameter between $\ell_{\lambda_g}$ and $\ell_{v}$, i.e, between $f_i$ and $\lambda_g^i$ allows for the independent maximization of each term. 
\end{rque}
\subsection{Goodness-of-fit  analysis}
We now intend to analyze the quality of fit of our model. We begin by examining the stability of the obtained parameters. Figure \ref{fig:baseline_params} displays the boxplots related to the baseline intensity parameters $(\mu_i^A)_{1\leq i \leq 14}$ and $(\mu_i^B)_{1\leq i \leq 14}$.
\begin{figure}[H]
    \centering
    \begin{subfigure}[b]{0.45\textwidth}
        \centering
        \includegraphics[width=\textwidth]{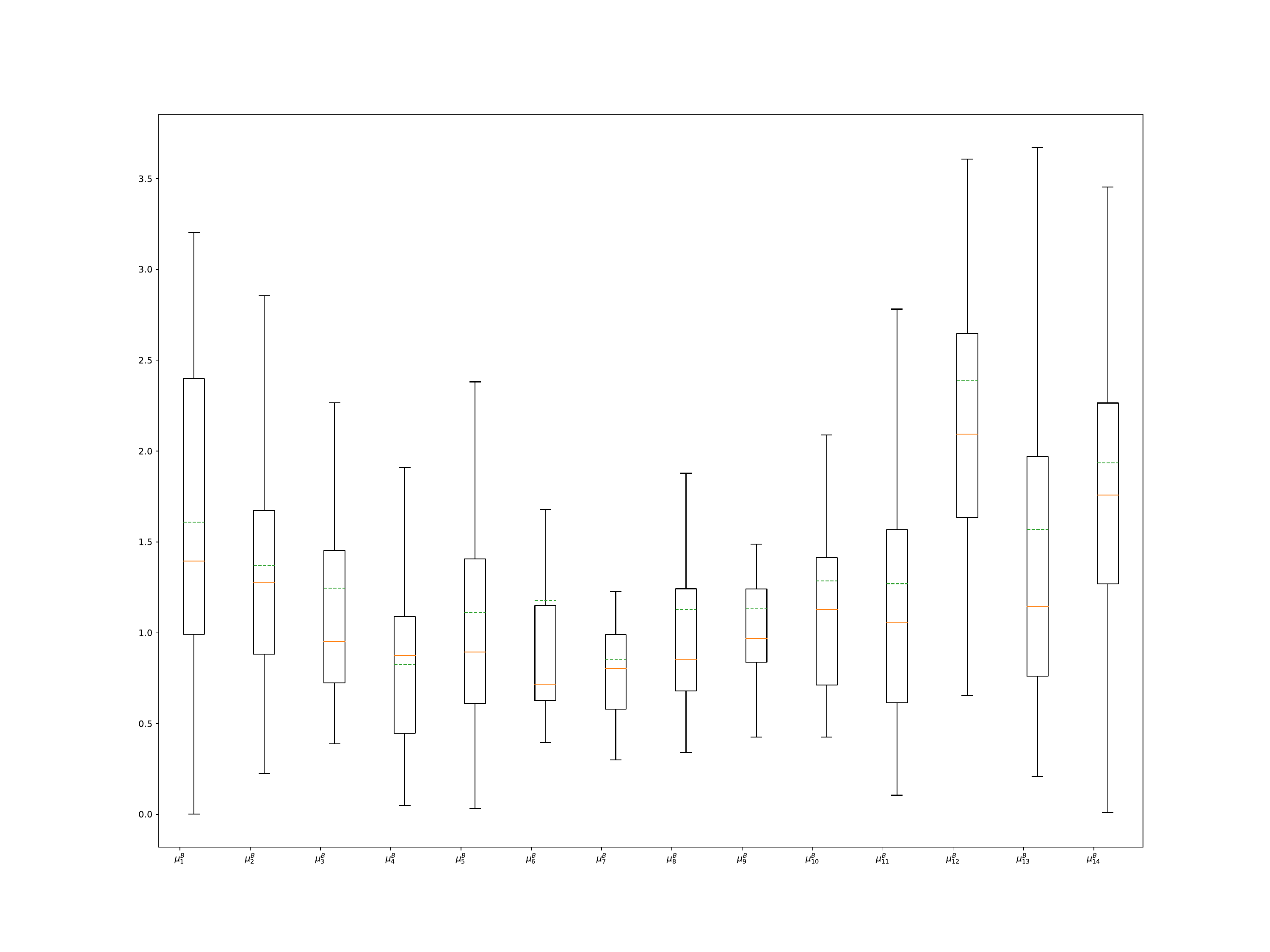}
    \end{subfigure}
    \begin{subfigure}[b]{0.45\textwidth}
        \centering
        \includegraphics[width=\textwidth]{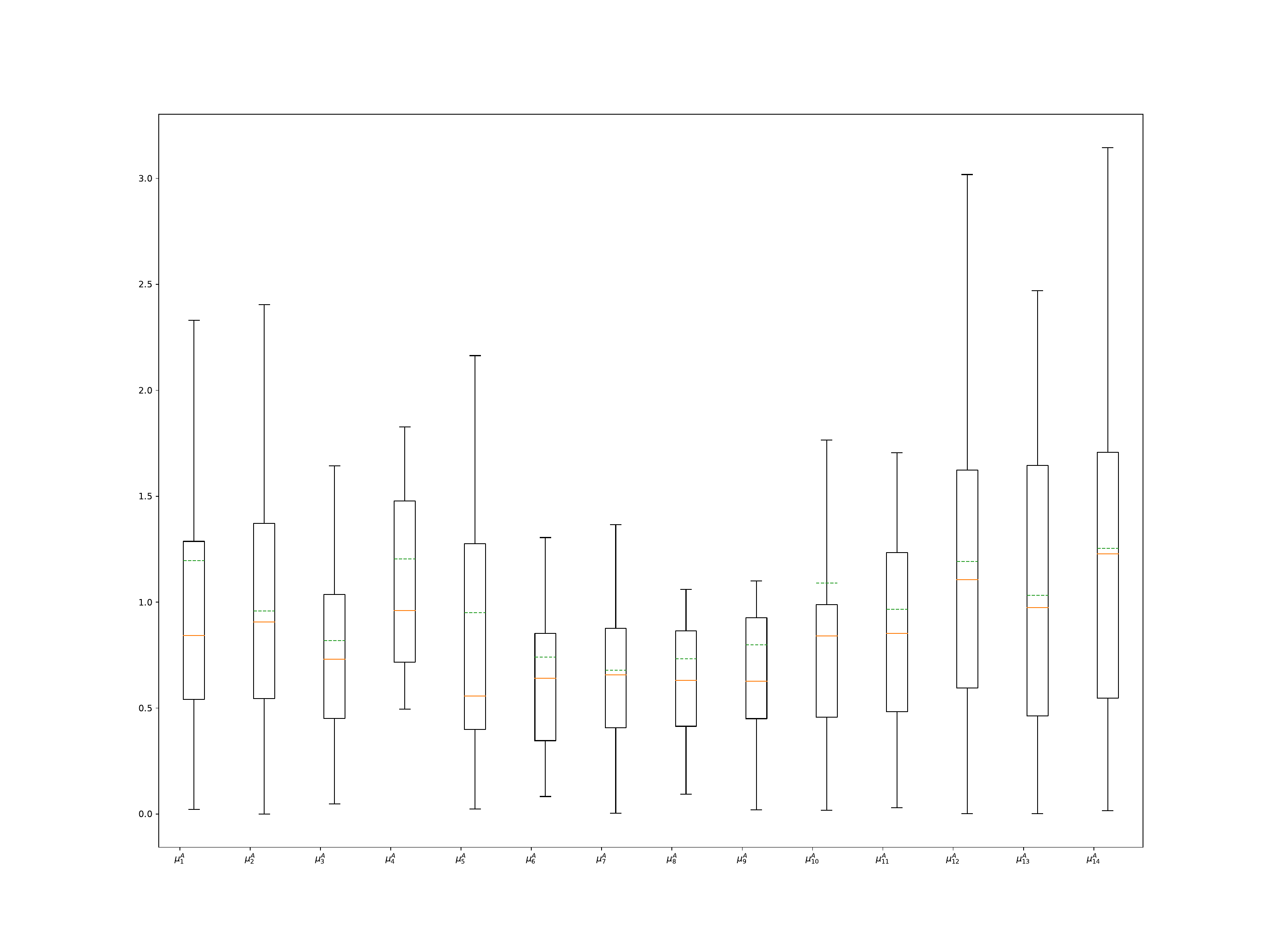}
    \end{subfigure}
       \caption{Tukey Boxplot of the values of the baseline intensity parameters for bid side (left figure) and ask side (right figure).}
       \label{fig:baseline_params}
\end{figure}
The exogenous part of the intensity function effectively captures the market activity profile with a noticeable U-shaped curve. To gain deeper insights, we further explore the remainder of the kernel. In Figure \ref{fig:endo_params}, we present boxplots for the values of the parameters $(\alpha_{ij})_{i,j\in\{A,B\}}$, $(\beta_{ij})_{i,j\in\{A,B\}}$, $(\eta_{i})_{i\in\{A,B\}}$, and $(\beta_{i})_{i\in\{A,B\}}$ related to the endogenous aspect of the kernel, providing valuable information about its self-excitation and mutual-excitation properties.
\begin{figure}[H]
        \centering
        \includegraphics[width=0.5\textwidth]{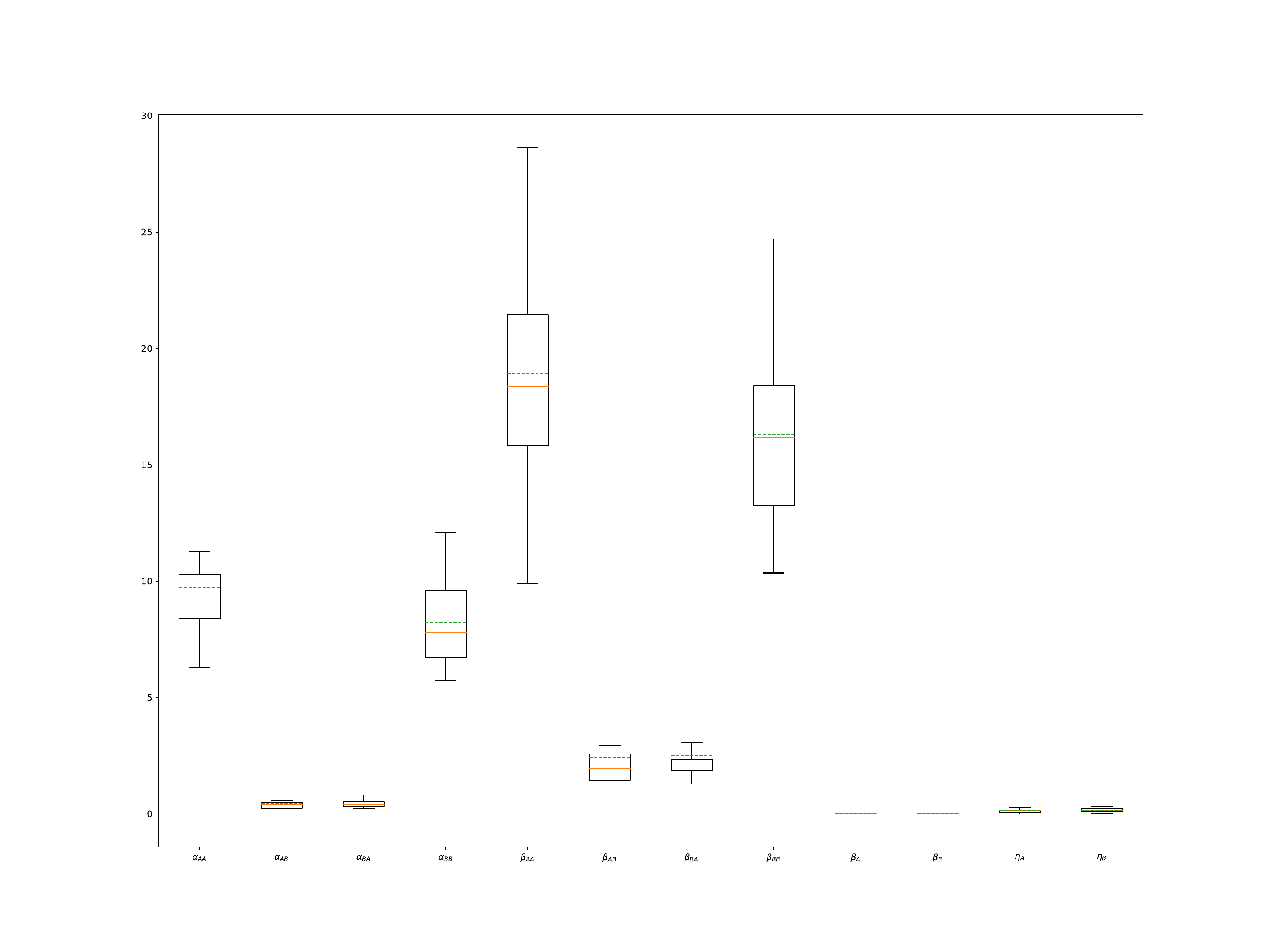}
       \caption{Tukey Boxplot of the values of the parameters $(\alpha_{ij})_{i,j\in\{A,B\}}$, $(\beta_{ij})_{i,j\in\{A,B\}}$, $(\eta_{i})_{i\in\{A,B\}}$ and $(\beta_{i})_{i\in\{A,B\}}$.}
       \label{fig:endo_params}
\end{figure}
We can observe that the values of $\beta_A$ and $\beta_B$ range between $0.009$ and $0.013$, which confirms our initial guess shown in Figure \ref{fig:ccdf_volumes}. Additionally, we notice that mutual excitation effects are weak compared to self-excitation effects, consistent with the findings reported by Toke et al \cite{toke2012modelling}. Furthermore, we observe that the parameters have low variance, and the conditions specified in \ref{non_exposition_cond} are consistently met with an average branching ratio of $0.83$. This indicates the stability of our model and suggests that we are operating in a sub-critical regime.

To assess the validity of our fit, we perform a residual analysis of our counting processes using the Time-Rescaling theorem (see Daley and Vere-Jones \cite{daley2007introduction}). Proposition \ref{res_analysis_recursive} enables us to recursively compute the distances $\Lambda^i\left(\tau_k^i\right)-\Lambda^i\left(\tau_{k-1}^i\right)$, reducing the complexity of the goodness-of-fit algorithm from $O(N^2)$ to $O(N)$, as outlined in Ozaki \cite{Ozaki1979}.

\begin{prop}
Consider a marked $D$-dimensional Hawkes process with exponential decays and an impact function of power law type, denoted as $(N^i)_{1\leq i \leq D}$. Let $\left\{\tau_k^i\right\}_{k\geq 0}$ represent the event times of the process $N^i$, and $N_g^i$ denote its related ground process for all $i \in {1,\dots,D}$. Then,
\begin{equation}
\begin{aligned}
V_k^i&:= \Lambda_g^i\left(\tau_{k}^i\right) -  \Lambda_g^i\left(\tau_{k-1}^i\right)\\&= \int_{\tau_{k-1}^i}^{\tau_k^i} \mu^i(s) \mathrm{d} s+\sum_{j=1}^D  \frac{\alpha_{i j}}{\beta_{i j}}\biggl[A^{i j}(k-1) \times\left(1-e^{-\beta_{i j}\left(\tau_k^i-\tau_{k-1}^i\right)}\right)\biggl.\\&\biggr.+\sum_{\tau_{k-1}^i \leq \tau_{k^{\prime}}^j<\tau_k^i}g_j(v_{k^{\prime}}^j)\left(1-e^{-\beta_{i j}\left(\tau_k^i-\tau_{k^{\prime}}^j\right)}\right) \biggr]
\end{aligned}
\end{equation}
where $A_{i j}(k)$ is defined recursively as
$$
A_{i j}(k)=e^{-\beta_{i j}\left(\tau_{k}^i-\tau_{k-1}^i\right)} A^{i j}(k-1)+\sum_{\tau_{k-1}^i \leq \tau_{k^{\prime}}^j<\tau_{k}^i} g_j\left(v_{k^{\prime}}^j\right)e^{-\beta_{i j}\left(\tau_{k}^i-\tau_{k^{\prime}}^j\right)}
$$
with initial condition:
$$
A_{i j}(0)=0
$$
and $\left\{V_k^1\right\}_{k\geq 0},\left\{V_k^2\right\}_{k\geq 0}, \ldots,\left\{V_k^D\right\}_{k\geq 0}$ are $D$ independent sequences of independent identically distributed exponential random variables with unit rate.
\label{res_analysis_recursive}
\end{prop}
We compare the deviation between the distribution of the transformed process and that of a stationary Poisson process of intensity equal to 1 through the Kolmogorov-Smirnov test and verify the independence of its observations through the Ljung–Box test. The Q-Q plot in Figure \ref{fig:qq_1feb_marked} shows a typical visual example of the fit of our model on the 01/02/2022 trades-through data. 

\begin{figure}[H]
    \centering
    \begin{subfigure}[b]{0.45\textwidth}
        \centering
        \includegraphics[width=\textwidth]{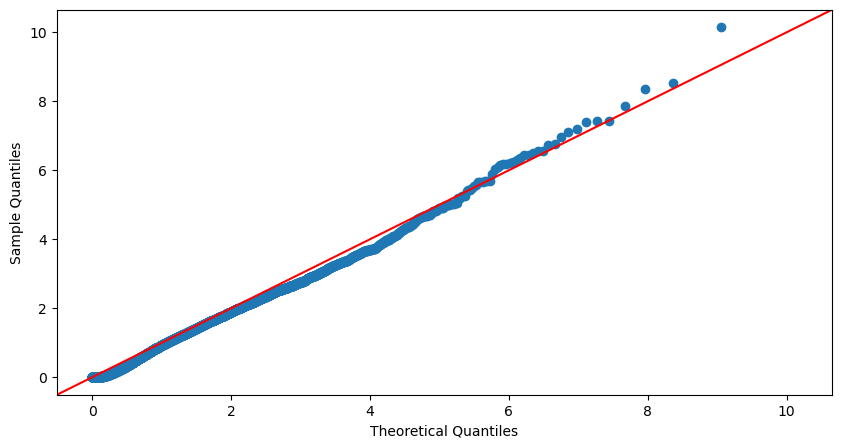}
    \end{subfigure}
    \begin{subfigure}[b]{0.45\textwidth}
        \centering
        \includegraphics[width=\textwidth]{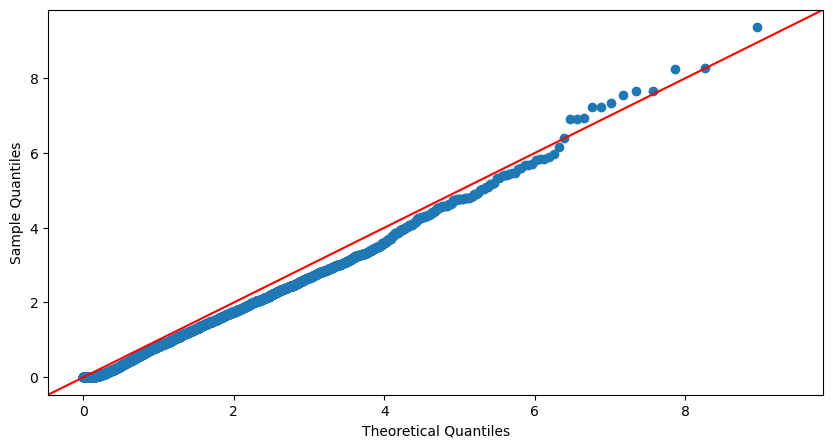}
    \end{subfigure}
       \caption{A depiction of a Q-Q plot showcasing the relationship between the theoretical quantiles, derived from an exponential distribution with a parameter of 1, and the sample quantiles generated through a Marked Hawkes model applied to Bid (or Ask) trades-through data for the BNP Paribas stock on 01/02/2022.}
       \label{fig:qq_1feb_marked}
\end{figure}
The outcomes of our calibration for the month of May 2022 indicate that our model successfully passes the Kolmogorov-Smirnov test 50.2\% of the time at a significance level of 5\%, 60.8\% of the time at a significance level of 2.5\%, and 82.5\% at a significance level of 1\%. Moreover, our model consistently passes the Ljung-Box test up to the twentieth term at a significance level of 5\%. Taking these factors into account, along with the notably convincing Q-Q plots and the presence of low-variance model parameters that meet stability criteria, we can assert that our model fits the market data effectively.



\subsection{Detection of shifts in liquidity regimes}
In this section, the results of the CUSUM procedure studied in section \ref{Sequential Change-Point Detection: CUSUM-based optimal stopping scheme} will be used to make a day-to-day comparative analysis of the liquidity state. Our approach involves selecting a specific day and treating the intensity of the Hawkes processes that model the associated trades-through as the baseline intensity of the financial instrument being analyzed. Sequential hypothesis testing is then performed based on this established intensity. In other words, this amounts to comparing the intensities of the trades-through processes related to two distinct days and thus to comparing the state of liquidity between two given days.

We will use the VSTOXX index and bid-ask spread as a marker/proxy of liquidity to assess the quality of our detection. The primary challenge of implementing our sequential hypothesis testing procedure through backtesting is the unavailability of data on disorder periods. Nonetheless, it is noticeable that the commencement of CME's open-outcry trading phase for major equity index futures at 2:30 PM Paris time, along with significant US macroeconomic news releases such as the ISM Manufacturing Index at 4 PM Paris time, leads to amplified market volatility in Europe. The impact of these events on BNP Paribas' stock prices will be utilized as an indicator to evaluate the detection's delay.\\
We start by fixing the parameters $m$ and $\rho$ as defined in section \ref{Sequential Change-Point Detection: CUSUM-based optimal stopping scheme} in order to define the hypotheses to be tested and to determine the average detection delays $\sum_{i =1}^D\mathbb{E}\left(N^i(\widetilde{T}_{\mathrm{C}})\right)$ and $\sum_{i =1}^D\mathbb{E}\left(N^i(\hat{T}_{\mathrm{C}})\right)$ which is relative to them. As the Average Run Length is an increasing function of $m$ for all $\rho>0$ and as Lorden's criterion is pessimistic (in the sense that it seeks to minimise the worst detection delay), it is preferable to choose $m$ small enough not to have a very large detection delay (i.e., second type error). This said, a too small $m$ will result in an increase of false alarms (i.e. first type error). A compromise must therefore be found between the two. We therefore propose to carry out our hypothesis tests for $\rho_{down} = 0.5$ and $\rho_{up} = 1.5$ and $m = 5$. Figures \ref{fig:gm(0)} and \ref{fig:hm(0)} show that these parameters yield $\sum_{i =1}^D\mathbb{E}\left(N^i(\widetilde{T}_{\mathrm{C}})\right) = 75.97$ and $\sum_{i =1}^D\mathbb{E}\left(N^i(\hat{T}_{\mathrm{C}})\right) = 60.4$. 
We initially compare the intensities of the trades through between the reference day 31/08/2022 and the day with the highest volatility 04/03/2022.
\begin{figure}[H]
    \centering
        \includegraphics[width=1\textwidth]{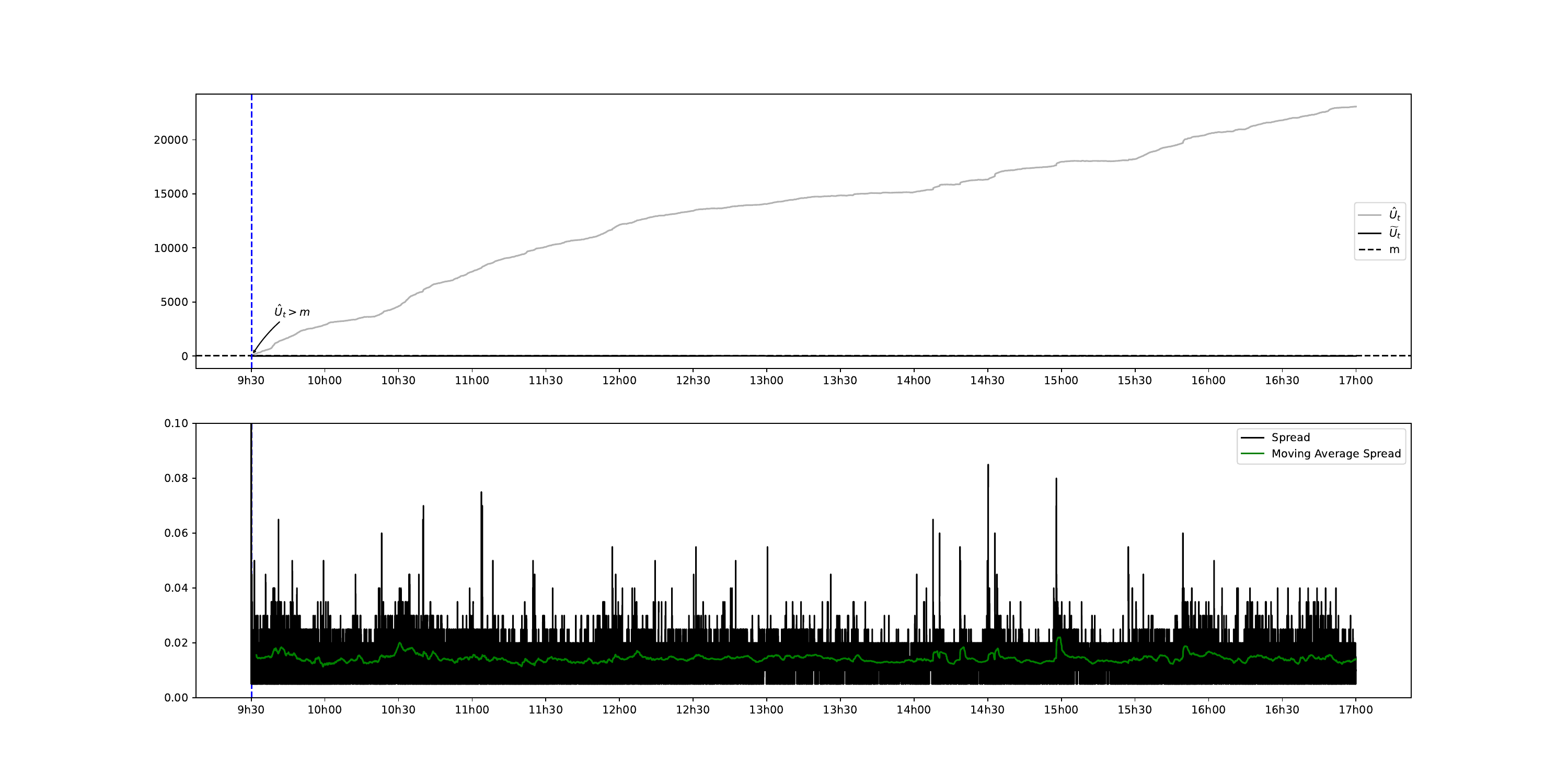}
       \caption{Performing hypothesis testing on the trade-through intensities between 31/08/2022 and 04/03/2022, with parameters $\rho_{up} = 1.5$, $\rho_{down} = 0.5$, and $m = 5$. The first figure (upside) presents the reflected processes $\widetilde{U}$ and $\hat{U}$. The second figure (downside) provides a comparison of the detection results against the spread.}
       \label{fig:sht_high_vol}
\end{figure}

The same comparison is undertaken only this time between the reference day 31/08/2022 and the day with the lowest volatility index 05/01/2022.
\begin{figure}[H]
    \centering
        \includegraphics[width=0.9\textwidth]{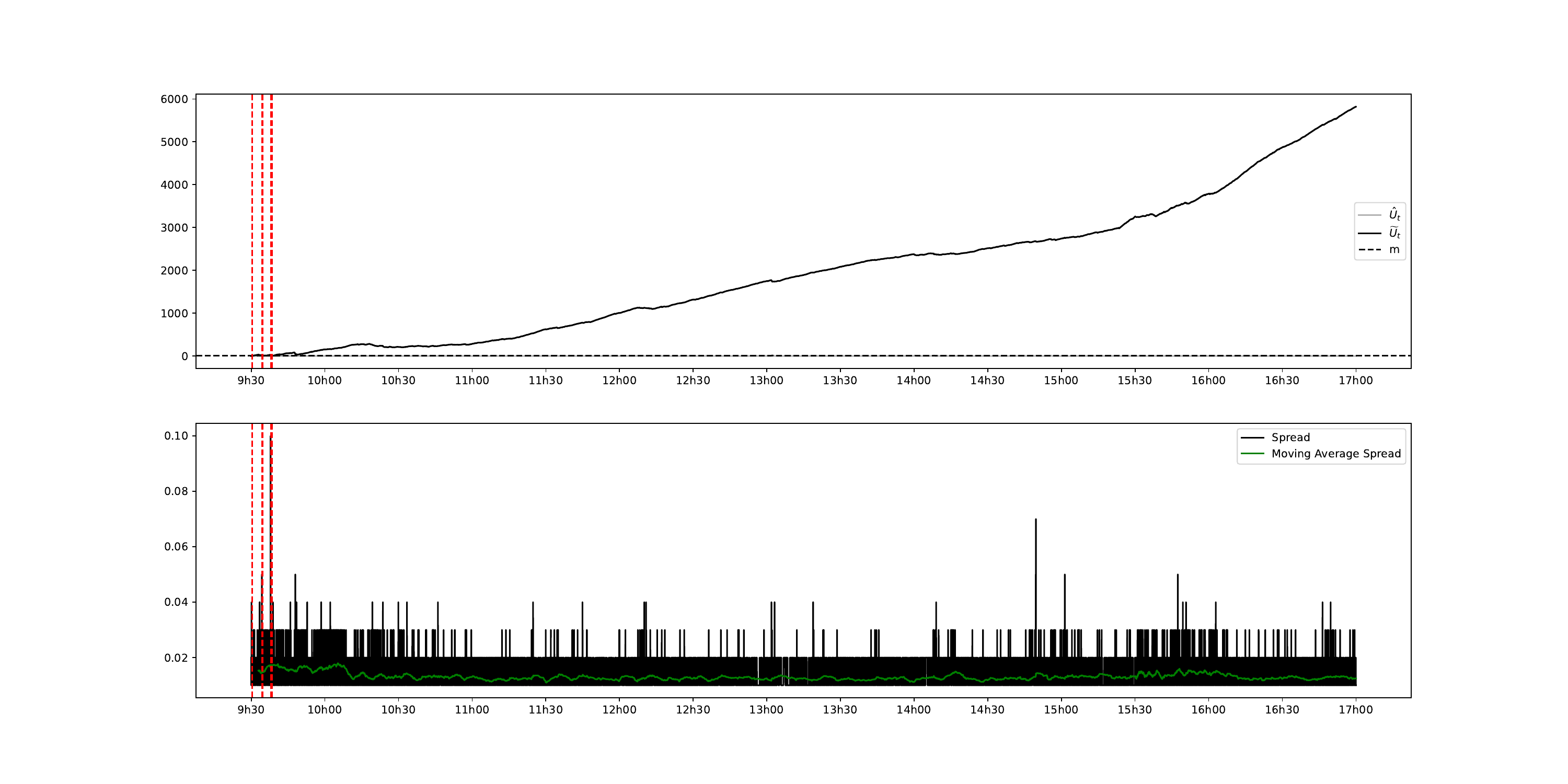}
       \caption{Performing hypothesis testing on the trade-through intensities between 31/08/2022 and 05/01/2022, with parameters $\rho_{up} = 1.5$, $\rho _{down} = 0.5$ and $m = 5$. The first figure (upside) presents the reflected processes $\widetilde{U}$ and $\hat{U}$. The second figure (downside) provides a comparison of the detection results against the spread.}
       \label{fig:sht_low_vol}
\end{figure}

It can be seen that comparing two days where the difference in volatility is quite large implies that the difference in liquidity between two days is detected quite quickly. However, this does not allow us to assess these variations from a more immediate point of view. Another approach would be to compare two relatively close days in order to monitor the evolution of the state of liquidity. This would allow us to better assess the variations in intra-day liquidity. A comparison between 10/05/2022 and 11/05/2022 is therefore suggested.
\begin{figure}[H]
    \centering
    \includegraphics[width=0.9\textwidth]{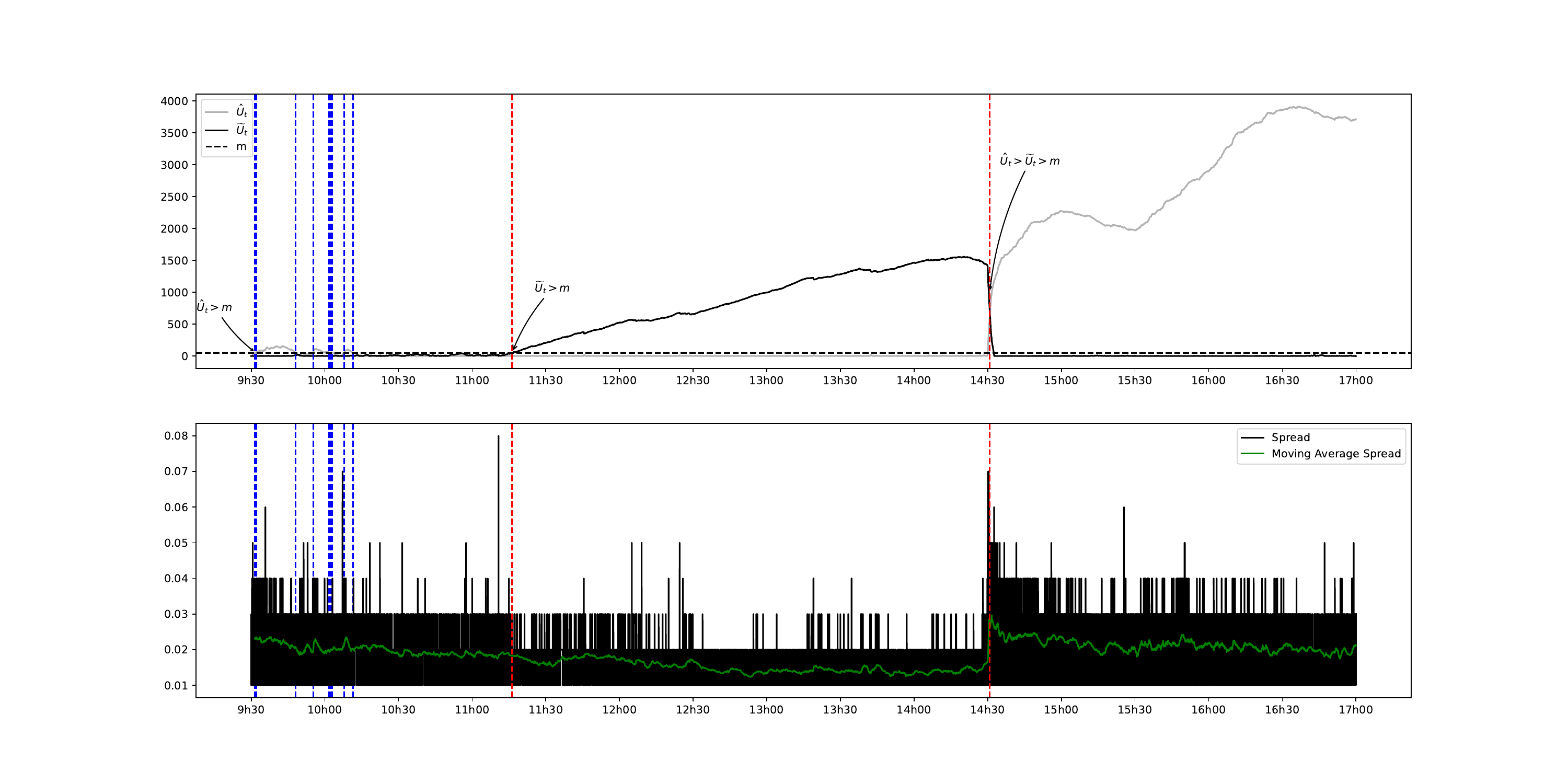}
       \caption{Performing hypothesis testing of the trade-through intensities between 10/05/2022 and 11/05/2022, with parameters $\rho_{up} = 1.5$, $\rho _{down} = 0.5$ and $m = 5$. The first figure (upside) presents the reflected processes $\widetilde{U}$ and $\hat{U}$. The second figure (downside) provides a comparison of the detection results against the spread.}
       \label{fig:may_10_11}
\end{figure}
It can be seen visually that there is a fairly consistent correlation between a low spread and a state where we are under the $\rho <$1 assumption and vice versa (i.e., a high intra-day liquidity). The quality of the detection can also be checked by directly inspecting the number of trades-through that have taken place during the day as shown in figure \ref{fig:cum_tt_may_11}.
\begin{figure}[H]
    \centering
        \includegraphics[width=0.7\textwidth]{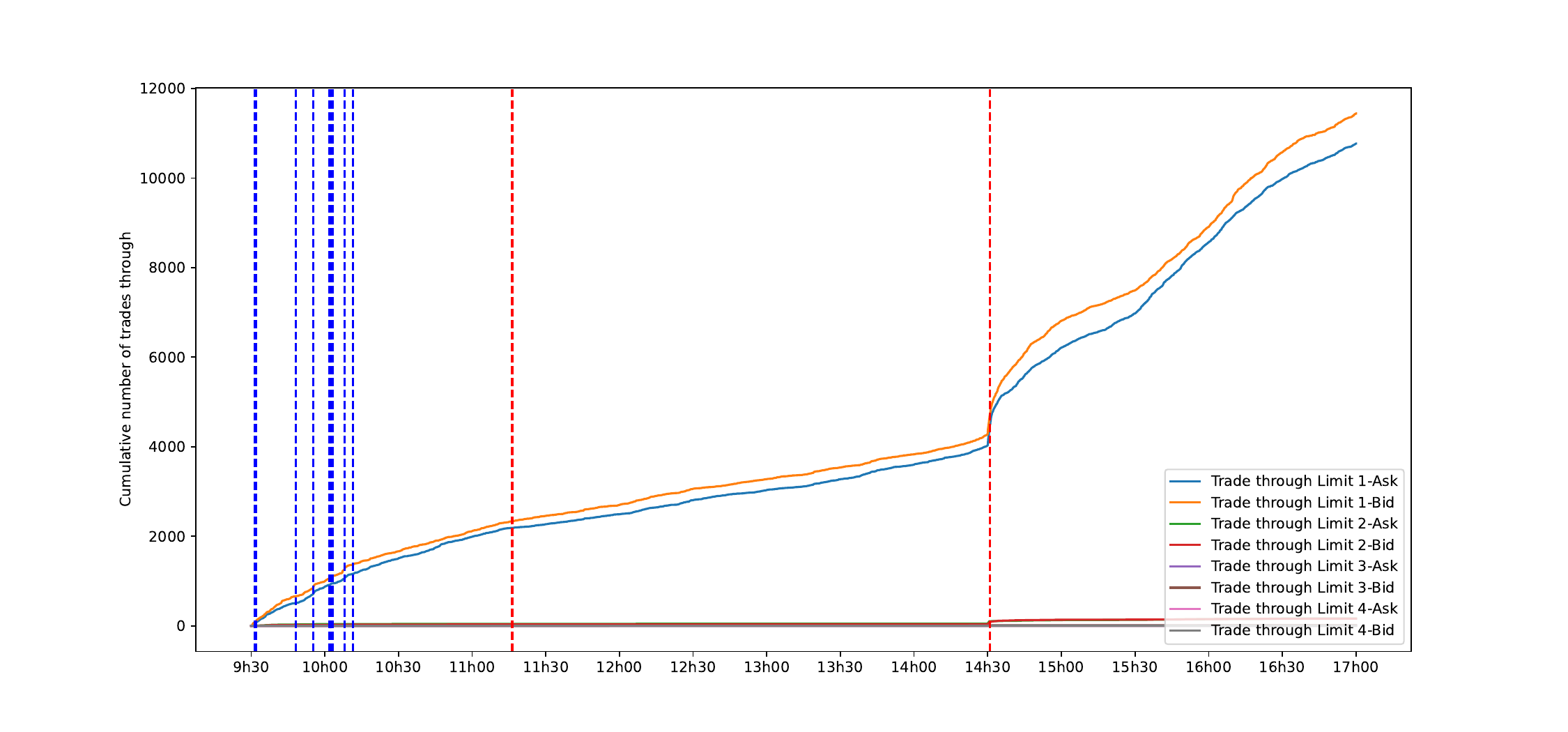}
       \caption{Number of cumulative trades-through on the first 4 limits during 11/05/2022.}
       \label{fig:cum_tt_may_11}
\end{figure}
We can clearly observe the fluctuations in the cumulative number of trades throughout 11:20 AM and 2:30 PM. The findings presented in Figure \ref{fig:cum_tt_may_11} thus validate the effectiveness of our detection method.
\begin{rque}
It is worth emphasizing that our methodology involves comparing intensities among counting processes. Consequently, selecting the appropriate reference day becomes imperative when conducting such comparisons. 
\end{rque}

\section{Conclusion and discussions}
In conclusion, this paper introduces a novel methodology to systematically quantify liquidity within a limit order book and detecting significant changes in its distribution. The use of Marked Hawkes Processes and a CUSUM-based optimal stopping scheme has led to the development of a reliable liquidity proxy with promising applications in financial markets.

The proposed metric holds great potential in enhancing optimal execution algorithms, reducing market impact, and serving as a reliable market liquidity indicator for market makers. The optimality results of the procedure, particularly in the case of a Cox process with simultaneous jumps and a finite time horizon, highlight the robustness of our approach. The tractable formulation of the average detection delay allows for convenient manipulation of the trade-off between the occurrence of false detections and the speed of identifying moments of liquidity disorder.

Moreover, the empirical validation using real market data further reinforces the effectiveness of our methodology, affirming its practicality and relevance in real-world trading scenarios. Future research can build upon these findings to explore its applications.

\printbibliography
\appendix
\section{Appendix}
\subsection{Proof of Section \ref{Sequential Change-Point Detection: CUSUM-based optimal stopping scheme} results}
In what follows, we will extend the results/proofs obtained in Moustakides \cite{Moustakides2008} and El Karoui et al. \cite{ElKaroui2017} to the case of simultaneous arrival times. The idea is to approximate $\sum_{i=1}^{D}N^i$ with a process (see definition \ref{def_transf1} for $\rho <1$ and definition \ref{def_transf2} for $\rho <1$) with jumps of size $1$, in order to extend these results to cases where there are simultaneous jumps. This will allow us to compute the Average Run Length (ARL) and explicit the lower bounds of the Expected Detection Delay (EDD) of the CUSUM procedure. Subsequently, we provide proof that the CUSUM procedure attains this lower bound, establishing its optimality.  
\begin{defi}
Let $\left(\tau^{i,\epsilon}_k\right)_{k \geq 1} = \left(\tau^i_k+\frac{i}{D}\epsilon\underset{1\leq l,l^{\prime} \leq D}{\inf}~~\underset{0 \leq \tau^l_j < \tau^{l^{\prime}}_{j^{\prime}} \leq \tau^i_{k}}{\inf}\mid \tau^l_j -\tau^{l^{\prime}}_{j^{\prime}}\mid\right)_{k \geq 1} = \left(\tau^i_k+\epsilon^i_{k}\right)_{k \geq 1}$ where $\tau^i_0 = 0$ and $\left(\tau^i_k\right)_{k \geq 1}$ are the event times of the process $N^i$ for all $i$ ranging from $1$ to $D$. In this context, we define $N^{i,\epsilon}$ as the counting process whose arrival times are $\left(\tau^{i,\epsilon}_k\right)_{k \geq 1}$, i.e, $N^{i,\epsilon}(t) = \sum_{k \geq 1}\mathbbm{1}_{\{\tau^{i,\epsilon}_k \leq t\}}$, $\forall 0\leq t\leq T$.
    ~We also define $\mathcal{F}^{\epsilon} = \left(\mathcal{F}^{\epsilon}_t\right)_{t\geq 0}$ as the natural filtration associated to the process $\sum_{i=1}^DN^{i,\epsilon}$.
 \label{def_transf1}
\end{defi}
\begin{rque}
\label{rque_de_transf1}
 Definition \ref{def_transf1} implies that each element $N^{i,\epsilon}$ is $\mathcal{F}$-adapted and that, $\forall 0\leq t \leq T$,
\[
\begin{split}
N^{i,\epsilon}(t) &= \sum_{k \geq 1}\mathbbm{1}_{\{\tau^{i,\epsilon}_k \leq t\}} \\ &\leq N^i\left(t - \frac{1}{D}\epsilon\underset{1\leq l,l^{\prime} \leq D}{\inf}~~\underset{0 \leq \tau^l_j < \tau^{l^{\prime}}_{j^{\prime}} \leq T}{\inf}\mid \tau^l_j -\tau^{l^{\prime}}_{j^{\prime}}\mid \right) \\ &\leq N^i(t),\quad \text{a.s}\end{split}\]
One can clearly observe that $N^i-N^{i,\epsilon}$ is an $\mathcal{F}$-sub-martingale. According to the Doob-Meyer decomposition, it can be inferred that $\Lambda^i-\Lambda^{i,\epsilon}$ is a predictable process that is non-decreasing and starts from zero. Consequently, $\forall 0\leq t \leq T$,
$$\Lambda^{i,\epsilon}(t)\leq \Lambda^i(t),\quad \text{a.s}$$
\end{rque}
\begin{nota}
In the subsequent discussion, we will refer to the process $t\mapsto\sup_{0\leq s\leq t}U(s)$ (resp. $t\mapsto\sup_{0\leq s\leq t}U^{\epsilon}(s)$) as $\overline{U}$ (resp. $\overline{U}^{\epsilon}$) where $U^{\epsilon}(t):=\sum_{i=1}^DN^{i,\epsilon}(t)-\beta(\rho) \Lambda^{i,\epsilon}(t)$, $\forall 0\leq t \leq T$. We will also use the notation $\widetilde{U}^{\epsilon}$ to represent $t\mapsto\widetilde{U}^{\epsilon}(t) := \sup _{0 \leq s \leq t} U^{\epsilon}(s) - U^{\epsilon}(t)$.
 \end{nota}
Figures \ref{N_Y_epsilon1} and \ref{N_Y_epsilon2} showcase a comparative simulation of the processes $\sum_{i=1}^DN^{i,\epsilon}$ and $\widetilde{U}^{\epsilon}$  with respect to $\sum_{i=1}^DN^i$ and $\widetilde{U}$.
\begin{figure}[H]
        \centering
    \begin{subfigure}[b]{0.49\textwidth}
        \centering
        \includegraphics[width=\textwidth]{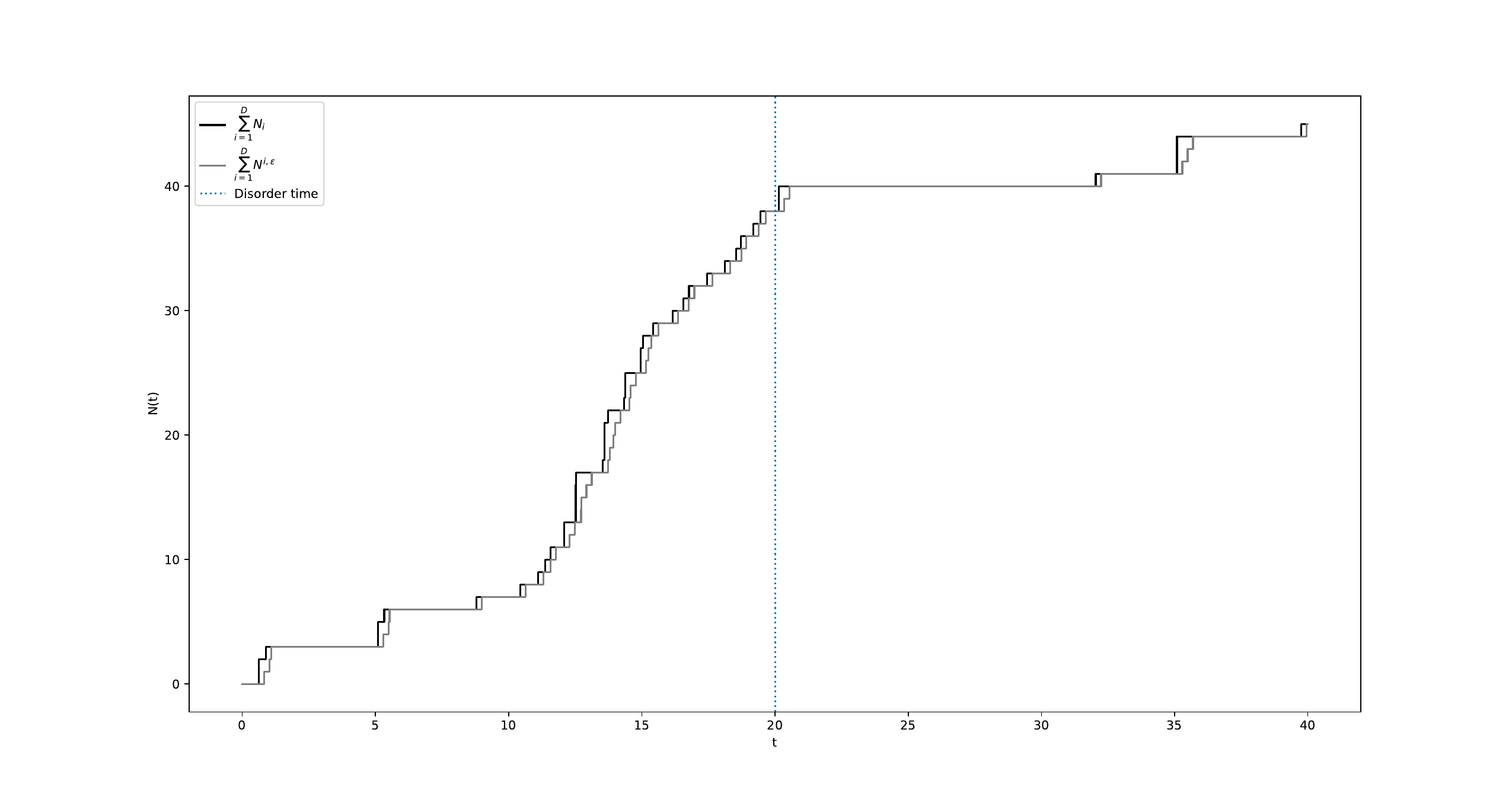}
        \caption{Simulation of $\sum_{i=1}^DN^{i,\epsilon}$}
        \label{N_Y_epsilon1}
    \end{subfigure}
    \begin{subfigure}[b]{0.49\textwidth}
        \centering
        \includegraphics[width=\textwidth]{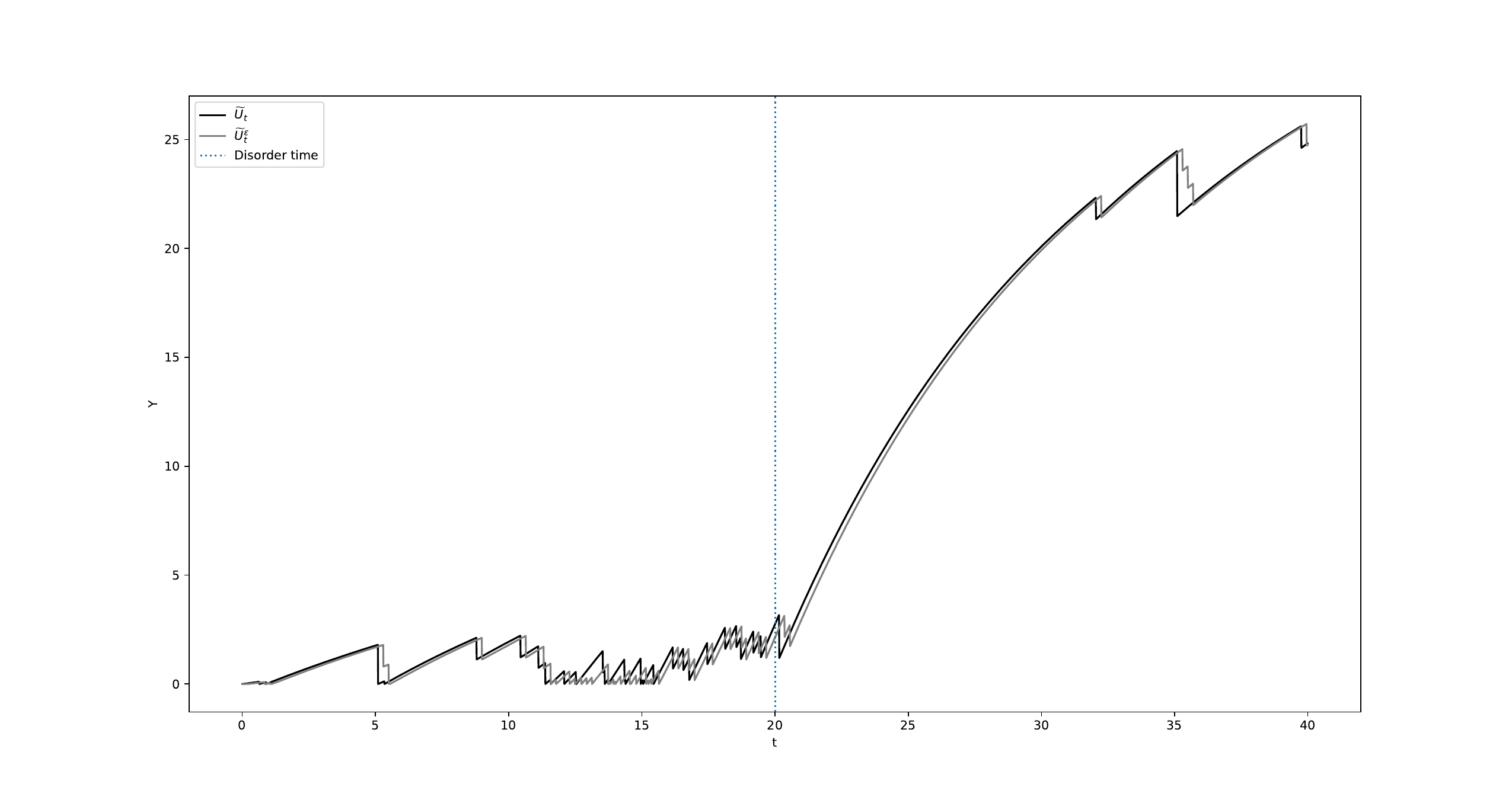}
        \caption{Simulation of $\widetilde{U}^{\epsilon}$}
        \label{N_Y_epsilon2}
    \end{subfigure}  
\end{figure}
\begin{rque}
     \label{ens_vide_conv}
         The limitations of definition \ref{def_transf1} in completely eliminating the simultaneous jumps found in $\sum_{i=1}^D N^i$ are noteworthy. Indeed, the set $A^{i,\epsilon} := \left\{\omega\in\Omega : \exists k\in\mathbb{N} / \tau_{k+1}^i(\omega)\leq\tau_k^{i,\epsilon}(\omega)\right\}$ may not be empty. This can be problematic as it can influence the order of arrival times in $\left(\tau^{i,\epsilon}_k\right)_{k \geq 1}$. However, it can be established through analysis that as $\epsilon$ tends to $0^+$, the limit of $A^{i,\epsilon}$ has a measure of zero. In fact,
     \begin{align*}
         \omega \in A^{i,\epsilon} &\Rightarrow \exists k\in\mathbb{N}: \tau_{k+1}^i(\omega)\leq\tau_k^{i,\epsilon}(\omega) \\&\Rightarrow \exists k\in\mathbb{N}: \tau_{k+1}^i(\omega)\leq \tau_k^i(\omega)+\epsilon T
     \end{align*}
     Consequently,
     $$A^{i,\epsilon}\subset \{\omega\in\Omega:0\leq k\leq N^i(\omega,T)-1/ \tau_{k+1}^i(\omega)-\tau_{k}^i(\omega)\leq \epsilon T\}$$
     This entails that,
     \begin{align*}
         \mathbb{P}\left(A^{i,\epsilon}\right)&\leq \mathbb{E}\left(\sum_{k\leq N^i(T)}\mathbb{P}\left(\tau_{k+1}^i-\tau_{k}^i\leq \epsilon T\mid N^i(T)\right)\right)\\&\leq \mathbb{E}\left(\sum_{k\leq N^i(T)}\left(\int_{\tau_{k}^i}^{\tau_{k}^i+\epsilon T}\lambda^i(s)\mathrm{d}s\right)e^{-\int_{\tau_{k}^i}^{\tau_{k}^i+\epsilon T}\lambda^i(s)\mathrm{d}s}\right)
     \end{align*}
      Once more, utilizing the theorem of dominated convergence yields the following conclusion and the fact that $N^i$ is almost surely finite, we can conclude that~:
     \begin{align*}
         \lim_{\epsilon \rightarrow 0^+}\mathbb{E}\left(\sum_{k\leq N^i(T)}\left(\int_{\tau_{k}^i}^{\tau_{k}^i+\epsilon T}\lambda^i(s)\mathrm{d}s\right)e^{-\int_{\tau_{k}^i}^{\tau_{k}^i+\epsilon T}\lambda^i(s)\mathrm{d}s}\right) &= \mathbb{E}\left(\sum_{k\leq N^i(T)}\lim_{\epsilon \rightarrow 0^+}\left(\int_{\tau_{k}^i}^{\tau_{k}^i+\epsilon T}\lambda^i(s)\mathrm{d}s\right)e^{-\int_{\tau_{k}^i}^{\tau_{k}^i+\epsilon T}\lambda^i(s)\mathrm{d}s}\right)\\&=0
     \end{align*}
     Hence, 
     \begin{equation}
         \lim_{\epsilon \rightarrow 0^+}\mathbb{P}\left(A^{i,\epsilon}\right) = 0
     \end{equation}
 \end{rque}
  We now introduce an intermediate result that will be useful later in the proof of Theorem \ref{ARL_rho_inf_1}.
\begin{lemma}
   For all $t\geq 0$ and $i\in\{ 1,\dots,D\}$,
   \begin{align*} 
   \lim_{\epsilon \rightarrow 0^+}\mathbb{E}\left(\mid N^i(t)-N^{i,\epsilon}(t)\mid\right) &=  \lim_{\epsilon \rightarrow 0^+}\mathbb{E}\left(\mid\Lambda^i(t)-\Lambda^{i,\epsilon}(t)\mid\right) \\&= 0\end{align*}
 \label{conv_transf1}
\end{lemma}
 \begin{proof}
 Let $0\leq t\leq T$, $i\in\{1,\dots,D\}$ and $\tau^i_{k_t}$ define the maximum arrival time of $N^i$ occurring before $t$. From Definition \ref{def_transf1}, it follows that, if $\tau_{k_t}^{i,\epsilon}<\tau_{k_t+1}^i$,
    \begin{align*}  
   N^i(t)-N^{i,\epsilon}(t) &= \begin{cases}
       0 \quad\text{if}\quad \tau_{k_t}^{i,\epsilon}\leq t <\tau_{k_t+1}^i\\
       1\quad\text{if}\quad \tau_{k_t}^i\leq t <\tau_{k_t}^{i,\epsilon}
   \end{cases} \end{align*}
   Hence, we get that
   \begin{align*}
       \mathbb{E}\left((N^i(t)-N^{i,\epsilon}(t))\mathbbm{1}_{A^{\epsilon,c}}\right) & \leq \mathbb{E}\left(\mathbbm{1}_{\{\tau_{k_t}^i\leq t <\tau_{k_t}^{i,\epsilon}\}}\right)\\&\leq\mathbb{P}\left(0\leq t- \tau_{k_t}^i\leq\epsilon_{k_t}^i\right)\\&\leq \mathbb{P}\left(0\leq t- \tau_{k_t}^i\leq\frac{i}{D}\epsilon T\right)
   \end{align*}
Based on the definition of a doubly stochastic Poisson process (see chapter 7 of \cite{Snyder1991}), we have~: \begin{align*}
       \mathbb{P}\left(0\leq t- \tau_{k_t}^i\leq\frac{i}{D}\epsilon T\right) &= \mathbb{E}\left(\mathbb{P}\left(t-\frac{i}{D}\epsilon T\leq\tau_{k_t}^i\leq t\mid \lambda^i(s): t-\frac{i}{D}\epsilon T\leq s\leq t \right)\right)\\& = \mathbb{E}\left(\frac{1}{k_t!}\left(\int_{t-\frac{i}{D}\epsilon T}^t\lambda^i(s)\mathrm{d}s\right)^{k_t}e^{-\int_{t-\frac{i}{D}\epsilon T}^t\lambda^i(s)\mathrm{d}s}\right)
   \end{align*} 
   Consequently, by employing the theorem of dominated convergence, it follows that~: 
   $$\lim_{\epsilon \rightarrow 0^+}\mathbb{P}\left(0\leq t- \tau_{k_t}^i\leq\frac{i}{D}\epsilon T\right) = 0$$
   and that~:
   \begin{equation}
       \lim_{\epsilon \rightarrow 0^+}\mathbb{E}\left((N^i(t)-N^{i,\epsilon}(t))\mathbbm{1}_{A^{\epsilon,c}}\right) = 0   \end{equation}
   Since $N^{i,\epsilon}\leq N^i$ (see Remark \ref{rque_de_transf1}) and using the fact that $\lim_{\epsilon \rightarrow 0^+}\mathbb{P}\left(A^{i,\epsilon}\right) = 0$ (see Remark \ref{ens_vide_conv}), we get that $\mathbb{E}\left(\mid N^i(t)-N^{i,\epsilon}(t)\mid\right) = \mathbb{E}\left(N^i(t)-N^{i,\epsilon}(t)\right)$ and that~:
   \begin{align*}  
   \lim_{\epsilon \rightarrow 0^+}\mathbb{E}\left(\mid N^i(t)-N^{i,\epsilon}(t)\mid \right)&=\lim_{\epsilon \rightarrow 0^+}\mathbb{E}\left((N^i(t)-N^{i,\epsilon}(t))\mathbbm{1}_{A^{\epsilon}}\right) + \mathbb{E}\left((N^i(t)-N^{i,\epsilon}(t))\mathbbm{1}_{A^{\epsilon,c}}\right)\\&= 0\end{align*}
   Hence,
   \begin{align*} 
   \lim_{\epsilon \rightarrow 0^+}\mathbb{E}\left(\mid N^i(t)-N^{i,\epsilon}(t)\mid\right) &=  \lim_{\epsilon \rightarrow 0^+}\mathbb{E}\left(\mid\Lambda^i(t)-\Lambda^{i,\epsilon}(t)\mid\right) \\&= 0\end{align*}
 \end{proof}
 
\begin{lemma}
    Let $t \in [0, T]$ and $\widetilde{U}$ (resp. $\widetilde{U}^{\epsilon}$) the CUSUM process of $\sum_{i=1}^D N^i$ (resp. $\sum_{i=1}^D N^{i,\epsilon}$) for $\rho <1$.  Then, $$\lim_{\epsilon \rightarrow 0^+} \mathbb{E}_y\left(\mid \widetilde{U}^{\epsilon}(t) - \widetilde{U}(t)\mid\right) = 0$$
    with $\mathbb{E}_y = \mathbb{E}(.\mid \widetilde{U}_0 = y)$.
    \label{y_borne}
\end{lemma}
\begin{proof}
Let $(\tau_k)_{k \geq 1}$ denote the distinct and ordered jump times of $\sum_{i=1}^D N^i$ and $\tau^{\epsilon}_k$ define the maximum arrival time of $\sum_{i=1}^D N^{i,\epsilon}$ occurring before $\tau_{k+1}$ for all $k\in \mathbb{N}^*$. The goal here is to demonstrate the convergence $\mathrm{L}^1$ by bounding $\mid \widetilde{U}^{\epsilon} - \widetilde{U}\mid$ by a process that converges to 0 under the $\mathrm{L}^1$ norm. One approach is to evaluate the difference between the processes $\overline{U}$ and $\overline{U}^{\epsilon}$ at the arrival times of $\tau_k$ and $\tau_k^{\epsilon}$, as $\widetilde{U}$ and $\widetilde{U}^{\epsilon}$ are continuous on the intervals $]\tau_k,\tau_{k}^{\epsilon}[$ and $]\tau_{k}^{\epsilon},\tau_{k+1}[$ and only exhibit jumps at those times. To initiate the proof, we start by employing an inductive argument to establish boundedness. As discussed in Remark \ref{ens_vide_conv}, the Definition \ref{def_transf1} is not completely satisfactory. In order to address the raised concerns, the inequalities established in the subsequent induction hold for events within the set $A^{\epsilon,c} := \Omega \backslash \bigcup_{i=1}^D A^{i,\epsilon}$. Subsequently, we will utilize the aforementioned result in Remark \ref{ens_vide_conv} to facilitate the passage to the limit.\\
 For $k=1$, since $\sum_{i=1}^DN^i$ (resp. $\sum_{i=1}^DN^{i,\epsilon}$) jumps at the moment $\tau_1$ (resp. $\tau^{\epsilon}_1$), two cases can arise~:
 \begin{mycase}
     \case $\sup_{0 \leq u \leq \tau_1}\sum_{i=1}^DN^i(u) - \beta(\rho) \Lambda^i(u) > \sum_{i=1}^DN^i(\tau_1) - \beta(\rho) \Lambda^i(\tau_1)$.\\
     This means that $\sum_{i=1}^D\beta(\rho) \Lambda^i(\tau_1) > \sum_{i=1}^DN^i(\tau_1)$. Since $\sum_{i=1}^DN^{i,\epsilon}(\tau_1) = 0$ by definition, it follows that $\overline{U}(\tau_{1}) = \overline{U}^{\epsilon}(\tau_{1}) = 0$ and that~:
     $$\mid\overline{U}(\tau_{1})-\overline{U}^{\epsilon}(\tau_{1})\mid = 0$$
     Moreover, \begin{equation*}
\begin{split}\mid\overline{U}(\tau^{\epsilon}_{1})-\overline{U}^{\epsilon}(\tau^{\epsilon}_{1})\mid &= \mid \sup_{u \leq \tau^{\epsilon}_{1}}\sum_{i=1}^DN^i(u) - \beta(\rho) \Lambda^i(u) - \sup_{u \leq \tau^{\epsilon}_{1}}\sum_{i=1}^DN^{i,\epsilon}(u) - \beta(\rho) \Lambda^{i,\epsilon}(u)\mid \\&= \begin{cases} \sum_{i=1}^DN^{i,\epsilon}(\tau^{\epsilon}_1) - \beta(\rho) \Lambda^{i,\epsilon}(\tau^{\epsilon}_1) \quad\text{if}~ \widetilde{U}^{\epsilon}(\tau^{\epsilon}_1) = 0\\ 0~\text{else}\end{cases}\\ &\leq \begin{cases} \beta(\rho)\sum_{i=1}^D\Lambda^i(\tau^{\epsilon}_1) -\Lambda^{i,\epsilon}(\tau^{\epsilon}_1)  \quad\text{if}~ \widetilde{U}^{\epsilon}(\tau^{\epsilon}_1) = 0\\ 0~\text{else}\end{cases}\end{split}
\end{equation*}
\case $\sup_{0 \leq u \leq \tau_1}\sum_{i=1}^DN^i(u) - \beta(\rho) \Lambda^i(u) = \sum_{i=1}^DN^i(\tau_1) - \beta(\rho) \Lambda^i(\tau_1)$.\\
The evaluation of the difference in $\tau_1$ is straightforward~:
\begin{equation*}
    \begin{split}\mid\overline{U}(\tau_{1})-\overline{U}^{\epsilon}(\tau_{1})\mid &= \sum_{i=1}^DN^i(\tau_1) - \beta(\rho) \Lambda^i(\tau_1)\\&\leq \sum_{i=1}^DN^i(\tau_1)\end{split}
\end{equation*}
In addition, if $\widetilde{U}^{\epsilon}(\tau^{\epsilon}_1)>0$, it follows that $\sum_{i=1}^D\beta(\rho) \Lambda^{i,\epsilon}(\tau^{\epsilon}_1) > \sum_{i=1}^DN^{i,\epsilon}(\tau^{\epsilon}_1)$. Given that $\sum_{i=1}^DN^i(\tau_1)  = \sum_{i=1}^DN^{i,\epsilon}(\tau^{\epsilon}_1)$, we can deduce that $\sum_{i=1}^D\beta(\rho) \Lambda^{i,\epsilon}(\tau^{\epsilon}_1) > \sum_{i=1}^DN^i(\tau_1)$ in this case. Thus~:
\begin{equation*}
\begin{split}\mid\overline{U}(\tau^{\epsilon}_{1})-\overline{U}^{\epsilon}(\tau^{\epsilon}_{1})\mid &= \mid \sup_{u \leq \tau^{\epsilon}_{1}}\sum_{i=1}^DN^i(u) - \beta(\rho) \Lambda^i(u) - \sup_{u \leq \tau^{\epsilon}_{1}}\sum_{i=1}^DN^{i,\epsilon}(u) - \beta(\rho) \Lambda^{i,\epsilon}(u)\mid \\&= \begin{cases} \beta(\rho)\sum_{i=1}^D\mid\Lambda^i(\tau_1) - \Lambda^{i,\epsilon}(\tau^{\epsilon}_1)\mid \quad\quad~\text{if}~ \widetilde{U}^{\epsilon}(\tau^{\epsilon}_1) = 0\\ \sum_{i=1}^DN^i(\tau_1) - \beta(\rho) \Lambda^i(\tau_1)\quad\text{else}\end{cases}\\ &\leq \begin{cases} \beta(\rho)\sum_{i=1}^D\mid\Lambda^i(\tau_1) - \Lambda^{i,\epsilon}(\tau^{\epsilon}_1)\mid \quad\quad~\text{if}~ \widetilde{U}^{\epsilon}(\tau^{\epsilon}_1) = 0\\ \beta(\rho)\sum_{i=1}^D\Lambda^{i,\epsilon}(\tau^{\epsilon}_1) - \Lambda^i(\tau_1)~\text{else}\end{cases}\end{split}
\end{equation*}
 \end{mycase}

Consequently, we have the following bounds for $k=1$~: \begin{equation}
    \begin{split}\mid\overline{U}(\tau^{\epsilon}_{1})-\overline{U}^{\epsilon}(\tau^{\epsilon}_{1})\mid &\leq \beta(\rho)\sum_{i=1}^D\mid\Lambda^{i,\epsilon}(\tau^{\epsilon}_1) - \Lambda^i(\tau_1)\mid +\mid\Lambda^{i,\epsilon}(\tau^{\epsilon}_1) - \Lambda^i(\tau^{\epsilon}_1)\mid 
    \end{split}
\end{equation}   
\begin{equation}
    \begin{split}\mid\overline{U}(\tau_{1})-\overline{U}^{\epsilon}(\tau_{1})\mid 
    &= \sum_{i=1}^DN^i(\tau_1) -  N^{i,\epsilon}(\tau_1)
    \end{split}
\end{equation}
  Let $k \geq 1$, such that $0\leq \tau_k \leq T$. Assume that: \begin{equation}
      \begin{split}
          \mid\overline{U}(\tau^{\epsilon}_{k})-\overline{U}^{\epsilon}(\tau^{\epsilon}_{k})\mid &\leq \beta(\rho)\sum_{j \leq k}\sum_{i=1}^D\mid\Lambda^{i,\epsilon}(\tau^{\epsilon}_j) - \Lambda^i(\tau_j)\mid +\mid\Lambda^{i,\epsilon}(\tau^{\epsilon}_j) - \Lambda^i(\tau^{\epsilon}_j)\mid \\&=\mathcal{L}^{\epsilon}_k 
          \label{induction3}
      \end{split}
\end{equation}  
\begin{equation}
      \begin{split}\mid\overline{U}(\tau_{k})-\overline{U}^{\epsilon}(\tau_{k})\mid &\leq \sum_{j \leq k}\sum_{i=1}^D N^i(\tau_j) -  N^{i,\epsilon}(\tau_j)+\mathcal{L}^{\epsilon}_{k-1}\\&=\mathcal{M}^{\epsilon}_k \end{split}
      \label{induction4}
\end{equation} 
where $\mathcal{L}^{\epsilon}_0 = 0$.\\
We will now demonstrate the validity of these bounds for $k+1$. The proof follows a similar approach to what we have previously seen for the case of $k=1$. We can categorize the analysis into two scenarios, depending on whether $U$ attains a supremum at $\tau_{k+1}$ or not.
 \begin{mycase} 
 \case $\sup_{u \leq \tau_{k+1}}\sum_{i=1}^DN^i(u) - \beta(\rho) \Lambda^i(u) > \sum_{i=1}^DN^i(\tau_{k+1}) - \beta(\rho) \Lambda^i(\tau_{k+1})$.\\
 In this situation, the maximum value in $U$ is attained prior to the $(k+1)$-th jump. Therefore, there exists $k_0\leq k$ such that $\overline{U}(\tau_{k+1}) =U(\tau_{k_0})$. In case $\widetilde{U}^{\epsilon}(\tau^{\epsilon}_{k+1}) = 0$, we have that~:\\  \begin{equation*}
     \begin{split}
         \overline{U}^{\epsilon}(\tau^{\epsilon}_{k+1}) &\geq \sum_{i=1}^D N^{i,\epsilon}(\tau^{\epsilon}_{k_0}) - \beta(\rho) \Lambda^{i,\epsilon}(\tau^{\epsilon}_{k_0}) 
         \\&= \overline{U}(\tau^{\epsilon}_{k+1})+ \beta(\rho)\sum_{i=1}^D\Lambda^i(\tau_{k_0})-  \Lambda^{i,\epsilon}(\tau^{\epsilon}_{k_0})\end{split}
\end{equation*}
\begin{equation}\overline{U}^{\epsilon}(\tau^{\epsilon}_{k+1})-\overline{U}(\tau^{\epsilon}_{k+1}) \geq \beta(\rho)\sum_{i=1}^D\Lambda^{i,\epsilon}(\tau^{\epsilon}_{k_0})-\Lambda^i(\tau_{k_0})
\label{induction1}\end{equation}
 Moreover,\begin{equation*}
\begin{split}\overline{U}^{\epsilon}(\tau^{\epsilon}_{k+1}) &= \sum_{i=1}^D N^{i,\epsilon}(\tau^{\epsilon}_{k+1}) - \beta(\rho) \Lambda^{i,\epsilon}(\tau^{\epsilon}_{k+1}) \\&= \sum_{i=1}^D N^i(\tau_{k+1}) -\beta(\rho) \Lambda^i(\tau_{k+1})+\beta(\rho) \Lambda^i(\tau_{k+1})- \beta(\rho) \Lambda^{i,\epsilon}(\tau^{\epsilon}_{k+1})
\\&\leq \overline{U}(\tau^{\epsilon}_{k+1}) + \beta(\rho)\sum_{i=1}^D \Lambda^i(\tau_{k+1}) -\Lambda^{i,\epsilon}(\tau^{\epsilon}_{k+1})\end{split}
\end{equation*}
i.e,\begin{equation}\overline{U}^{\epsilon}(\tau^{\epsilon}_{k+1}) - \overline{U}(\tau^{\epsilon}_{k+1})\leq \beta(\rho)\sum_{i=1}^D \Lambda^i(\tau_{k+1}) -\Lambda^{i,\epsilon}(\tau^{\epsilon}_{k+1})\label{induction2}\end{equation}
 Therefore, according to equations \ref{induction1} and \ref{induction2},
 \begin{equation*}
\begin{split}\mid\overline{U}(\tau^{\epsilon}_{k+1})-\overline{U}^{\epsilon}(\tau^{\epsilon}_{k+1})\mid &\leq \begin{cases} \beta(\rho)\sum_{i=1}^D\mid\Lambda^{i,\epsilon}(\tau^{\epsilon}_{k_0})-\Lambda^i(\tau_{k_0})\mid+\mid\Lambda^i(\tau_{k+1}) -\Lambda^{i,\epsilon}(\tau^{\epsilon}_{k+1})\mid~\text{if}~ \widetilde{U}^{\epsilon}(\tau^{\epsilon}_{k+1}) = 0\\ \mid \overline{U}(\tau^{\epsilon}_{k}) - \overline{U}^{\epsilon}(\tau^{\epsilon}_{k})\mid\quad\text{else}\end{cases}\\ &\leq \begin{cases} \mathcal{L}^{\epsilon}_{k+1}\quad~\text{if}~ \widetilde{U}^{\epsilon}(\tau^{\epsilon}_{k+1}) = 0\\ \mathcal{L}^{\epsilon}_k \quad\text{else}\end{cases}\\&\leq \mathcal{L}^{\epsilon}_{k+1}\end{split}
\end{equation*}
Furthermore, by making use of assumption \ref{induction3}, \begin{equation*}
    \begin{split}\mid\overline{U}(\tau_{k+1})-\overline{U}^{\epsilon}(\tau_{k+1})\mid &= \mid\overline{U}(\tau^{\epsilon}_{k})-\overline{U}^{\epsilon}(\tau^{\epsilon}_{k})\mid \\&\leq \mathcal{L}^{\epsilon}_k\\&\leq\mathcal{M}^{\epsilon}_{k+1}\end{split}
\end{equation*}
Thus, the induction has been established for the first case.
\case $\sup_{0 \leq u \leq \tau_{k+1}}\sum_{i=1}^DN^i(u) - \beta(\rho) \Lambda^i(u) = \sum_{i=1}^DN^i(\tau_{k+1}) - \beta(\rho) \Lambda^i(\tau_{k+1})$. \\
 Suppose that $\widetilde{U}^{\epsilon}(\tau^{\epsilon}_{k+1}) > 0$, then there exists $k_0\leq k$ such that $\overline{U}^{\epsilon}(\tau^{\epsilon}_{k+1}) =U^{\epsilon}(\tau^{\epsilon}_{k_0})$ and~:\\  \begin{equation*}
\begin{split}\overline{U}(\tau^{\epsilon}_{k+1}) &\geq \sum_{i=1}^D N^i(\tau_{k_0}) - \beta(\rho) \Lambda^i(\tau_{k_0}) \\&= \sum_{i=1}^DN^{i,\epsilon}(\tau^{\epsilon}_{k_0}) - \beta(\rho) \Lambda^{i,\epsilon}(\tau^{\epsilon}_{k_0}) + \beta(\rho) \Lambda^{i,\epsilon}(\tau^{\epsilon}_{k_0})- \beta(\rho) \Lambda^i(\tau_{k_0})\\&= \overline{U}^{\epsilon}(\tau^{\epsilon}_{k+1})+ \beta(\rho)\sum_{i=1}^D  \Lambda^{i,\epsilon}(\tau^{\epsilon}_{k_0}) - \Lambda^i(\tau_{k_0})\end{split}\end{equation*}
 i.e, \begin{equation}\overline{U}(\tau^{\epsilon}_{k+1})-\overline{U}^{\epsilon}(\tau^{\epsilon}_{k+1}) \geq \beta(\rho)\sum_{i=1}^D\Lambda^i(\tau_{k_0}) - \Lambda^{i,\epsilon}(\tau^{\epsilon}_{k_0})\label{induction5}\end{equation}
 In addition,\begin{equation*}
\begin{split}\overline{U}(\tau^{\epsilon}_{k+1}) &= \sum_{i=1}^D N^i(\tau^{\epsilon}_{k+1}) - \beta(\rho) \Lambda^i(\tau^{\epsilon}_{k+1}) \\&= \sum_{i=1}^D N^{i,\epsilon}(\tau^{\epsilon}_{k+1}) -\beta(\rho) \Lambda^{i,\epsilon}(\tau^{\epsilon}_{k+1})+\beta(\rho) \Lambda^{i,\epsilon}(\tau^{\epsilon}_{k+1})- \beta(\rho) \Lambda^i(\tau^{\epsilon}_{k+1})\\ &\leq \overline{U}^{\epsilon}(\tau^{\epsilon}_{k+1}) + \beta(\rho)\sum_{i=1}^D \Lambda^{i,\epsilon}(\tau^{\epsilon}_{k+1}) -\Lambda^i(\tau^{\epsilon}_{k+1})\end{split}
\end{equation*}
i.e,\begin{equation}
    \overline{U}(\tau^{\epsilon}_{k+1}) - \overline{U}^{\epsilon}(\tau^{\epsilon}_{k+1}) \leq \beta(\rho)\sum_{i=1}^D \Lambda^{i,\epsilon}(\tau^{\epsilon}_{k+1}) -\Lambda^i(\tau^{\epsilon}_{k+1})\label{induction6}\end{equation}
 Therefore, based on equations \ref{induction5} and \ref{induction6},
 \begin{equation*}
\begin{split}\mid\overline{U}(\tau^{\epsilon}_{k+1})-\overline{U}^{\epsilon}(\tau^{\epsilon}_{k+1})\mid &= \begin{cases} \mid\sum_{i=1}^DN^i(\tau_{k+1}) - N^{i,\epsilon}(\tau^{\epsilon}_{k+1}) -\beta(\rho) \Lambda^i(\tau_{k+1}) +\beta(\rho) \Lambda^{i,\epsilon}(\tau^{\epsilon}_{k+1})\mid ~\text{if}~ \widetilde{U}^{\epsilon}(\tau^{\epsilon}_{k+1}) = 0\\ \beta(\rho)\sum_{i=1}^D \mid\Lambda^{i,\epsilon}(\tau^{\epsilon}_{k+1}) -\Lambda^i(\tau^{\epsilon}_{k+1})\mid+\mid\Lambda^i(\tau_{k_0}) - \Lambda^{i,\epsilon}(\tau^{\epsilon}_{k_0})\mid\quad\text{else}\end{cases}\\ &\leq \begin{cases} \beta(\rho)\sum_{i=1}^D\mid\Lambda^{i,\epsilon}(\tau^{\epsilon}_{k+1}) -  \Lambda^i(\tau_{k+1})\mid \quad\quad~\text{if}~ \widetilde{U}^{\epsilon}(\tau^{\epsilon}_{k+1}) = 0\\ \mathcal{L}^{\epsilon}_{k+1}\quad\text{else}\end{cases}\\&\leq\mathcal{L}^{\epsilon}_{k+1}\end{split}
\end{equation*}
Moreover, \begin{equation*}
\begin{split}\overline{U}(\tau_{k+1})-\overline{U}^{\epsilon}(\tau_{k+1}) &\leq \sum_{i=1}^DN^i(\tau_{k+1}) -\beta(\rho) \Lambda^i(\tau_{k+1}) -N^{i,\epsilon}(\tau_{k+1}) +\beta(\rho) \Lambda^{i,\epsilon}(\tau_{k+1})\\&\leq \sum_{i=1}^DN^i(\tau_{k+1}) -N^{i,\epsilon}(\tau_{k+1})\end{split}
\end{equation*}
and \begin{equation*}
\begin{split}\overline{U}^{\epsilon}(\tau_{k+1}) - \overline{U}(\tau_{k+1})&\leq \sum_{i=1}^DN^{i,\epsilon}(\tau^{\epsilon}_{k_0}) -\beta(\rho) \Lambda^{i,\epsilon}(\tau^{\epsilon}_{k_0}) -N^i(\tau^{\epsilon}_{k_0}) +\beta(\rho) \Lambda^i(\tau^{\epsilon}_{k_0})\\&\leq \beta(\rho)\sum_{i=1}^D  \Lambda^i(\tau^{\epsilon}_{k_0})- \Lambda^{i,\epsilon}(\tau^{\epsilon}_{k_0})\end{split}
\end{equation*}
Thus,
\begin{equation*}
    \begin{split}\mid\overline{U}(\tau_{k+1}) - \overline{U}^{\epsilon}(\tau_{k+1})\mid&\leq \sum_{i=1}^DN^i(\tau_{k+1}) -N^{i,\epsilon}(\tau_{k+1}) + \beta(\rho)\sum_{i=1}^D  \Lambda^i(\tau^{\epsilon}_{k_0})- \Lambda^{i,\epsilon}(\tau^{\epsilon}_{k_0})\\ &\leq \mathcal{M}^{\epsilon}_{k+1}
    \end{split}
    \end{equation*}
which completes the recurrence.
\end{mycase}
We proceed to complete the proof by employing a convergence argument. Let $t \in [0, T]$ and $\tau_{k_t}$ be the largest arrival time of $\sum_{i=1}^D N^i(t)+ N^{i,\epsilon}(t)$ less than or equal to t. Therefore, we have,
\begin{equation*}
\begin{split}\mid \widetilde{U}^{\epsilon}(t) - \widetilde{U}(t)\mid &= \mid\overline{U}(\tau_{k_t})-\overline{U}^{\epsilon}(\tau_{k_t}) + \sum_{i=1}^DN^{i,\epsilon}(\tau_{k_t}) - N^i(\tau_{k_t}) -\beta(\rho) \Lambda^{i,\epsilon}(t) + \beta(\rho) \Lambda^i(t)\mid\\&\leq \mathcal{L}^{\epsilon}_{k_t} + \mathcal{M}^{\epsilon}_{k_t} +\sum_{i=1}^DN^{i,\epsilon}(\tau_{k_t}) - N^i(\tau_{k_t}) -\beta(\rho) \Lambda^{i,\epsilon}(t) + \beta(\rho) \Lambda^i(t) \end{split}
\end{equation*}

In addition, according to Lemma \ref{conv_transf1}, $$\lim_{\epsilon \rightarrow 0^+}\mathbb{E}_y\left(\mid\sum_{i=1}^DN^i(t)-N^{i,\epsilon}(t)\mid\mathbbm{1}_{A^{\epsilon,c}}\right) = \lim_{\epsilon \rightarrow 0^+}\mathbb{E}_y\left(\mid\sum_{i=1}^D\Lambda^i(t)-\Lambda^{i,\epsilon}(t)\mid\mathbbm{1}_{A^{\epsilon,c}}\right) = 0$$ Since $\sum_{i=1}^DN^i(t)<+\infty$ (i.e, $\mathbb{E}_y\left(card\left(\{k\in \mathbb{N};\tau_k\leq t\}\right)\right) < +\infty$), we conclude that $$\lim_{\epsilon \rightarrow 0^+}\mathbb{E}_y\left(\left(\mathcal{L}^{\epsilon}_{k_t} + \mathcal{M}^{\epsilon}_{k_t} +\sum_{i=1}^DN^{i,\epsilon}(\tau_{k_t}) - N^i(\tau_{k_t}) -\beta(\rho) \Lambda^{i,\epsilon}(t) + \beta(\rho) \Lambda^i(t)\right)\mathbbm{1}_{A^{\epsilon,c}}\right) = 0$$ and that$$\lim_{\epsilon \rightarrow 0^+} \mathbb{E}_y\left(\mid \widetilde{U}^{\epsilon}(t) - \widetilde{U}(t)\mid\mathbbm{1}_{A^{\epsilon,c}}\right) = 0$$
As stated in Remark \ref{ens_vide_conv}, $\lim_{\epsilon \rightarrow 0^+}\mathbb{P}\left(\bigcup_{i=1}^D A^{i,\epsilon}\right) = 0$. 
By utilizing the Cauchy Schwartz inequality, we get that $\lim_{\epsilon \rightarrow 0^+}\mathbb{E}_y\left(\mid\widetilde{U}^{\epsilon}(t)-\widetilde{U}(t)\mid\mathbbm{1}_{A^\epsilon}\right) = 0$. Hence,  
\begin{align*}
    \lim_{\epsilon \rightarrow 0^+}\mathbb{E}_y\left(\mid\widetilde{U}^{\epsilon}(t)-\widetilde{U}(t)\mid\right) &= \lim_{\epsilon \rightarrow 0^+}\mathbb{E}_y\left(\mid\widetilde{U}^{\epsilon}(t)-\widetilde{U}(t)\mid\mathbbm{1}_{A^\epsilon}\right)+\mathbb{E}_y\left(\mid\widetilde{U}^{\epsilon}(t)-\widetilde{U}(t)\mid\mathbbm{1}_{A^{\epsilon,c}}\right)\\&=0
\end{align*}
\end{proof}
We now introduce the functions $g_m$ and $h_m$ (resp. $\tilde{g}_m$ and $\tilde{h}_m$ for $\Tilde{\rho}=\frac{1}{\rho}$) solutions of the delayed differential equations define in \ref{dde_gm} and \ref{dde_hm}. This will allow us to compute the average delay of the disorder detection\footnote{Often referred to as the Average Run Length (ARL).} $\mathbb{E}\left(\sum_{i=1}^DN^i({\widetilde{T}_{\mathrm{C}}})\right)$ as a function of the performance/threshold criterion $m$ that defines the CUSUM statistic and to determine the Expected Detection Delay $\mathbb{E}^{\theta}\left[\left(\sum_{i=1}^DN^i(T)-N^i(\theta)\right)^{+} \mid \mathcal{F}_\theta\right]$.
\begin{prop}[El Karoui et al. \cite{ElKaroui2017}]
Let functions $g_m$ and $h_m$ the regular finite variations solutions of delayed differential equations (DDE), denoted $\operatorname{DDE}(\beta)$, restricted to the interval $[0, m]:$
\begin{equation}
 \beta g_m^{\prime}(x)=g_m(x)-g_m\left((x-1)^{+}\right)-1
 \label{dde_gm}
\end{equation}
with $g_m(0)=0$ (Cauchy problem),
\begin{equation}
  \beta h_m^{\prime}(x)=h_m(x+1)-h_m(x)+1
  \label{dde_hm}
\end{equation}
with $h_m(0)=0$ and $h_m^{\prime}(0)=0$ (Neumann problem).\\
The same properties hold true for the functions $\tilde{g}_m$ and $\tilde{h}_m$, solutions of the same system where $\beta$ is replaced by $\tilde{\beta}$, with $\tilde{\rho}=1 / \rho$, and $\tilde{\beta}=\beta(\tilde{\rho})=\beta(1 / \rho)=\beta(\rho) / \rho$. Then,
\begin{align}
&g_m(y)=\int_y^m W(z) \mathrm{d} z, \quad \text { with }~ g_m(0)=\int_0^m W(z) \mathrm{d} z \\
&h_m(x)=W(m-x) \frac{W(m)}{W^{\prime}(m)}-\int_0^{m-x} W(y) \mathrm{d} y, \quad h_m(m^-)=\frac{W(m)}{\beta W^{\prime}(m)}\\
&\tilde{g}_m(y)=\rho \int_y^m \rho^z W(z) \mathrm{d} z, \\
&\tilde{h}_m(x)=\widetilde{W}(m-x) \frac{\widetilde{W}(m)}{\widetilde{W}^{\prime}(m)}-\int_0^{m-x} \widetilde{W}(z) \mathrm{d} z, \quad \tilde{h}_m(m^-)=\frac{\widetilde{W}(m)}{\tilde{\beta} \widetilde{W}^{\prime}(m)}
\end{align}
where $$W(x)=\frac{1}{\beta} \sum_{k=0}^{\lfloor x\rfloor} \frac{(-1)^k}{k !}((x-k) / \beta)^k \exp ((x-k) / \beta) $$$$\int_0^x W(y) \mathrm{d} y = \sum_{k=0}^{\lfloor x\rfloor}\left(e^{(x-k) / \beta}\left(\sum_{i=0}^k \frac{(-1)^j}{j !}((x-k) / \beta)^j\right)-1\right)$$ if $\beta>1$ and $W(x)=\tilde{\rho} \tilde{\rho}^x \widetilde{W}(x)$ if $\beta<1$.
\end{prop}
\begin{proof}
    Refer to sections 4 and 6 in the work of El Karoui et al. \cite{ElKaroui2017}.
\end{proof}

\begin{proof}[Proof of Theorem \ref{ARL_rho_inf_1}]
Let $\rho<1$, $m \geq 0$, $0 \leq t \leq T$, $U(t) = \sum_{i=1}^D N^i(t) - \beta(\rho) \Lambda^i(t)$, $\overline{U}(t) = \sup_{0\leq s\leq t}U(s)$, $\widetilde{U}(t)=\sup _{0 \leq s \leq t} U(s) - U(t)$ and $g_{m}$ the continuous solution of the equation \ref{dde_gm}. Since 
 $N^i - \beta(\rho) \Lambda^i$ decreases between two event times of $N^i$, the process $\overline{U}$ changes only when $N^i$ jumps at time $t$ such that $\widetilde{U}(t)=0$. The jumps of $\widetilde{U}$ are therefore negative with sizes less than or equal to D. This means that, $\forall 0\leq t\leq T$, $$\overline{U}(t) - \overline{U}(t^-) = \mathbbm{1}_{\{\widetilde{U}(t)=0\}}\left(\sum_{i=1}^D N^i(t) - \beta(\rho) \Lambda^i(t) - \widetilde{U}(t^-) - \sum_{i=1}^DN^i(t^-) + \beta(\rho) \Lambda^i(t^-)\right)$$ Because $(\Lambda^i)_{1\leq i\leq D}$ are continuous, we have that $\sum_{i=1}^D\Lambda^i(t) = \sum_{i=1}^D\Lambda^i(t^-)$ and that $\overline{U}(t) - \overline{U}(t^-) = \mathbbm{1}_{\{\widetilde{U}(t)=0\}}\left(\sum_{i=1}^DN^i(t) -N^i(t^-) -\widetilde{U}(t^-)\right)$. Thus, $\overline{U}$ and $\widetilde{U}$ are solutions of the following stochastic differential equations~:
\begin{equation}
    \mathrm{d}\overline{U}(t)= \sum_{i=1}^D\mathbbm{1}_{\{\widetilde{U}(t)=0\}}\left(1 - \frac{\widetilde{U}(t^-)}{\sum_{\substack{j=1}}^{D}N^j(t)-N^j(t^-)}\right)\mathrm{d}N^i(t)
\end{equation}
\begin{equation}
\mathrm{d}\widetilde{U}(t) = - \sum_{i=1}^D\left(1 \wedge \frac{\widetilde{U}(t^-)}{\sum_{\substack{j=1}}^{D}N^j(t)-N^j(t^-)}\right)\mathrm{d}N^i(t) + \beta(\rho)\sum_{i=1}^D\mathrm{d}\Lambda^i(t)
\end{equation}
 The transformation described in \ref{def_transf1} is formulated in a manner that ensures the jumps of  $\widetilde{U}^{\epsilon}$ depend on one direction only when considering events that are elements of set $A^{\epsilon,c}$, i.e,  
\[
\mathrm{d}\widetilde{U}^{\epsilon}(t) = - \sum_{i=1}^D\left(1 \wedge \widetilde{U}^{\epsilon}(t^-)\right)\mathrm{d}N^{i,\epsilon}(t) + \beta(\rho)\sum_{i=1}^D\mathrm{d}\Lambda^{i,\epsilon}(t)
\]

    It should be noted that, for clarity in notation, we have omitted the fact that the expectations evaluated before taking the limit should include $\mathbbm{1}_{A^{\epsilon,c}}$. 
    
Let $\mathrm{d}N^{i,m,\widetilde{U}^{\epsilon}}(t) = \mathbbm{1}_{[0,m]}\left(\widetilde{U}^{\epsilon}(t^-)\right)\mathrm{d}N^{i,\epsilon}(t),~\forall i \in \{1,\dots,D\}$. Applying the itô formula to $G^{m,\epsilon}(t) := g_{m}(\widetilde{U}^{\epsilon}(t)) - g_{m}(\widetilde{U}(0)) + \sum_{i=1}^DN^{i,m,\widetilde{U}^{\epsilon}}(t)$ yields the following decomposition~:
 \begin{equation*}
    \begin{split}
    \mathrm{d}G^{m,\epsilon}(t) &= \mathrm{d}g_{m}(\widetilde{U}^{\epsilon}(t)) + \sum_{i=1}^D\mathrm{d}N^{i,m,\widetilde{U}^{\epsilon}}(t)\\ &=\beta(\rho)g^{'}_{m}(\widetilde{U}^{\epsilon}(t^-))\sum_{i=1}^D\mathrm{d}\Lambda^{i,\epsilon}(t) + g_{m}(\widetilde{U}^{\epsilon}(t))-g_{m}(\widetilde{U}^{\epsilon}(t^-))+ \sum_{i=1}^D\mathrm{d}N^{i,m,\widetilde{U}^{\epsilon}}(t)\\
    &=\beta(\rho)g^{'}_{m}(\widetilde{U}^{\epsilon}(t^-))\sum_{i=1}^D\mathrm{d}\Lambda^{i,\epsilon}(t) \\&\quad- \sum_{i=1}^D\left(g_{m}(\widetilde{U}^{\epsilon}(t^-))-g_{m}((\widetilde{U}^{\epsilon}(t^-)-1)^+)\right)\mathrm{d}N^{i,m,\widetilde{U}^{\epsilon}}(t)+ \sum_{i=1}^D\mathrm{d}N^{i,m,\widetilde{U}^{\epsilon}}(t)
    \end{split}
\end{equation*}
 Since $\beta g_{m}^{\prime}(x)=g_{m}(x)-g_{m}\left((x-1)^{+}\right)-1$ on $[0,m]$ and $g_{m}(x) = g^{\prime}_{m}(x) = 0$ if $x>m$ or $x<0$~:
 \begin{equation*}
    \begin{split}
    \mathrm{d}G^{m,\epsilon}(t) &= \beta(\rho)g_{m}^{\prime}(\widetilde{U}^{\epsilon}(t^-))\sum_{i=1}^D\mathrm{d}\left(\Lambda^{i,\epsilon}(t)-N^{i,\epsilon}(t)\right)
    \end{split}
\end{equation*}
Thus, $G^{m,\epsilon}$ is a $\mathcal{F}$-martingale. According to Doob's stopping theorem~:
\begin{equation}
\begin{split}
    \mathbb{E}_y\left(g_{m}(\widetilde{U}^{\epsilon}(\widetilde{T}_{\mathrm{C}})) - g_{m}(\widetilde{U}(0)) + \sum_{i=1}^D N^{i,m,\widetilde{U}^{\epsilon}}(\widetilde{T}_{\mathrm{C}}^{\epsilon})\right) &= \mathbb{E}_y\left(g_{m}(\widetilde{U}^{\epsilon}(0))\right) - \mathbb{E}_y\left(g_{m}(\widetilde{U}(0))\right)
    \\&=0
\end{split}
\label{dem_arl_equality_decrease}
\end{equation}
where $\widetilde{T}_{\mathrm{C}}$ is the following $\mathcal{F}$-stopping time~:
\begin{equation*}
\widetilde{T}_{\mathrm{C}} = \inf \left\{t \leq T: \widetilde{U}(t)>m\right\}
\end{equation*} 
The stopping time $\widetilde{T}_{\mathrm{C}}$ ensures that $\widetilde{U}^{\epsilon}(\widetilde{T}_{\mathrm{C}})$ converges to $\widetilde{U}(\widetilde{T}_{\mathrm{C}})$ such that $\mathbbm{1}_{[0,m]}\left(\widetilde{U}^{\epsilon}
({\widetilde{T}^{-}_{\mathrm{C}}})\right)$ does not nullify. Indeed, by employing a straightforward induction argument, we can demonstrate that $\widetilde{U}^{\epsilon}\leq \widetilde{U}$ holds over the intervals $\left[\tau^{\epsilon}_k,\tau_{k+1}\right[$ for all $k \in \left\{1,\dots,\sum_{i=1}^D N^i(T)\right\}$. Leveraging the fact that there does not exist any $k\in\mathbb{N}$ and $\omega\in\Omega$ such that $\widetilde{T}_C(\omega) = \tau_k(\omega)$, as a jump in $\sum_{i=1}^D N^i$ corresponds to a decrease in $\widetilde{U}$, we get that $\lim_{\epsilon\rightarrow 0^+}\mathbb{P}\left(B^{\epsilon}\right) = 0$, where $B^{\epsilon}:=\left\{\omega\in\Omega: \exists k \in \mathbb{N}/ \widetilde{T}_C(\omega)\in]\tau_k(\omega),\tau^{\epsilon}_k(\omega)]\right\}$. Hence, for sufficiently small $\epsilon$, it follows that $\lim_{\epsilon\rightarrow 0^+}\mathbb{P}\left(\tau^{\epsilon}_k\leq \widetilde{T}_C <\tau_{k+1}\right)=1$. Drawing upon the results established in Lemma \ref{y_borne}, we have $\lim_{\epsilon \rightarrow 0^+}\widetilde{U}^{\epsilon}(\widetilde{T}^-_{\mathrm{C}}) = {\widetilde{U}(\widetilde{T}^-_{\mathrm{C}})} = m,~~\text{under}~~\mathcal{L}^1$. Consequently, $\lim_{\epsilon \rightarrow 0^+}\mathbb{E}_y\left(\sum_{i=1}^D N^{i,m,\widetilde{U}^{\epsilon}}(\widetilde{T}_{\mathrm{C}})\right) =\lim_{\epsilon \rightarrow 0^+}\mathbb{E}_y\left(\mathbbm{1}_{B^{\epsilon,c}}\sum_{i=1}^D N^{i,m,\widetilde{U}^{\epsilon}}(\widetilde{T}_{\mathrm{C}})\right)= \mathbb{E}_y\left(\sum_{i=1}^DN^i(\widetilde{T}_{\mathrm{C}})\right)$.

Out of continuity of $g_{m}$ and using the theorem of dominated convergence, we have $$\lim_{\epsilon \rightarrow 0^+}\mathbb{E}_y\left(g_{m}\left(\widetilde{U}^{\epsilon}(\widetilde{T}_{\mathrm{C}})\right)\right) = \mathbb{E}_y\left(g_{m}(\lim_{\epsilon \rightarrow 0^+}\widetilde{U}^{\epsilon}(\widetilde{T}_{\mathrm{C}}))\right) = g_{m}(m) = 0$$ Taking $\epsilon$ to 0 in equation \ref{dem_arl_equality_decrease}, yields the following result~:
$$\mathbb{E}_y\left(\sum_{i=1}^D N^i(\widetilde{T}_{\mathrm{C}})\right) = g_{m}(y),\quad \forall m\geq 0$$
\end{proof}
In the previous proof, we introduced the processes $(N^{i,\epsilon})_{1\leq i\leq D}$ because we looked for a transformation $\widetilde{U}^{\epsilon}$ of the process $\widetilde{U}$ with non-simultaneous jumps and that is inferior to the original reflected process $\widetilde{U}(t)$ outside of the intervals $\left[\tau_k,\tau^{\epsilon}_k \right[$ in order to converge leftwards to $\widetilde{U}$ at $\widetilde{T}_{\mathrm{C}}$ with $k \in \left\{1,\dots,\sum_{i=1}^D N^i(T)\right\}$. This allowed us to apply the results cited in El Karoui et al. \cite{ElKaroui2017} and Poor and Hadjiliadis \cite{VPoor2009} to the $\sum_{i=1}^D N^{i,\epsilon}$ and $\widetilde{U}^{\epsilon}$ processes and then go to the limit to obtain the desired results on the original processes. \\In what follows, we propose to introduce new processes $\widetilde{N}^{i,\epsilon}$ and $\hat{U}^{\epsilon}$ to adapt this line of reasoning to the case where $\rho>1$.
\begin{defi}
Let $(\widetilde{N}^{i,\epsilon})_{1\leq i \leq D}$ be the counting process whose arrival times are~: $(\widetilde{\tau}^{i,\epsilon}_k)_{k \geq 1} = \left(\tau^i_k-\frac{i}{D}\epsilon\underset{1 \leq l,l^{\prime} \leq D}{\inf}~~\underset{0 \leq \tau^l_j < \tau^{l^{\prime}}_{j^{\prime}} \leq \tau^i_k}{\inf}\mid \tau^l_j -\tau^{l^{\prime}}_{j^{\prime}}\mid\right)_{k \geq 1} = \left(\tau^i_k-\epsilon^i_k\right)_{k \geq 1}$ where $\tau^i_0 = 0$, $(\tau^i_k)_{k \geq 1}$ are the event times of the process $N^i$, $\forall i \in \{1,\dots,D\}$.  In this context, we define $\widetilde{N}^{i,\epsilon}$ as the counting process whose arrival times are $\left(\widetilde{\tau}^{i,\epsilon}_k\right)_{k \geq 1}$, i.e, $\widetilde{N}^{i,\epsilon}(t) = \sum_{k \geq 1}\mathbbm{1}_{\{\widetilde{\tau}^{i,\epsilon}_k \leq t\}}$, $\forall 0\leq t\leq T$. We also define $\Tilde{\mathcal{F}}^{\epsilon} = \left(\Tilde{\mathcal{F}}^{\epsilon}_t\right)_{t\geq 0}$ as the natural filtration associated to $\sum_{i=1}^D\widetilde{N}^{i,\epsilon}$ with $\Tilde{\mathcal{F}}^{\epsilon}_t = \sigma(\widetilde{N}^{i,\epsilon}(s), 1 \leq i \leq D; s \leq t)$. 
\label{def_transf2}
\end{defi}
\begin{rque}
It is worth noting that $\mathcal{F}^{\epsilon}_t \subset \mathcal{F}_{t-\epsilon T}$ while $\Tilde{\mathcal{F}}^{\epsilon}_t \subset \mathcal{F}_{t+\epsilon T}$, for all $0\leq t \leq T$. Since $\forall 1\leq i\leq D$, $N^i$ is cadlàg, we have $\lim_{\epsilon \rightarrow 0^+} \widetilde{N}^{i,\epsilon}(t) = \lim_{\epsilon \rightarrow 0^+} N^i(t+\epsilon_i) = N^i(t^+) = N^i(t) ,~a.s.$ Therefore, we find almost-sure convergence for the process $\widetilde{N}^{i,\epsilon}$ and we limit ourselves to convergence in $\mathrm{L}^1$ for the process $N^{i,\epsilon}$.
\label{filtrations}
\end{rque}
\begin{rque}
    This means that $\widetilde{N}^{i,\epsilon}$ is $\Tilde{\mathcal{F}}^{\epsilon}$-adapted and that, $\forall 0\leq t \leq T$,
\[
\begin{split}
\widetilde{N}^{i,\epsilon}(t) &= \sum_{k \geq 1}\mathbbm{1}_{\{\Tilde{\tau}^{i,\epsilon}_k \leq t\}} \\ &\geq N^i(t + \frac{1}{D}\epsilon\underset{1\leq l,l^{\prime} \leq D}{\inf}~~\underset{0 \leq \tau^l_j < \tau^{l^{\prime}}_{j^{\prime}} \leq T}{\inf}\mid \tau^l_j -\tau^{l^{\prime}}_{j^{\prime}}\mid ) \\ &\geq N^i(t),\quad \text{a.s}\end{split}\]
Note that unlike $N^i-N^{i,\epsilon}$ (see \ref{rque_de_transf1}), $N^i-\widetilde{N}^{i,\epsilon}$ is an $\mathcal{F}^{\epsilon}$-sub-martingale. According to the Doob-Meyer decomposition, it can be inferred that $\Lambda^i-\widetilde{\Lambda}^{i,\epsilon}$ is a predictable process that is decreasing and starts from zero. Consequently, $\forall 0\leq t \leq T$,
$$\Lambda^i(t) \leq \widetilde{\Lambda}^{i,\epsilon}(t),\quad \text{a.s}$$
\end{rque}
\begin{nota}
 In the subsequent discussion, we shall refer to the process $\inf_{0\leq s\leq t}U(s)$ (resp. $\inf_{0\leq s\leq t}U^{\epsilon}(s)$) as $\underline{U}(t)$ (resp. $\underline{U}^{\epsilon}(t)$) where, in the case $\rho>1$,  $U^{\epsilon}(t):=\sum_{i=1}^D\widetilde{N}^{i,\epsilon}(t)-\beta(\rho) \widetilde{\Lambda}^{i,\epsilon}(t)$, $\forall 0\leq t \leq T$. We will also use the notation $\hat{U}^{\epsilon}$ to represent  $\hat{U}^{\epsilon}(t) := \sum_{i=1}^D\widetilde{N}^{i,\epsilon}(t)+\beta(\rho) \Tilde{\Lambda}^{i,\epsilon}(t) - \inf_{0 \leq s \leq t} \left(\sum_{i=1}^D\widetilde{N}^{i,\epsilon}(s)-\beta(\rho) \Tilde{\Lambda}^{i,\epsilon}(s)\right)$.
\end{nota}
Figures \ref{N_V_epsilon1} and \ref{N_V_epsilon2} showcase a comparative simulation of the processes $\sum_{i=1}^D\widetilde{N}^{i,\epsilon}$ and $\hat{U}^{\epsilon}$  with respect to $\sum_{i=1}^DN^i$ and $\hat{U}$.
\begin{figure}[H]

        \centering
    \begin{subfigure}[b]{0.49\textwidth}  
        \centering
        \includegraphics[width=\textwidth]{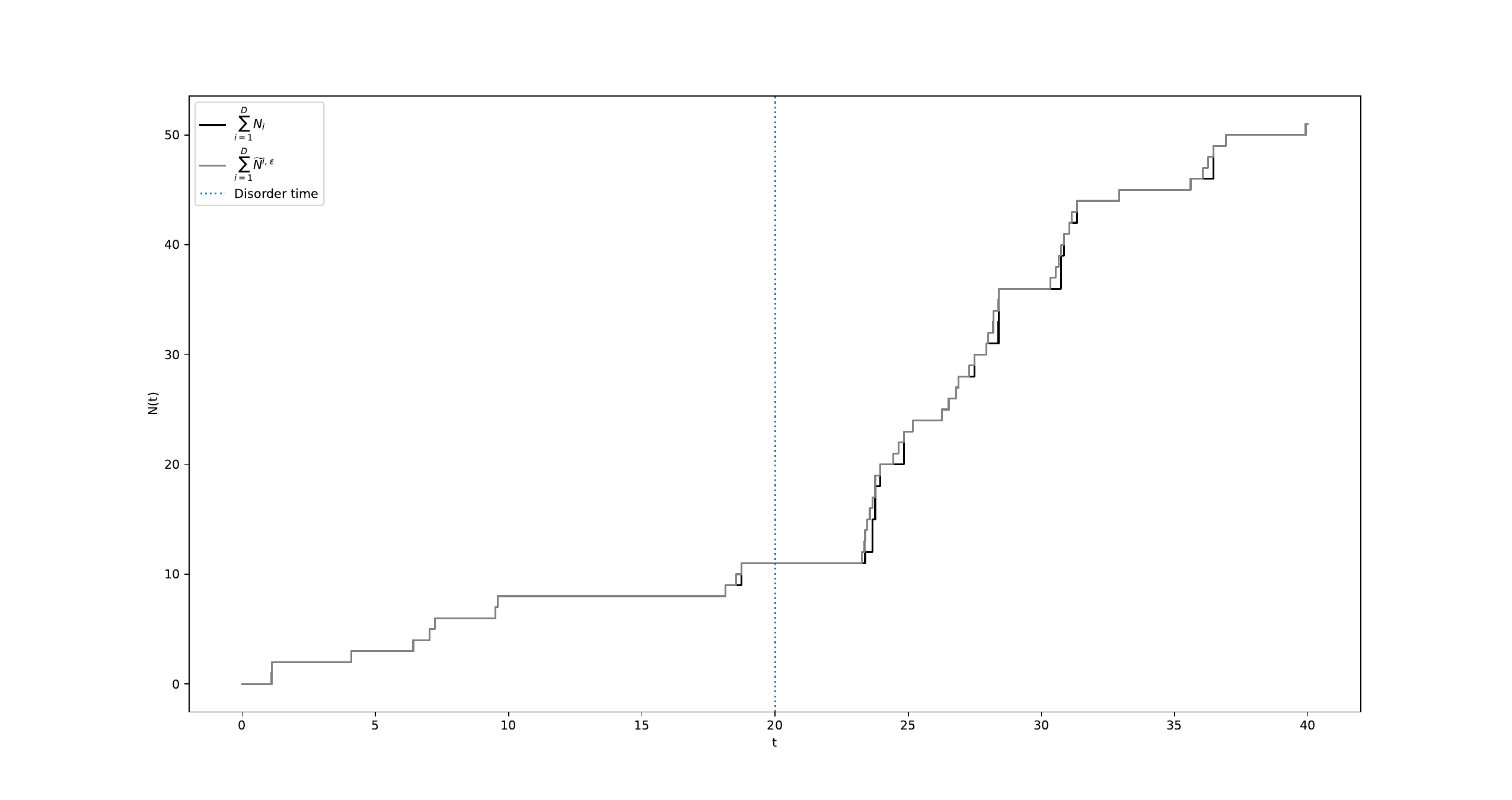}
        \caption{Illustration of $\widetilde{N}^{i,\epsilon}(t)$}
        \label{N_V_epsilon1}
    \end{subfigure}
    \begin{subfigure}[b]{0.49\textwidth}
    
        \centering
        \includegraphics[width=\textwidth]{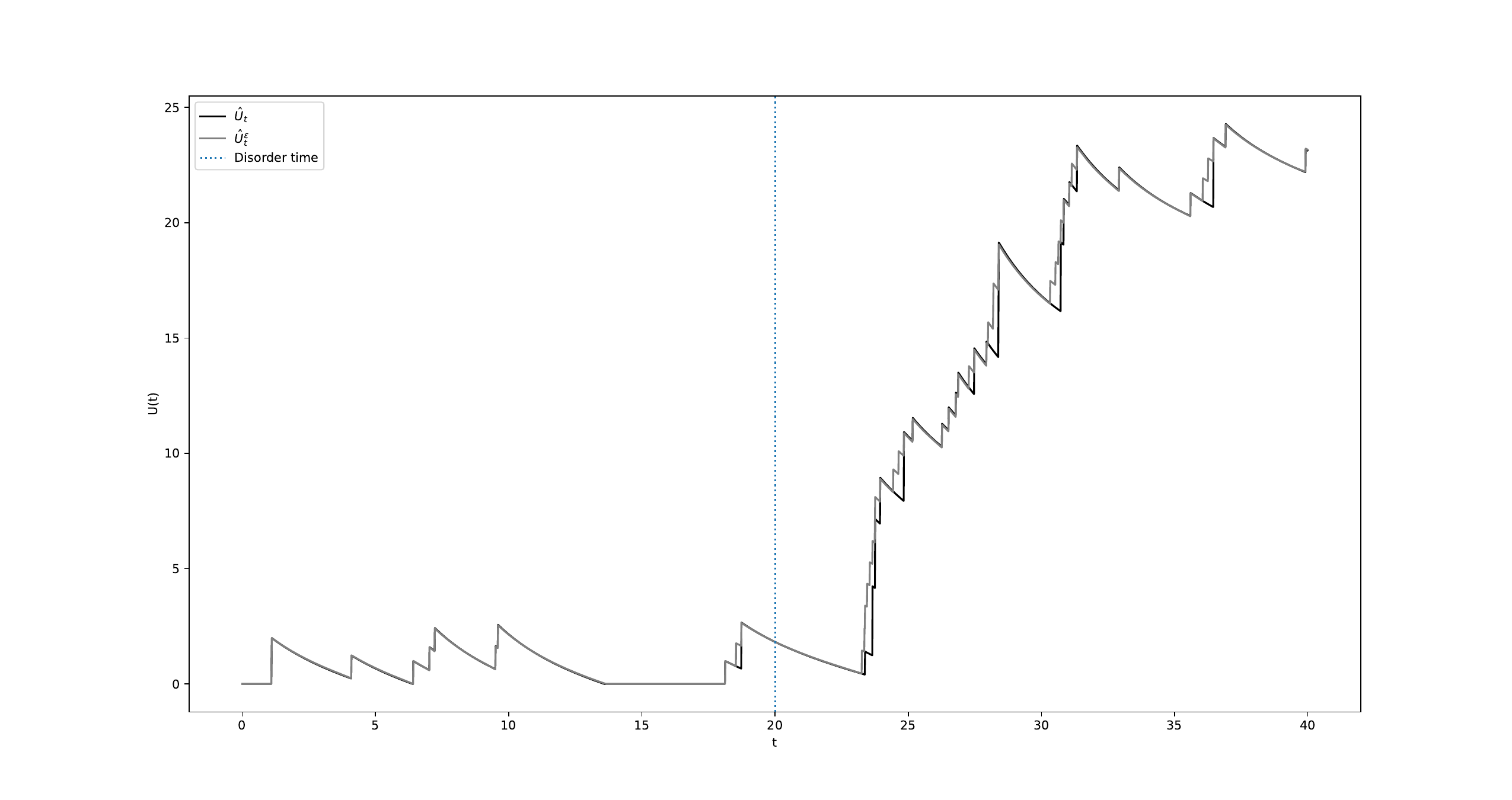}
        \caption{Illustration of $\hat{U}^{\epsilon}(t)$}
        \label{N_V_epsilon2}
    \end{subfigure}  
\end{figure}
\begin{lemma}
    Let $t \in [0, T]$ and $\hat{U}(t)$ (resp. $\hat{U}^{\epsilon}(t)$) the CUSUM process of $\sum_{i=1}^D N^i(t)$ (resp. $\sum_{i=1}^D \widetilde{N}^{i,\epsilon}(t)$) for $\rho >1$. Thus, $$\lim_{\epsilon \rightarrow 0^+} \mid \hat{U}^{\epsilon}(t) - \hat{U}(t)\mid = 0,\quad a.s.$$
    \label{conv_increase}
\end{lemma}
\begin{proof}
The proof is similar to that of the Lemma \ref{y_borne}. Indeed, ($\widetilde{U}(t)$, $\widetilde{U}^{\epsilon}(t)$) and ($\hat{U}^{\epsilon}(t)$,$\hat{U}(t)$) play symmetric roles.
\end{proof}
\begin{proof}[Proof of Theorem \ref{ARL_rho_inf_2}]
Let $\rho>1$, $0 \leq t \leq T$, $U(t) = \sum_{i=1}^DN^i(t) - \beta(\rho) \Lambda^i(t)$ and $h_{m}$ a continuous solution of the equation \ref{dde_hm}.   
We adopt the same notations as in the prior part and introduce the process $\hat{U}^{\epsilon}(t)$. Note that $\hat{U}^{\epsilon}(t)$ only increases if a jump takes place and that $\underline{U}^{\epsilon}(t) = \inf_{0 \leq s \leq t}\sum_{i=1}^D\widetilde{N}^{i,\epsilon}(s) - \beta(\rho) \Tilde{\Lambda}^{i,\epsilon}(s)$ is continuous. Applying the Itô formula to $\left(\hat{U}^{\epsilon}(t)\right)^{\mu+1}$, we get~:
\begin{align*}
    \mathrm{d}\left(\hat{U}^{\epsilon}(t)\right)^{\mu+1}=\sum_{i=1}^D&\left(\left(\hat{U}^{\epsilon}(t^-)+1\right)^{\mu+1}-\left(\hat{U}^{\epsilon}(t^-)\right)^{\mu+1}\right) \mathrm{d} \widetilde{N}^{i,\epsilon}(t)\\&-(1+\mu) \left(\hat{U}^{\epsilon}(t^-)\right)^{\mu}\left(\beta(\rho) \mathrm{d} \sum_{i=1}^D\Tilde{\Lambda}^{i,\epsilon}(t) - \mathrm{d}\underline{U}^{\epsilon}(t)\right)\end{align*} Since $\underline{U}^{\epsilon}$ only decreases when $\hat{U}^{\epsilon}(t) = \hat{U}^{\epsilon}(t^-) =0$, then $$\left(\hat{U}^{\epsilon}(t^-)\right)^{\mu}\left(\beta(\rho)\sum_{i=1}^D \mathrm{d} \Tilde{\Lambda}^{i,\epsilon}(t)-\mathrm{d} \underline{U}^{\epsilon}(t)\right) = \left(\hat{U}^{\epsilon}(t^-)\right)^{\mu}\beta(\rho) \mathrm{d} \sum_{i=1}^D\Tilde{\Lambda}^{i,\epsilon}(t)$$ Going to the limit $\mu \rightarrow 0$, we get~:$$\mathrm{d} \hat{U}^{\epsilon}(t)=\sum_{i=1}^D\mathrm{d} \widetilde{N}^{i,\epsilon}(t)-\mathbbm{1}_{\left\{\hat{U}^{\epsilon}(t)>0\right\}} \beta(\rho) \mathrm{d}\Tilde{\Lambda}^{i,\epsilon}(t),\quad a.s$$
 Let $J^{m, \hat{U}^{\epsilon}}(t):=\lim_{\mu \mapsto 0} \sum_{\mu \leq s \leq t} \mathbbm{1}_{\left\{\hat{U}^{\epsilon}(s-\mu)>m>\hat{U}^{\epsilon}(s+\mu), \hat{U}^{\epsilon}(s)=\hat{U}^{\epsilon}(s^-)=m\right\}}$\footnote{Also referred to as discontinuous local time.}, $\forall 0\leq t\leq T$. The Itô-Tanaka-Meyer formula applied to $H^{m,\epsilon}(t) := h_{m}\left(\hat{U}^{\epsilon}(t)\right) - h_m(\hat{U}(0)) + \sum_{i=1}^D\widetilde{N}^{i,m,\hat{U}^{\epsilon}}(t) + h_m(m^-)J^{m, \hat{U}^{\epsilon}}(t)$ results in the following decomposition~:
 \begin{equation*}
    \begin{split}
    \mathrm{d}H^{m,\epsilon}(t) &= \mathrm{d}h_{m}(\hat{U}^{\epsilon}(t)) + \sum_{i=1}^D\mathrm{d}\widetilde{N}^{i,m,\hat{U}^{\epsilon}}(t)+h_m(m^-)\mathrm{d}J^{m, \hat{U}^{\epsilon}}(t)\\ &=\beta(\rho)h^{'}_{m}(\hat{U}^{\epsilon}(t^-))\mathbbm{1}_{\left\{\hat{U}^{\epsilon}(t^-)>0\right\}}\sum_{i=1}^D\mathrm{d}\Tilde{\Lambda}^{i,\epsilon}(t) + h_{m}(\hat{U}^{\epsilon}(t))-h_{m}(\hat{U}^{\epsilon}(t^-))+ \sum_{i=1}^D\mathrm{d}\widetilde{N}^{i,m,\hat{U}^{\epsilon}}(t)\\
    &=\beta(\rho)h^{'}_{m}(\hat{U}^{\epsilon}(t^-))\mathbbm{1}_{\left\{\hat{U}^{\epsilon}(t^-)>0\right\}}\sum_{i=1}^D\mathrm{d}\Tilde{\Lambda}^{i,\epsilon}(t) \\&\quad+ \sum_{i=1}^D\left(h_m(\hat{U}^{\epsilon}(t^-)+1)-h_m(\hat{U}^{\epsilon}(t^-))\right)\mathrm{d}\widetilde{N}^{i,\epsilon}(t)+ \sum_{i=1}^D\mathrm{d}\widetilde{N}^{i,m,\hat{U}^{\epsilon}}(t)
    \end{split}
\end{equation*}
where $\mathrm{d}\widetilde{N}^{i,m,\hat{U}^{\epsilon}(t)}(t) = \mathbbm{1}_{[0,m]}\left(\hat{U}^{\epsilon}(t^-)\right)\mathrm{d}\widetilde{N}^{i,\epsilon}(t)$ and $m \geq 0$.\\
Since $\beta h_m^{\prime}(x)=h_m(x+1)-h_m(x)+1$ if $x\in [0,m]$, $h_m(0) = h^{\prime}_m(0) = 0$ and $h_m(x) = 0$ if $x\geq m$, thus~:
 \begin{equation}
 \label{dH_m}
    \begin{split}
    \mathrm{d}H^{m,\epsilon}(t) &= \beta(\rho)h_m^{\prime}(\hat{U}^{\epsilon}(t^-))\mathbbm{1}_{\left\{\hat{U}^{\epsilon}(t)>0\right\}}\sum_{i=1}^D\mathrm{d}\left(\Tilde{\Lambda}^{i,\epsilon}(t)-\widetilde{N}^{i,\epsilon}(t)\right)
    \end{split}
\end{equation}
Hence, $H^{m,\epsilon}$ is an $\Tilde{\mathcal{F}}$-martingale, with $\left(\Tilde{\mathcal{F}}_t\right)_{t\geq 0} = \left(\mathcal{F}_{t+\epsilon T}\right)_{t\geq 0}$. Therefore, according to Doob's stopping theorem~:
\begin{equation*}
\begin{split}
\mathbb{E}_v\left(h_m\left(\hat{U}^{\epsilon}(\hat{T}_{\mathrm{C}}^{\epsilon})\right) - h_m(\hat{U}(0)) - h_m(m^-) + \sum_{i=1}^D \widetilde{N}^{i,m,\hat{U}^{\epsilon}}\left(\hat{T}_{\mathrm{C}}^{\epsilon}\right)\right) &= \mathbb{E}_v\left(h_m(\hat{U}^{\epsilon}(0))\right) - \mathbb{E}_v\left(h_m(\hat{U}(0))\right)
\end{split}
\end{equation*}
where $\hat{T}_{\mathrm{C}}$ and $\hat{T}_{\mathrm{C}}^{\epsilon}$ are the following $\Tilde{\mathcal{F}}$-stopping times~: \begin{equation*}
\begin{split}\hat{T}_{\mathrm{C}} = \inf \left\{t \leq T: \hat{U}(t)>m\right\}\quad\text{and}\quad \hat{T}_{\mathrm{C}}^{\epsilon} &= \hat{T}_{\mathrm{C}}+\epsilon T \end{split}
\end{equation*}
Just as in the case where $\rho<1$, the stopping time $\hat{T}_{\mathrm{C}}^{\epsilon}$ is introduced to ensure that $\hat{U}^{\epsilon}(\hat{T}_{\mathrm{C}}^{\epsilon})$ converges to $\hat{U}(\hat{T}_{\mathrm{C}})$ (see Lemma \ref{conv_increase}) such that the process $\mathbbm{1}_{[0,m]}\left(\hat{U}^{\epsilon}({\hat{T}_{\mathrm{C}}^{\epsilon-}})\right)$ does not nullify since. This is achieved by maintaining $\hat{U}^{\epsilon}\leq \hat{U}$ over the intervals $\left[\tau^{\epsilon}_k,\tau_{k+1}\right[$ for all $k \in \left\{1,\dots,\sum_{i=1}^D N^i(T)\right\}$. In order to achieve this, it is sufficient to show that $$\lim_{\epsilon\rightarrow 0^+}\mathbb{P}\left(C^{\epsilon}\right) = 0$$ where $C^{\epsilon}:=\left\{\omega\in\Omega: \exists k \in \mathbb{N}/ \tau_k(\omega)<\widetilde{T}_C(\omega)\leq\tau^{\epsilon}_{k+1}(\omega)\leq\widetilde{T}^{\epsilon}_C(\omega)\right\}$.\\ Given that the demonstration bears resemblance to that which was undertaken in Lemma \ref{conv_transf1}, we will not detail it here. Since $\lim_{\epsilon \rightarrow 0^+}\hat{U}^{\epsilon}(\hat{T}_{\mathrm{C}}^{\epsilon}) = {\hat{U}(\hat{T}_{\mathrm{C}})}^- = m^-,~~\text{under}~~\mathrm{L}^1$, by continuity of $h_{m}$, we have $\lim_{\epsilon \rightarrow 0^+}h_{m}(\hat{U}^{\epsilon}(\hat{T}^{\epsilon}_{\mathrm{C}})) = h_{m}(\lim_{\epsilon \rightarrow 0^+}\hat{U}^{\epsilon}(\hat{T}^{\epsilon}_{\mathrm{C}})) = h_m(m^-)$. Moreover, since $N^i$ is càdlàg and finite almost-surely, then $\lim_{\epsilon \rightarrow 0^+}\widetilde{N}^{i,\epsilon}(t) = N^i(t)$ and $$\lim_{\epsilon \rightarrow 0^+}\mathbb{E}_v\left(\sum_{i=1}^D \widetilde{N}^{i,m,\hat{U}^{\epsilon}(t)}(\hat{T}^{\epsilon}_{\mathrm{C}})\right) = \lim_{\epsilon \rightarrow 0^+}\mathbb{E}_v\left(\sum_{i=1}^D\mathbbm{1}_{C^{\epsilon,c}}\widetilde{N}^{i,m,\hat{U}^{\epsilon}(t)}(\hat{T}^{\epsilon}_{\mathrm{C}})\right)= \mathbb{E}_v\left(\sum_{i=1}^D N^i(\hat{T}_{\mathrm{C}})\right)$$ Taking $\epsilon$ to 0, yields the following result~:
$$\mathbb{E}_v\left(\sum_{i=1}^D N^i(\hat{T}_{\mathrm{C}})\right) = h_m(v),\quad \forall m\geq 0$$
\end{proof}
The remaining part of the proof is derived from the results established previously in this section on the Average Run Length (ARL). The underlying approach involves initially establishing lower bounds on the Lorden criterion for both $\rho<1$ and $\rho>1$. Subsequently, it is demonstrated that the CUSUM stopping time attains these lower bounds, thereby establishing its optimality.

\begin{proof}[Proof of Theorem \ref{opt_decrease}]
 Let $D^{\theta,\epsilon} = \left\{\omega\in\Omega : \exists k\in\mathbb{N} / \tau_{k}^i(\omega)\leq \theta<\tau_k^{\epsilon}(\omega)\right\}$ and $D^{\theta,\epsilon,c} =\Omega\backslash D^{\theta,\epsilon}$, for all $\theta\in\mathbb{R}_+$. We also define the modified criterion $\widetilde{C}^{\epsilon}$, for all $\mathcal{F}$-stopping times $\tau$, as~: \begin{equation}
     \widetilde{C}^{\epsilon}(\tau):=\sup_{\theta \in[0, +\infty]} \esssup \sum_{i=1}^D\mathbb{E}^{\theta}\left[\left(N^{i,\epsilon}(\tau)-N^{i,\epsilon}(\theta)\right)^{+}\mathbbm{1}_{D^{\theta,\epsilon,c}} \big\lvert \mathcal{F}_\theta\right]
 \end{equation} 
We start with calculations under the probability $\widetilde{\mathbb{P}}$ and take the primitive of the conditional criterion $\widetilde{\Gamma}_t^{\tau}=\mathbb{E}^{\theta}\left(\mathbbm{1}_{D^{t,\epsilon,c}}\int_t^{\tau} \mathrm{d} \sum_{i=1}^D N^{i,\epsilon}(s) \big\lvert \mathcal{F}_t\right)$ with respect to the non-decreasing process $\rho^{-\bar{U}^{\epsilon}}=\tilde{\rho}^{\bar{U}^{\epsilon}}$. 
 An integration by parts yields~: $$\begin{aligned}\mathbb{E}^{\theta}\left[\int_t^{\tau} \widetilde{\Gamma}_s^{\tau} \mathrm{d} \rho^{-\bar{U}^{\epsilon}(s)} \big\lvert \mathcal{F}_t\right]&=\mathbb{E}^{\theta}\left[\int_t^{\tau}\mathbbm{1}_{D^{s,\epsilon,c}}\left(\int_{]s, \tau]} \mathrm{d} \left(\sum_{i=1}^DN^{i,\epsilon}(u)\right)\right) \mathrm{d} \rho^{-\bar{U}^{\epsilon}(s)} \big\lvert \mathcal{F}_t\right] \\&= \mathbb{E}^{\theta}\left[\int_t^{\tau} \mathrm{d} \left(\sum_{i=1}^DN^{i,\epsilon}(u)\right)\left(\rho^{-\bar{U}^{\epsilon}(u^-)}-\rho^{-\bar{U}^{\epsilon}(t)}\right)\mathbbm{1}_{D^{u,\epsilon,c}} \big\lvert \mathcal{F}_t\right]\end{aligned}$$
 Since $\hat{C}^{\epsilon}(\tau)\geq \widetilde{\Gamma}_t^{\tau}$, we have~:
$$
\begin{aligned}
\widetilde{C}^{\epsilon}(\tau) \mathbb{E}^{\theta}\left[\rho^{-\bar{U}^{\epsilon}(\tau)} \big\lvert \mathcal{F}_t\right] & \geq \mathbb{E}^{\theta}\left[\int_t^{\tau} \sum_{i=1}^D \mathrm{d}N^{i,\epsilon}(u) \rho^{-\bar{U}^{\epsilon}(u^-)}\mathbbm{1}_{D^{u,\epsilon,c}} \big\lvert \mathcal{F}_t\right]\\&+\mathbb{E}^{\theta}\left[\widetilde{C}^{\epsilon}(\tau)\rho^{-\bar{U}^{\epsilon}(t)}-\int_t^{\tau} \sum_{i=1}^D \mathrm{d}N^{i,\epsilon}(u)\rho^{-\bar{U}^{\epsilon}(t)}\mathbbm{1}_{D^{u,\epsilon,c}} \big\lvert \mathcal{F}_t\right] \\& \geq \mathbb{E}^{\theta}\left[\int_t^{\tau} \sum_{i=1}^D dN^{i,\epsilon}(u) \rho^{-\bar{U}^{\epsilon}(u^-)}\mathbbm{1}_{D^{u,\epsilon,c}} \big\lvert \mathcal{F}_t\right]\\&+\mathbb{E}^{\theta}\left[\rho^{-\bar{U}^{\epsilon}(t)}\left(\widetilde{C}^{\epsilon}(\tau)-\int_t^{\tau} \mathrm{d}\sum_{i=1}^D N^{i,\epsilon}(u)\right) \big\lvert \mathcal{F}_t\right] \\& \geq \mathbb{E}^{\theta}\left[\int_t^{\tau} \sum_{i=1}^D\mathrm{d}N^{i,\epsilon}(u) \rho^{-\bar{U}^{\epsilon}(u^-)}\mathbbm{1}_{D^{u,\epsilon,c}} \big\lvert \mathcal{F}_t\right]
\end{aligned}
$$
 This yields the following inequalities for all $\mathcal{F}$-stopping times $\tau$~:
\begin{equation*}
\rho\mathbb{E}^{\theta}\left[\int_t^{\tau} \sum_{i=1}^D \mathrm{d}N^{i,\epsilon}(u) \rho^{-\widetilde{U}^{\epsilon}(u^-)}\mathbbm{1}_{D^{u,\epsilon,c}}\big\lvert \mathcal{F}_t\right]\leq \widetilde{C}^{\epsilon}\left(\tau\right) \mathbb{E}^{\theta}\left[\rho^{-\widetilde{U}^{\epsilon}(\tau)}\big\lvert \mathcal{F}_t\right]
\end{equation*}
Taking the expectation on both sides and using the fact that $\rho^{U}$ is the martingale density of $\widetilde{\mathbb{P}}$ with respect to $\mathbb{P}$, we establish the following lower bound~:
\begin{equation}
\rho\mathbb{E}\left[\int_t^{\tau} \sum_{i=1}^D \mathrm{d}N^{i,\epsilon}(u) \rho^{-\widetilde{U}^{\epsilon}(u^-)}\rho^{U^{\epsilon}(u^-)-U(u^-)}\mathbbm{1}_{D^{u^-,\epsilon,c}}\right]\leq \widetilde{C}^{\epsilon}\left(\tau\right) \mathbb{E}\left[\rho^{-\widetilde{U}^{\epsilon}(\tau)}\rho^{U^{\epsilon}(\tau)-U(\tau)}\right]
  \label{lower_bound_epsilon}
\end{equation}
Since, $\forall \tau,\theta\geq 0$,
\begin{align*}
    \sum_{i=1}^D\mathbb{E}^{\theta}\left[\left(N^{i,\epsilon}(\tau)-N^{i,\epsilon}(\theta)\right)^{+}\mathbbm{1}_{D^{\theta,\epsilon,c}} \big\lvert \mathcal{F}_\theta\right] &=\sum_{i=1}^D\mathbb{E}^{\theta}\left[\left(N^{i,\epsilon}(\tau)-N^{i}(\theta)\right)^{+}\mathbbm{1}_{D^{\theta,\epsilon,c}} \big\lvert \mathcal{F}_\theta\right]\\&\leq \sum_{i=1}^D\mathbb{E}^{\theta}\left[\left(N^{i}(\tau)-N^{i}(\theta)\right)^{+}\mathbbm{1}_{D^{\theta,\epsilon,c}} \big\lvert \mathcal{F}_\theta\right]\\&\leq\sum_{i=1}^D\mathbb{E}^{\theta}\left[\left(N^{i}(\tau)-N^{i}(\theta)\right)^{+}\big\lvert \mathcal{F}_\theta\right]
\end{align*}
Hence,\begin{align*}
        \widetilde{C}^{\epsilon}(\tau)&\leq \widetilde{C}(\tau) 
    \end{align*}
     Note that we can establish the proof for $\lim_{\epsilon \rightarrow 0^+}\mathbb{P}\left(D^{\theta,\epsilon}\right)=1$ for all $\theta\geq 0$ using a similar approach as shown in \ref{conv_transf1}. Drawing upon the insights presented in equation \ref{lower_bound_epsilon} and leveraging the convergence findings established in Lemma \ref{y_borne}, we can derive the ensuing expression:~
    \begin{equation}
    \label{lower_bound_decrease}
    \rho\mathbb{E}\left[\int_t^{\tau} \sum_{i=1}^D \mathrm{d}N^{i}(u) \rho^{-\widetilde{U}(u^-)} \right]\leq \widetilde{C}(\tau) \mathbb{E}\left[\rho^{-\widetilde{U}(\tau)}\right]\end{equation}  
    Applying (i) of Theorem 7 in \cite{ElKaroui2017} to the process $\sum_{i=1}^DN^{i,\epsilon}$, we get that~:
    \begin{equation}
        \tilde{g}_m(0) \mathbb{E}\left(\rho^{-\widetilde{U}^{\epsilon}(\tau)}\right)\leq\rho \mathbb{E}\left(\int_0^{\tau} \sum_{i=1}^D\mathrm{d}N^{i,\epsilon}(u)\rho^{-\widetilde{U}^{\epsilon}(u^-)}\right)
        \end{equation}
       By taking the limit and utilizing equation \ref{lower_bound_decrease}, we obtain~:
        $$
        \tilde{g}_m(0)\leq \widetilde{C}(\tau)
        $$
    In accordance with Theorem \ref{ARL_rho_inf_1}, we can conclude that the CUSUM stopping time $\widetilde{T}_{\mathrm{C}}$ achieves the lower bound of the Lorden criterion $\Tilde{C}$. This establishes the optimality of the CUSUM stopping time.
\end{proof}
\begin{proof}[Proof of Theorem \ref{opt_increase}]
The proof of this theorem is very similar to proof \ref{opt_decrease}.
\end{proof}

\subsection{Proof of section \ref{experimental_results}}
\begin{proof}[Proof of Proposition \ref{res_analysis_recursive}] 

\begin{align*}
\Lambda_g^i\left(\tau_{k}^i\right) -  \Lambda_g^i\left(\tau_{k-1}^i\right)&= \int_{\tau_{k-1}^i}^{\tau_k^i} \lambda_g^i(s) \mathrm{d} s \\&= \int_{\tau_{k-1}^i}^{\tau_k^i} \mu^i(s) \mathrm{d} s +\sum_{j=1}^D \sum_{\tau_{k^{\prime}}^j<\tau_{k-1}^i} \frac{\alpha_{i j}}{\beta_{i j}}g_j\left(v_{k^{\prime}}^j\right)\left[e^{-\beta_{i j}\left(\tau_{k-1}^i-\tau_{k^{\prime}}^j\right)}-e^{-\beta_{i j}\left(\tau_k^i-\tau_{k^{\prime}}^j\right)}\right] \\
& \quad+\sum_{j=1}^D \sum_{\tau_{k-1}^i \leq \tau_{k^{\prime}}^j<\tau_k^i} \frac{\alpha_{i j}}{\beta_{i j}}g_j\left(v_{k^{\prime}}^j\right)\left[1-e^{-\beta_{i j}\left(\tau_k^i-\tau_{k^{\prime}}^j\right)}\right] \\&= \int_{\tau_{k-1}^i}^{\tau_k^i} \mu^i(s) \mathrm{d} s +\sum_{j=1}^D \frac{\alpha_{i j}}{\beta_{i j}}\left[1-e^{-\beta_{i j}\left(\tau_k^i-\tau_{k-1}^i\right)}\right]\sum_{\tau_{k^{\prime}}^j<\tau_{k-1}^i}g_j\left(v_{k^{\prime}}^j\right)e^{-\beta_{i j}\left(\tau_{k-1}^i-\tau_{k^{\prime}}^j\right)} \\
& \quad+\sum_{j=1}^D \sum_{\tau_{k-1}^i \leq \tau_{k^{\prime}}^j<\tau_k^i} \frac{\alpha_{i j}}{\beta_{i j}}g_j\left(v_{k^{\prime}}^j\right)\left[1-e^{-\beta_{i j}\left(\tau_k^i-\tau_{k^{\prime}}^j\right)}\right] 
\end{align*}
Following Ozaki \cite{Ozaki1979}, we have that:
\begin{align*}
A^{i j}(k) & =\sum_{\tau_{k^{\prime}}^j<\tau_{k}^i} g_j\left(v_{k^{\prime}}^j\right)e^{-\beta_{i j}\left(\tau_{k}^i-\tau_{k^{\prime}}^j\right)} \\
& =e^{-\beta_{i j}\left(\tau_{k}^i-\tau_{k-1}^i\right)} A^{i j}(k-1)+\sum_{\tau_{k-1}^i \leq \tau_{k^{\prime}}^j<\tau_{k}^i} g_j\left(v_{k^{\prime}}^j\right)e^{-\beta_{i j}\left(\tau_{k}^i-\tau_{k^{\prime}}^j\right)}
\end{align*}
Hence, $\forall k \geq 2$ :
\begin{align*}
\Lambda_g^i\left(\tau_{k}^i\right) -  \Lambda_g^i\left(\tau_{k-1}^i\right)&= \int_{\tau_{k-1}^i}^{\tau_k^i} \mu^i(s) \mathrm{d} s +\sum_{j=1}^D  \frac{\alpha_{i j}}{\beta_{i j}}\biggl[A^{i j}(k-1) \times\left(1-e^{-\beta_{i j}\left(\tau_k^i-\tau_{k-1}^i\right)}\right)\biggl.\\&\biggr.+\sum_{\tau_{k-1}^i \leq \tau_{k^{\prime}}^j<\tau_k^i}g_j\left(v_{k^{\prime}}^j\right)\left(1-e^{-\beta_{i j}\left(\tau_k^i-\tau_{k^{\prime}}^j\right)}\right)\biggr]
\end{align*}
where $A^{i j}(0)=0,\quad \forall i,j\in\{1, \ldots, D\}$.

By virtue of the Time-Rescaling theorem (\cite{daley2007introduction}), we can draw the conclusion that the elements $\left\{V_k^1\right\}_{k\geq 0},\left\{V_k^2\right\}_{k\geq 0}, \ldots,\left\{V_k^D\right\}_{k\geq 0}$ are $D$ sequences of independent identically distributed exponential random variables with unit rate.
\end{proof}

\end{document}